\documentclass[prx,aps,superscriptaddress,nofootinbib,twocolumn,notitlepage,floatfix,10pt]{revtex4-2}
\usepackage{amsmath,mathtools,amsthm,amssymb,pifont}
\usepackage[utf8]{inputenc}
\usepackage[american]{babel}
\usepackage{graphicx,xcolor,bbold,titlesec}
\usepackage{braket}
\usepackage{MnSymbol}

\usepackage[normalem]{ulem}

\usepackage[colorlinks,
bookmarksopen,
bookmarksnumbered,
citecolor=red,
linkcolor=red,
pdfstartview=false,
urlcolor=red]{hyperref}
\usepackage{tikz,ifthen}
\usepackage{bbold}
\usepackage{tikz-network}
\usetikzlibrary{patterns,decorations.pathreplacing,calligraphy}
\usepackage{orcidlink}
\usepackage{color}
\usepackage{MnSymbol}
\definecolor{myyellow}{RGB}{240,188,66}
\definecolor{myorange}{RGB}{255,102,0}
\definecolor{myorangel}{RGB}{255,204,153}
\definecolor{myblue}{RGB}{66,135,245}

\newcommand{\Com}{\mathrm{Com}}

\newcommand{\MG}{\mathfrak{M}}
\newcommand{\Tr}{\mathrm{Tr}}

\newcommand{\Var}{\mathrm{Var}}
\newcommand{\Wg}{\mathrm{Wg}}

\renewcommand{\boxed}[1]{%
  \framebox{\raisebox{0pt}[0.4\baselineskip][0.025\baselineskip]{\hbox to 0.25cm{\hss#1\hss}}}}

\newtheorem{thm}{Theorem}

\newtheorem{lem}{Lemma}
\newtheorem{conj}{Conjecture}

\usepackage{physics}

 \newcommand{\titleinfo}{Unitary Designs from Doped Matchgate Circuits}

\begin{document}
\title{\titleinfo}

\author{Fabian Ballar Trigueros~\orcidlink{0000-0001-6395-8767}}
\email{fabian.ballar@uni-a.de}
\affiliation{Theoretical Physics III, Center for Electronic Correlations and Magnetism,
Institute of Physics, University of Augsburg, 86135 Augsburg, Germany}

\author{Zheng-Hang Sun}
\affiliation{Theoretical Physics III, Center for Electronic Correlations and Magnetism,
Institute of Physics, University of Augsburg, 86135 Augsburg, Germany}

\author{Xhek Turkeshi~\orcidlink{0000-0003-1093-3771}}
\affiliation{Institut f\"ur Theoretische Physik, Universit\"at zu K\"oln, Z\"ulpicher Strasse 77, 50937 K\"oln, Germany}

\author{Piotr Sierant~\orcidlink{0000-0001-9219-7274}}
\email{piotr.sierant@bsc.es}
\affiliation{Barcelona Supercomputing Center Plaça Eusebi G\"uell, 1-3 08034, Barcelona, Spain}

\author{Poetri Sonya Tarabunga~\orcidlink{0000-0001-8079-9040}}
\email{poetri.tarabunga@tum.de}
\affiliation{Technical University of Munich, TUM School of Natural Sciences, Physics Department, 85748 Garching, Germany}
\affiliation{Munich Center for Quantum Science and Technology (MCQST), Schellingstraße 4, 80799 M{\"u}nchen, Germany}

\begin{abstract}
Matchgate circuits realize free-fermion dynamics: they are efficiently classically simulable, yet cannot on their own generate the generic randomness required for universal computation or unitary design formation. We study a controlled route beyond this integrable limit by doping matchgate circuits with non-Gaussian gates—physically, the injection of fermionic interactions into an otherwise free system. Using the matchgate commutant framework, we obtain analytic control over unitary $2$-design formation. For globally scrambled dynamics, the design problem maps exactly onto a classical birth–death Markov chain with an Ornstein–Uhlenbeck continuum limit, recasting the emergence of quantum randomness in terms of spectral gaps and mixing times and yielding rigorous bounds on the number of non-Gaussian gates needed for approximate $2$-designs. These bounds hold for a broad class of parity-preserving non-Gaussian gates, independently of microscopic details, with numerics indicating that the same mechanism governs higher-order designs. Used as local building blocks in a glued-circuit architecture, they yield approximate parity-preserving $2$-designs in polylogarithmic depth with a sparse non-Gaussian gate count, with implications for Page-like entanglement growth and fermionic classical-shadow protocols. Finally, locality reshapes this picture: in local brickwork dynamics, design formation is diffusion-limited and far slower. Our results establish doped matchgate circuits as a controlled, analytically tractable route from free fermions to interaction-generated quantum designs.
\end{abstract}

\maketitle

\section{Introduction}

The emergence of randomness is a unifying theme across quantum many-body physics and quantum information. 
Sufficiently generic quantum dynamics is expected to look random, a feature that underlies quantum chaos and thermalization~\cite{rigol2008thermalization, dalessio2016from, srednicki94chaos} and that powers practical protocols such as randomized benchmarking~\cite{magesan2011scalable, elben2023the,heinrich2023randomized,knill2008randomized} and shadow tomography~\cite{huang2020predicting, huang2022learning, zhao2021fermionic,king2024triplyefficient, Majsak2025}.
These applications call for randomness in different amounts, and the notion of unitary designs~\cite{dankert2009exact, gross2007evenly, roberts2017chaos} makes this precise: a unitary $k$-design is an ensemble of unitaries whose first $k$ moments match those of the Haar measure, reproducing fully random behavior up to a finite order. 
How and when such design behavior emerges from structured, physically motivated dynamics is a fundamental problem in quantum science.

Fermionic Gaussian states form one of the most fundamental classes of states in physics. 
Fully characterized by their two-point covariance matrix, they provide the natural language for quadratic and mean-field descriptions of matter: from Hartree--Fock Slater determinants for atoms and molecules~\cite{surace2022fermionic, mbeng24the, echenique2007a,bach1994generalized} to Bardeen--Cooper--Schrieffer descriptions of superconductors and nuclear superfluidity~\cite{bardeen1957theory,dean2003pairing}, as well as the free or quasifree reference points underlying weak-coupling treatments in quantum field theory and lattice gauge theory~\cite{funai2025gaussian,sala2018variational,zohar2015fermionic}. 
In all these settings, effects beyond the quadratic mean-field backbone (correlations, interactions, and genuine many-body complexity) are encoded in departures from Gaussianity, equivalently in higher-order correlations beyond Wick factorization. 
Their quantum-computational incarnation is the class of matchgate circuits, the fermionic Gaussian unitaries generated by quadratic fermionic Hamiltonians~\cite{valiant2001quantum, valiant2012quantum, terhal2002classical,jozsa2008matchgates,knill2001fermionic,brod2011extending}, which are efficiently classically simulable yet display nontrivial scrambling within symmetry-constrained sectors.

This motivates the study of \emph{doped matchgate circuits}~\cite{mele25efficient,paviglianiti2026emergence}, where non-Gaussian gates are inserted in a controlled way into otherwise Gaussian dynamics. 
Such \emph{doping} introduces interactions as a tunable resource, providing a controlled departure from free-fermion integrability toward generic, universal dynamics~\cite{oszmaniec2022fermionic}. 
In this setting, \emph{non-Gaussianity} is the mechanism that drives the buildup of complexity and the emergence of randomness beyond Gaussian integrability.
Physically, each dopant can be interpreted as an interaction beyond the quadratic, mean-field level, so that the dopant number serves as a natural complexity parameter controlling the growth of non-Gaussian correlations and, consequently, of randomness. 
The Gaussian part remains efficiently simulable, while the best-known algorithms for handling non-Gaussian insertions scale exponentially in their number~\cite{dias2024classical,reardonsmith2024improved, oh2026classical}. 

Beyond this role as a tunable non-Gaussian resource, dopants are also the costly ingredient in matchgate computation: in the context of universal fermionic quantum computation, non-Gaussian gates require the costly magic-state injection~\cite{hebenstreit2019all}, analogous to magic-state injection in Clifford circuits~\cite{bravyi2005universal}. 
Minimizing their number is therefore both a theoretical and a practical goal.

A closely related problem has been studied extensively in Clifford circuits, where non-Clifford resources drive otherwise efficiently simulable dynamics toward Haar randomness and unitary designs~\cite{9hcw-7fl6,veitch2014theresourcetheory, leone2021quantum, haferkamp2022efficient, liu2022manybody, gu25magicinduced, magni25anticoncentration, magni2025quantum,magni2025anticoncentrationstatedesigndoped,leone2026nonclifford,zhang2026designs}. 
Remarkably, only a constant number of non-Clifford resources is sufficient to generate designs~
\cite{haferkamp2022efficient,leone2026nonclifford,zhang2026designs}.
In contrast, the situation for matchgate circuits is fundamentally different: Ref.~\cite{poetri26fermionic} establishes a fundamental lower bound showing that a system-size-dependent number of non-Gaussian gates (at least scaling as $\sqrt{n}$) is required for the generation of state designs (and, by extension, unitary designs). 
Nevertheless, no corresponding explicit construction of designs from doped matchgate circuits is known, leaving a substantial gap in our understanding of how non-Gaussian resources generate randomness in matchgate-based dynamics.
Closing this gap is one of the main aims of this work.

In this work, we develop an analytical framework for randomness generation in doped fermionic Gaussian circuits. 
Building on recent progress on fermionic commutants~\cite{wan2023matchgate, sierant2026theory, lastres26geometry, braccia2026the}, we formulate the replicated dynamics of doped Gaussian
circuits so that the emergence of randomness can be analyzed through a classical stochastic process. 
A schematic overview of the setup is shown in Fig.~\ref{fig:scheme}. We consider two complementary classes of doped dynamics: globally scrambled circuits, in which non-Gaussian gates are interspersed between global, Haar-random fermionic Gaussian unitaries, and locally constrained brickwork circuits, in which the non-Gaussian gate is applied repeatedly on a fixed bond.

Our central finding is that, in the globally scrambled setting, design formation admits a rigorous classical description. 
Averaging each random Gaussian layer over the matchgate group projects this replicated dynamics onto the low-dimensional matchgate commutant, so that the sole effect of a doping gate is to redistribute weight among a small number of commutant sectors. For $k=2$ this reduction takes a particularly transparent form: the replicated dynamics becomes a classical birth--death chain on the spectrum of the two-copy bridge operator, whose stationary distribution is the corresponding Haar distribution in the relevant parity sector.

Strikingly, the buildup of quantum randomness is thereby governed by classical
Markov-chain theory. In the large-system limit, the chain converges to an emergent
Ornstein--Uhlenbeck (OU) process, and questions about design formation map onto
textbook quantities of the associated chain. The rate at which the doped ensemble
approaches Haar randomness is then controlled by the spectral gap, relaxation, and
mixing times of this Markov chain. Importantly, this framework also allows direct access to the trace distance to the Haar measure, which is typically difficult to compute analytically in interacting many-body systems, but becomes tractable through the stochastic description.

This correspondence yields rigorous analytical control over approximate unitary and state
$2$-design formation. While state-level randomness emerges on scales linear in
system size, convergence of the full unitary ensemble is parametrically slower,
governed by the nontrivial mixing of the chain:
\begin{equation}
\Omega\!\left(n\max\!\bigl\{\log n,\log(1/\epsilon)\bigr\}\right)
\le
t_{2\text{-des}}
\le
O\!\left(n^2+n\log(1/\epsilon)\right).
\end{equation}
By analogy with the Ehrenfest urn~\cite{LevinPeresWilmer2006}, we further conjecture that the upper bound is
loose and that the true amount of doping required is
$\Theta\!\bigl(n\log(n/\epsilon)\bigr)$.
Remarkably, the same Ornstein--Uhlenbeck description, and with it the same
convergence bounds, emerges for an entire family of parity-preserving $q$-local
doping gates, with the microscopic gate entering only through an effective
diffusion constant.

We then extend the construction to the three-copy setting. 
Although the replicated commutant is structurally richer and precludes a comparably complete analytical treatment, its sparse combinatorial structure still allows us to construct and apply the replicated transfer matrix efficiently, giving access to higher-order frame potentials at system sizes far beyond those accessible to naive statevector simulations. 
The resulting numerics provide strong evidence that the same stochastic mechanism continues to govern higher-order design formation, with direct implications for protocols such as classical-shadow tomography.

Finally, we investigate the role of locality by contrasting globally scrambled fermionic Gaussian dynamics with locally constrained brickwork circuits. 
We find that locality qualitatively changes the mechanism of design emergence: while global scrambling rapidly mixes the relevant commutant sectors, local circuits are constrained by diffusive transport and converge much more slowly. 
The emergence of quantum complexity is thus controlled not only by the amount of non-Gaussianity injected into the circuit, but also by the spatial structure through which this resource spreads.

Building on the two-copy bounds, we use the resulting doped blocks as
local building blocks of a glued-circuit architecture, obtaining approximate
parity-preserving unitary $2$-designs in polylogarithmic depth with only
$\mathcal{O}(n\log(n/\epsilon))$ non-Gaussian gates, and we apply these guarantees to
fermionic classical-shadow protocols and generic entanglement growth.

The remainder of this work is organized as follows. 
Sec.~\ref{sec:setup} sets up doped matchgate circuits, the matchgate commutant, and the design diagnostics we use. 
Sec.~\ref{sec:Commutant_approach} constructs the commutant bases and the transfer matrix of the global protocol. 
Sec.~\ref{sec:markov} maps the two-copy dynamics to a classical birth--death chain and derives its continuum limit, spectral gap, and mixing time. 
Sec.~\ref{sec:unitary_design} converts these into rigorous bounds on unitary $2$-design convergence. 
Sec.~\ref{sec:general_q_local} extends the construction to general parity-preserving $q$-local doping gates, and Sec.~\ref{sec:glued} uses it to build efficient unitary designs. 
Sec.~\ref{sec:Applications} discusses applications to entanglement growth and fermionic classical shadows, and we conclude in Sec.~\ref{sec:conclusions}. 
Technical details, the three-copy construction, and the probabilistic protocol are given in the appendices.

\section{Doped matchgate circuits}
\label{sec:setup}

\subsection{Majorana operators and matchgate circuits}

\begin{figure*}
    \centering
    \includegraphics[width=1\linewidth]{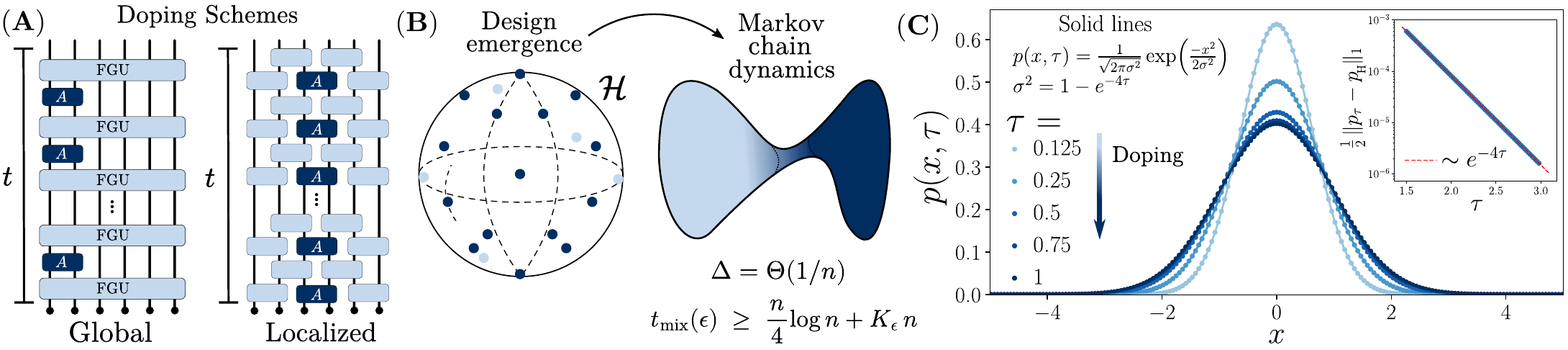}
    \caption{Doping protocols and emergence of randomness in matchgate circuits. 
    (\textbf{A}) schematic illustration of the global and localized doping constructions, in which non-Gaussian gates $A$ are inserted into otherwise fermionic Gaussian dynamics. 
    (\textbf{B}) Repeated doping drives the circuit ensemble toward unitary design behaviour in Hilbert space, which we analytically map to effective Markov-chain dynamics. The key Markov-chain quantities, such as the spectral gap $\Delta$ and the mixing time $t_\mathrm{mix}$, govern the design emergence.
    (\textbf{C}) 
    Relaxation of the distribution $p_{\tau}$ (Eq.~\eqref{eq:pt_def}) toward the stationary Haar distribution $p_{\mathrm H}$ (Eq.~\eqref{eq:Haar_dist}) in doped matchgate circuits initialized in the vacuum state $\psi_0=|0\rangle\langle 0|^{\otimes n}$, encoding the convergence to state $2$-design in the global doping scheme. The dynamics are governed by the Fokker–Planck equation, which describes the broadening of the Gaussian profile toward its stationary width, with deviations from the stationary distribution decaying as $e^{-4\tau}$. The circles show numerical data obtained at $n=10^4$, while the solid line denotes the analytical prediction from the Fokker–Planck solution.
    Inset: total-variation distance between the distribution $p_{\tau}$ and the stationary Haar distribution $p_{\mathrm H}$, illustrating exponential relaxation toward equilibrium in the effective Markov dynamics. This quantity is precisely the additive error of the associated state design. }
    \label{fig:scheme}
\end{figure*}

We consider a chain of $n$ qubits with Hilbert space
$\mathcal{H}_n=(\mathbb{C}^2)^{\otimes n}$ of dimension
$d:=2^n$.  Via the Jordan--Wigner transformation,
\begin{equation}
  \gamma_{2j-1} := \Bigl(\prod_{k<j}Z_k\Bigr)X_j,
  \quad
  \gamma_{2j}   := \Bigl(\prod_{k<j}Z_k\Bigr)Y_j,
  \quad j=1,\ldots,n,
  \label{eq:JW}
\end{equation}
the qubit algebra is equivalent to a system of $N:=2n$ Majorana
fermion operators $\{\gamma_\mu\}_{\mu=1}^{N}$, each Hermitian
($\gamma_\mu^\dagger=\gamma_\mu$) and obeying the canonical
anticommutation relations
\begin{equation}
  \{\gamma_\mu,\gamma_\nu\}
  = 2\,\delta_{\mu\nu}\,\mathbf{1},
  \qquad
  \gamma_\mu^2=\mathbf{1}.
  \label{eq:CAR}
\end{equation}
Ordered products $\gamma_S:=\gamma_{\mu_1}\gamma_{\mu_2}\cdots
\gamma_{\mu_{|S|}}$ with $\mu_1<\mu_2<\cdots<\mu_{|S|}$ and
$S\subseteq[N]:=\{1,\ldots,N\}$ form an orthogonal basis of the
full operator algebra on $\mathcal H_n$, with Hilbert--Schmidt
inner product $\Tr(\gamma_S^\dagger\gamma_{S'})=d\,\delta_{S,S'}$.

The \emph{matchgate group} $\MG_n\subset U(d)$ is the subgroup of
unitaries generated by Majorana bilinears
$e^{i\theta\gamma_\mu\gamma_\nu}$; equivalently, the fermionic
Gaussian unitaries, whose defining property is that they conjugate
single Majorana operators into linear combinations of single
Majoranas,
\begin{equation}
  U\,\gamma_\mu\,U^\dagger
  = \sum_{\nu=1}^{N}O_{\mu\nu}\,\gamma_\nu,
  \qquad U\in\MG_n,\quad O\in \mathcal{O}(N).
  \label{eq:matchgate_action}
\end{equation}
The map $U\mapsto O$ is the double cover
$\MG_n\twoheadrightarrow\mathrm{SO}(N)$ (plus parity), and the
matchgate group is the quantum-computational incarnation of
free-fermion dynamics: it is classically simulable, but
non-universal for quantum computation.

Universality is restored by injecting non-Gaussian gates into the circuit, a procedure we refer to as \emph{doping} (see Fig.~\ref{fig:scheme} (\textbf{A})). Doped matchgate circuits are the central object of our interest. 

We analyse the simplest such setting: a fixed non-Gaussian gate~$A$, repeatedly inserted between matchgate layers.  Fix a unitary $A\in U(d)$ supported on $q$ qubits (equivalently, on $2q$ Majorana modes indexed by $T\subset[N]$ with $|T|=2q$), and assume $A\notin\MG_n$.  A \emph{$t$-doped matchgate circuit} is the random unitary
\begin{equation}
  U_{\mathrm{circ}}^{(t)}
  \;:=\;
  U_{t+1}\,A\,U_t\,A\,U_{t-1}\,\cdots\,A\,U_1,
  \quad U_i\stackrel{\mathrm{i.i.d.}}{\sim}\mu_{\MG_n},
  \label{eq:circuit}
\end{equation}

The central question is: \emph{how fast does the ensemble $\{U_{\mathrm{circ}}^{(t)}\}$ converge, as $t\to\infty$, to a unitary $k$-design on $U(d)$?}

The construction in Eq.~\eqref{eq:circuit} corresponds to a \emph{global}
doping protocol, in which the same fixed non-Gaussian gate $A$ is inserted
deterministically between successive matchgate layers. This is the protocol
over which we have the most complete analytical control, for a structural
reason that we now sketch.

Design convergence is a statement about the first $k$ moments of the
ensemble, that is, about expectation values built from $k$ copies of
$U_{\mathrm{circ}}^{(t)}$ and $k$ copies of its conjugate. When such a
$k$-copy quantity is averaged over a random matchgate layer, only the
component invariant under $U^{\otimes k}$ for every matchgate $U$ survives.
This invariant component lives in the \emph{matchgate commutant}~\cite{sierant2026theory,lastres26geometry,braccia2026the}
\begin{equation}
  \Com_k(\MG_n)
  \;:=\;
  \bigl\{X\in\mathrm{End}(\mathcal H_n^{\otimes k}):
  [U^{\otimes k},\,X]=0\ \forall\,U\in\MG_n\bigr\},
  \label{eq:commutant_def}
\end{equation}
the set of $k$-copy operators that commute with every matchgate. Because
matchgates are free-fermion operations, this space is small and explicitly
characterized: at $k=2$ it has dimension $2n+1$, and for general $k$ its
dimension grows only polynomially in $n$~\cite{sierant2026theory,lastres26geometry,braccia2026the}. 
For
comparison, a fully Haar-random ensemble on $U(d)$ has a two-dimensional
two-copy commutant, spanned by the identity and the swap of the two copies~\cite{roberts2017chaos};
the additional dimensions of the matchgate commutant are precisely what
encode the free-fermion structure.

This is the origin of our analytical control over the global protocol. Each
random matchgate layer projects the replicated dynamics onto the commutant,
so the only nontrivial effect of a doping round is how the gate $A$
redistributes weight among commutant elements. The two-copy evolution
therefore collapses from the $d^4$-dimensional operator space to a finite
linear map on the $(2n+1)$-dimensional commutant. Each commutant element
carries a single integer label, and doping drives the relevant moment
operator within this space toward its Haar-invariant component; the rate of
that relaxation is the design-convergence rate. We construct the transfer
matrix in Sec.~\ref{sec:Commutant_approach}, and in
Sec.~\ref{sec:markov} we show that for $k=2$ it is exactly a classical
Markov chain on that integer label, which is what makes the global case
analytically solvable.

While this setting is analytically convenient, it represents only one particular way of injecting non-Gaussianity into the circuit. More generally, one can consider doping schemes that differ in their spatial structure and degree of randomness, as illustrated in Fig.~\ref{fig:scheme}(\textbf{A}). In a localized scheme, the circuit is arranged in a brickwork pattern of nearest-neighbour matchgates. The non-Gaussian gate $A$ is applied at a fixed spatial location in each layer (for instance, by replacing a designated bond in the brickwork by $A$), so that the spreading of non-Gaussianity is mediated dynamically by the surrounding Gaussian evolution. Local brickwork circuits of this type serve as a standard minimal model for local many-body dynamics, capturing Hamiltonian evolution subject to time-dependent noise with a localized source of impurities~\cite{anders2005realtime,thoenniss2023nonequilibrium,li2025dynamics,bravyi17impurity}. 
Understanding how the emergence of unitary design behaviour depends on these different doping protocols is a central goal of this work.

Remarkably, as we will show, these different doping protocols can lead to qualitatively distinct scaling laws for the emergence of design behaviour, highlighting that the spatial structure of non-Gaussian resources is as important as their total density.

\subsection{Quantum designs}
\label{sec:designs}

The doped ensembles defined above interpolate between structured
free-fermion dynamics and fully random unitaries. To make the notion of
``approaching randomness'' precise, we use the language of quantum designs.

Let $\mu_{\mathrm H}$ be the Haar measure on $U(d)$. An ensemble $\mathcal E$
of unitaries is an exact \emph{unitary $k$-design} if its $k$-th moments
coincide with those of $\mu_{\mathrm H}$, equivalently if the $k$-fold twirl
channel
\begin{equation}
  \Phi^{(k)}_{\mathcal E}(X)
  :=
  \mathbb E_{U\sim\mathcal E}\bigl[U^{\otimes k}\,X\,(U^\dagger)^{\otimes k}\bigr]
  \label{eq:k_twirl}
\end{equation}
equals its Haar counterpart $\Phi^{(k)}_{\mathrm H}$~\cite{roy2009unitary,mele2024introductiontohaar}.
Because the moments are matched identically, no experiment that queries the
unitary $k$ times can distinguish an exact $k$-design from a Haar-random
unitary. The state-level notion replaces the orbit of unitaries by the orbit
of a fixed initial state: writing $\psi_0:=|\psi_0\rangle\langle\psi_0|$, the
ensemble $\mathcal E$ is an exact \emph{state $k$-design} for input $\psi_0$ if
$\Phi^{(k)}_{\mathcal E}(\psi_0^{\otimes k})
=\mathbb E_{U\sim\mathcal E}[(U\psi_0 U^\dagger)^{\otimes k}]
=\Phi^{(k)}_{\mathrm H}(\psi_0^{\otimes k})$. The state notion is strictly
weaker: it constrains only how the ensemble acts on $\psi_0$, whereas the
unitary notion constrains the full channel.

Physical ensembles meet these conditions only approximately, so the relevant
objects are $\epsilon$-approximate designs, and it is here that the operational
content of the design property acquires a quantitative meaning. Two standard
notions of closeness are used, of differing strength. The ensemble
$\mathcal E$ is an \emph{additive-error} $\epsilon$-approximate unitary
$k$-design if
\begin{equation}
  \bigl\|\Phi^{(k)}_{\mathcal E}-\Phi^{(k)}_{\mathrm H}\bigr\|_\diamond\le\epsilon,
  \qquad
  \|\Phi\|_\diamond
  :=\max_{\substack{\|X\|_1=1\\ X\ge0}}\bigl\|(\Phi\otimes\mathbb I)(X)\bigr\|_1,
  \label{eq:additive_unitary}
\end{equation}
where the ancilla has the same dimension as the input of $\Phi$ and
$\|\cdot\|_\diamond$ is the diamond norm~\cite{aharonov1998quantum}. It is a
\emph{relative-error} $\epsilon$-approximate unitary $k$-design if
\begin{equation}
  (1-\epsilon)\,\Phi^{(k)}_{\mathrm H}
  \;\preceq\;\Phi^{(k)}_{\mathcal E}\;\preceq\;
  (1+\epsilon)\,\Phi^{(k)}_{\mathrm H},
  \label{eq:relative_unitary}
\end{equation}
in the sense that the differences are completely positive maps. The two
notions carry distinct operational guarantees: an additive-error design is
indistinguishable from Haar-random by any protocol making $k$
\emph{non-adaptive} (parallel) queries to $U$, whereas a relative-error design
is indistinguishable even by any \emph{adaptive} $k$-query
algorithm~\cite{schuster2024random, laracuente24approximate}; relative error is the strictly stronger
requirement. 

For state designs, the natural figure of merit is operational from the outset.
The trace distance $D(\rho,\sigma):=\tfrac12\|\rho-\sigma\|_1$ equals the
optimal probability of distinguishing $\rho$ from $\sigma$ by measurement, so
the additive state-design error
\begin{equation}
  D\bigl(\Phi^{(k)}_{t}(\psi_0^{\otimes k}),\,
  \Phi^{(k)}_{\mathrm H}(\psi_0^{\otimes k})\bigr)
  \label{eq:state_design_distance}
\end{equation}
is exactly how well any experiment using $k$ copies of the output state can
tell the $t$-doped ensemble from a Haar-random one. 
A relative-error state
notion exists as well but is not operationally necessary: additive closeness
already certifies that no quantum algorithm can distinguish the two state
ensembles~\cite{heinrich2025anticoncentrationalmostneed,grevink2025glueshortdepthdesignsunitary,zhang2026designs}, and we therefore use the additive notion for state
designs throughout. 
The key hierarchy is that unitary-design closeness is strictly stronger than state-design closeness, as the former implies the latter. 
In the presence of fermion-parity symmetry, the appropriate reference is the Haar measure within a fixed parity sector, so the Haar values quoted below are the parity-conserving ones.

Trace and diamond distances are typically difficult to evaluate directly, so it is
standard to work with a more tractable two-norm proxy, the \emph{frame
potential}. For an ensemble $\mathcal E$ the $k$-th unitary frame potential is~\cite{leone2026nonclifford,dowling2025freeindependenceunitarydesign}
\begin{equation}
  \mathcal F^{(k)}(\mathcal E)
  :=
  \mathbb E_{U,V\sim\mathcal E}\bigl|\Tr(U^\dagger V)\bigr|^{2k},
  \label{eq:frame_potential}
\end{equation}
and the $k$-th state frame potential~\cite{christopoulos2025universal,sauliere2025universality,lami2025quantumstatedesignemergent,lami2025anticoncentration}, for fixed $\psi_0$, is
\begin{equation}
  \mathcal S^{(k)}(\mathcal E;\psi_0)
  :=
  \mathbb E_{U,V\sim\mathcal E}
  \bigl|\langle\psi_0|U^\dagger V|\psi_0\rangle\bigr|^{2k}.
  \label{eq:state_frame_potential}
\end{equation}
Both use two independent draws from $\mathcal E$: $\mathcal F^{(k)}$ probes all
matrix elements of $U^\dagger V$ through the full trace, whereas
$\mathcal S^{(k)}$ probes only the single overlap
$\langle\psi_0|U^\dagger V|\psi_0\rangle$.

The frame potentials control the design distances through three exact facts,
which we state precisely as they motivate the rest of the paper. First, the
frame potential certifies \emph{exact} designs: on $U(d)$ one has
$\mathcal F^{(k)}(\mathcal E)\ge \mathcal F^{(k)}=k!$, with equality
iff $\mathcal E$ is a unitary $k$-design~\cite{mele2024introductiontohaar,roberts2017chaos};
in a fixed parity sector the reference is $\mathcal F^{(k)}_{\mathrm H}=2^{k-1}k!$.
The state frame potential obeys $\mathcal S^{(k)}\ge\mathcal S^{(k)}_{\mathrm H}$,
with $\mathcal S^{(k)}_{\mathrm H}=\tfrac12\binom{d+k-1}{k}^{-1}$ in the relevant
parity sector and equality iff $\mathcal E$ is a state $k$-design. Second, the
excess over the Haar value is the squared Hilbert--Schmidt distance between the
$k$-th moment operator and the Haar one,
\begin{equation}
  \mathcal F^{(k)}(\mathcal E)-\mathcal F^{(k)}_{\mathrm H}
  =
  \bigl\|\Phi^{(k)}_{\mathcal E}-\Phi^{(k)}_{\mathrm H}\bigr\|_2^2,
  \label{eq:frame_HS}
\end{equation}
with the analogous identity for the state object obtained by replacing
$\Phi^{(k)}_{\mathcal E}$ with $\Phi^{(k)}_{\mathcal E}(\psi_0^{\otimes k})$~\cite{mele2024introductiontohaar}.
Third---and this is the point on which the two design questions diverge---
converting this two-norm bound into an operational statement costs a
dimensional prefactor for unitary designs but none for state designs.
For unitary designs, the Hilbert--Schmidt bound \eqref{eq:frame_HS} translates
into the relative-error condition \eqref{eq:relative_unitary} only with a
factor polynomial in the dimension,
\begin{equation}
  \epsilon_{\mathrm{rel}}
  \;\le\;
  d^{2k}\,\sqrt{\mathcal F^{(k)}(\mathcal E)-\mathcal F^{(k)}_{\mathrm H}}
  \label{eq:frame_to_relative}
\end{equation}
(and the additive error \eqref{eq:additive_unitary} is no larger, since
relative error implies additive error)~\cite{brandao2016local,mele2024introductiontohaar}.
For $k=2$ the prefactor is $d^{2k}=d^4=2^{4n}$; as we show in
Sec.~\ref{sec:unitary_design}, this bound results in an additional factor of $\mathcal O(n)$ in the convergence time. For state designs, by contrast, the \emph{normalized}
excess in~\eqref{eq:frame_HS} bounds the trace
distance \eqref{eq:state_design_distance} directly with no dimensional prefactor~\cite{leone2026nonclifford},
\begin{equation}
  D\bigl(\Phi^{(k)}_{t}(\psi_0^{\otimes k}),\,
  \Phi^{(k)}_{\mathrm H}(\psi_0^{\otimes k})\bigr)
  \;\le\;
  \tfrac12\sqrt{\,\mathcal S^{(k)}(\mathcal E;\psi_0)/\mathcal S^{(k)}_{\mathrm H}-1\,}.
  \label{eq:frame_to_trace}
\end{equation}
The frame potential is therefore a computable,
large-system \emph{diagnostic} and an upper bound on the design distance, but not the operational distance itself.

For the doped ensemble $\mathcal E_t=\{U^{(t)}_{\mathrm{circ}}\}$ all of these
quantities become functions of the doping depth $t$, written
$\mathcal F^{(k)}(t)$ and $\mathcal S^{(k)}(t;\psi_0)$. The case $k=1$ is
trivial; $k=2$ controls two-point and out-of-time-order correlators; and $k=3,4$ enter randomized benchmarking and shadow tomography.

\subsection{Our approach and results}
\label{sec:approach_results}

The setups introduced above raise two related questions: how doped matchgate
circuits approach Haar-random behaviour, and how this approach depends on the
way non-Gaussianity is injected. We address them with a combination of exact
commutant methods, Markov-chain analysis, and large-scale numerics of
low-order moment diagnostics. This subsection summarises the
resulting picture; the derivations occupy
Secs.~\ref{sec:Commutant_approach}--\ref{sec:unitary_design}.

\paragraph{Doping gate and reduction to a transfer matrix.}
A central ingredient in all protocols is the choice of non-Gaussian gate. The
framework applies to a broad class of parity-preserving two-qubit gates; for
concreteness, and in all simulations, we use $A=e^{i\frac{\pi}{4}Z_1 Z_2}$, a
minimal quartic fermionic interaction. This gate is equivalent to a SWAP up to
conjugation by Gaussian unitaries, and since Gaussian transformations preserve
the matchgate commutant, the induced transfer matrix introduced in
Sec.~\ref{sec:setup} is unchanged by this equivalence.

\paragraph{Analytical results.}
For $k=2$ this structure becomes especially transparent. After each random
matchgate layer, the replicated dynamics is projected onto the two-copy
matchgate commutant. The remaining evolution is therefore captured by a
probability distribution $p_t(\nu)$ (defined below in Sec.~\ref{sec:markov_setup}) over a single integer coordinate $\nu$,
which labels the relevant commutant sectors. A doping round acts as $p_{t+1}=p_t T$,
where $T$ is an explicit birth--death transition matrix whose stationary
distribution is the Haar distribution in the corresponding parity sector. Thus,
in the globally scrambled protocol, the quantum problem of design formation is
reduced to the relaxation of a one-dimensional classical Markov chain.

The continuum limit makes this stochastic structure particularly transparent.
On the natural scale $x=\nu / \sqrt{2n},\ \tau=t/n$,
the action of the transition matrix on smooth functions takes the Ornstein--Uhlenbeck form
\begin{equation}
T f(x) = f(x)+\frac{1}{n}\mathcal L f(x)+\mathcal O(n^{-2}), \quad \mathcal L=2\partial_x^2-2x\partial_x .
\end{equation}
Equivalently, the distribution $p_\tau$ in the continuum limit obeys the Fokker--Planck equation
\begin{equation}
\partial_\tau p = \partial_x(2xp)+2\partial_x^2p .
\end{equation}
In the state-design setting, the input is given by the vacuum
state $\psi_0=|0\rangle\langle 0|^{\otimes n}$. In this case the initial
distribution is localized at $\nu=0$, and the Fokker--Planck equation has a
centered Gaussian solution with variance $\sigma^2(\tau)=1-e^{-4\tau}$. Its
stationary limit is the rescaled Haar distribution
$p_{\mathrm H}(x)=(2\pi)^{-1/2}e^{-x^2/2}$. As illustrated in
Fig.~\ref{fig:scheme}(\textbf{C}), the approach to Haar is therefore seen as the
broadening of this Gaussian profile toward its stationary width, with deviations
controlled by $e^{-4\tau}$. In Sec.~\ref{sec:gap_bounds} we make the
corresponding finite-$n$ relaxation statement rigorous by defining the
Markov-chain spectral gap $\Delta$, proving $\Delta=\Theta(1/n)$, and comparing
the result with the spectral data of Fig.~\ref{fig:gap_layout}.

The two design questions are answered separately because they ask different
things of the same chain. For state designs, as noted above, a fixed input state determines an
initial distribution over the Markov-chain sectors, and convergence to the Haar
state ensemble is controlled by the relaxation of this distribution to
stationarity. For the vacuum input discussed above, the Fokker--Planck solution gives the
explicit continuum prediction
\begin{equation}
  t \simeq \frac{n}{4}\log\frac{1}{\epsilon}.
\end{equation}
for additive state $2$-design convergence. For unitary designs, the requirement
is stronger. The finite-$n$ relaxation and mixing analysis of the same chain
gives concrete bounds
\begin{equation}
  \Omega\!\left(n\max\{\log n,\log(1/\epsilon)\}\right)
  \le t_{2\text{-des}}
  \le O\!\left(n^2+n\log(1/\epsilon)\right).
\end{equation}

The gap between these two timescales has a clear origin. The state design requires
only relaxation of the distribution associated with a fixed input state, whereas
a unitary design requires the stronger relative-error condition introduced in
Sec.~\ref{sec:designs}, which is sensitive to the full channel rather than to
its action on a single initial state. Thus, the two-copy Markov chain supplies
both a physical picture of relaxation and rigorous bounds on unitary
$2$-design formation.

\paragraph{Numerical evidence.}
As emphasised in Sec.~\ref{sec:designs}, the frame potentials are not
the operational design distances, but upper bounds on them; we use them here as
computable diagnostics that test the analytical results at large system sizes. 
For the global protocol of Eq.~\eqref{eq:circuit},
Fig.~\ref{fig:global_2f}(\textbf{A}) shows the unitary $2$-frame potential
against the rescaled depth $\tau = t/n$: the excess over the $\mathbb{Z}_2$-symmetric
Haar value collapses across system sizes onto the thermodynamic-limit prediction, $\mathcal{F}(t)-\mathcal{F}^{(2)}_{\mathrm H}\to 4/(e^{4\tau}-1)$,
obtained analytically from the emergent OU process
(App.~\ref{app:OU_perturbation}); this confirms relaxation on the
extensive scale predicted by the gap analysis. Panel~(\textbf{B}) shows the state $2$-frame
potential, which bounds the trace distance to the Haar state
ensemble via Eq.~\eqref{eq:frame_to_trace}; the finite-size data again follow the continuum prediction.

\begin{figure}[h!]
    \centering
    \includegraphics[width=1\linewidth]{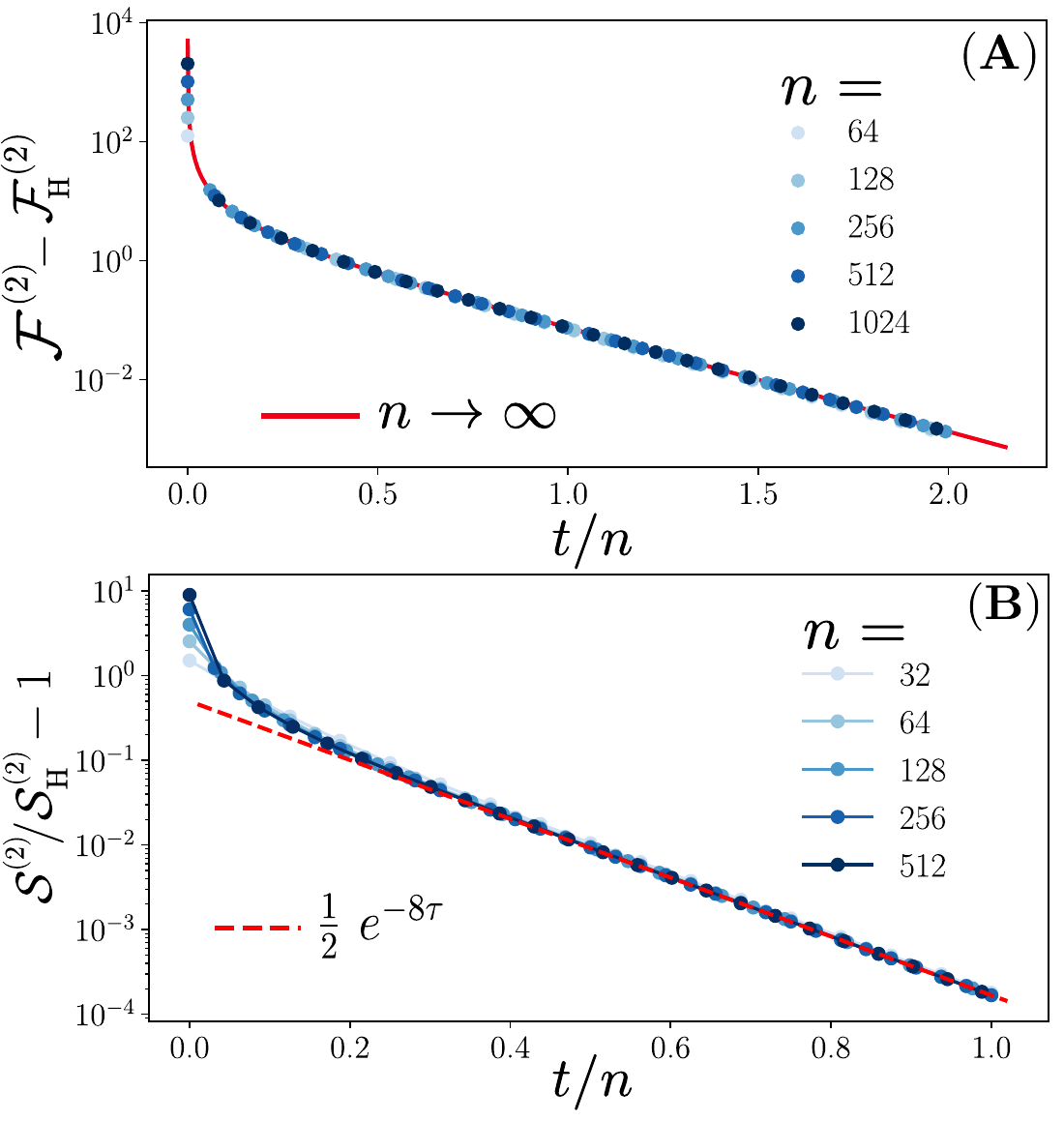}
    \caption{(\textbf{A}) Deviation of the unitary $2$-frame potential from the $\mathbb{Z}_2$-symmetric Haar value, as a function of doping depth. The solid red line shows the thermodynamic-limit prediction from the perturbative approach, while markers denote values obtained from the exact doping-matrix method. (\textbf{B}) Relative deviation of the state $2$-frame potential, as a function of doping depth.}
    \label{fig:global_2f}
\end{figure}

This behaviour is not restricted to the lowest moment. At $k=3$, the matchgate commutant remains explicitly characterizable, allowing us to construct the corresponding three-copy transfer matrix analytically. As detailed in App.~\ref{app:3copy}, locality of the dopant reduces the three-copy transfer matrix to an explicit finite combinatorial formula for its matrix elements. Although the dimension of this matrix grows as $\sim n^3$, the resulting operator is sparse. We exploit this sparsity within a stochastic-trace-estimation scheme, applying $C^{(3)}$ through sparse matrix--vector products generated from the combinatorial transition rules. In this way, we can probe effective dimensions of order $2\times10^7$. Figure~\ref{fig:global_fp3} shows the unitary and state $3$-frame potentials.

\begin{figure}[h!]
    \centering
    \includegraphics[width=1\linewidth]{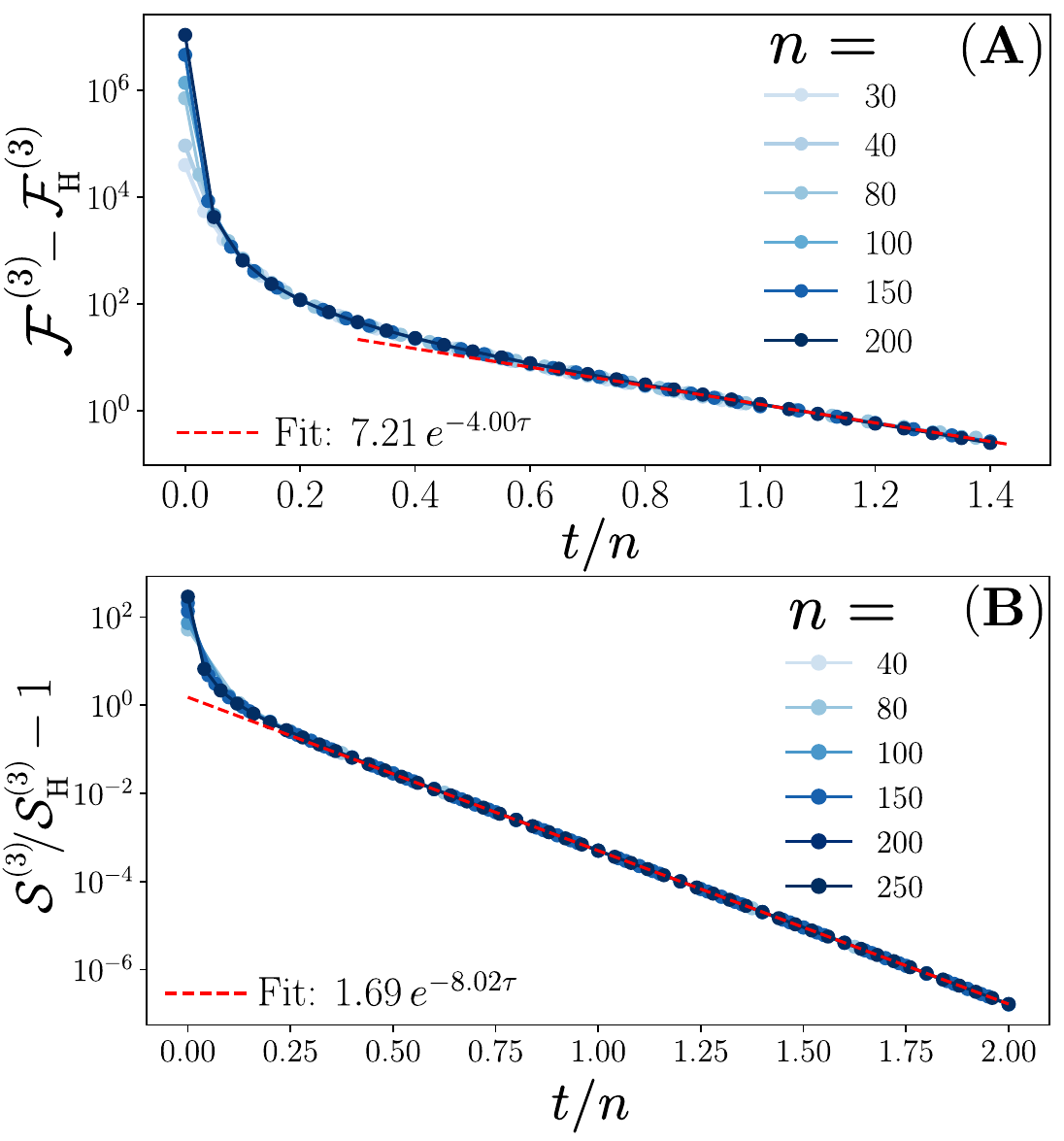}
    \caption{(\textbf{A}) Deviation of the unitary $3$-frame potential from the $\mathbb{Z}_2$-symmetric Haar value, as a function of doping depth. (\textbf{B}) Relative deviation of the state $3$-frame potential, as a function of doping depth.}
    \label{fig:global_fp3}
\end{figure}

The resulting three-copy data collapse as a function of $t/n$ and exhibit the same leading exponential relaxation observed at $k=2$. This provides evidence that the same mechanism persists beyond the second moment,. In particular, it provides numerical evidence that the same upper bounds on the required doping to generate unitary and state designs continue to hold for $k\geq3$, namely $\mathcal{O}(n^2)$ and $\mathcal{O}(n)$, respectively. Note that convergence to a $2$-design also provides a necessary condition for convergence to a $k$-design, since any $k$-design with $k \geq 2$ is, in particular, a $2$-design, both in the unitary and state settings. This, in particular, suggests that state $3$-design formation occurs on a timescale $\Theta(n)$, the same timescale as state $2$-design formation. Small-system exact-diagonalisation data for the state $4$-frame potential, shown in App.~\ref{app:4-frame}, are consistent with the same scaling, again suggesting that state $4$-design formation occurs on the same timescale $\Theta(n)$.

\paragraph{Role of locality.}
The picture changes qualitatively away from global scrambling. For the
localized protocol of Fig.~\ref{fig:scheme}(\textbf{A}), in which $A$ acts on a fixed bond
of a brickwork circuit, the frame potentials collapse in terms of $t/n^2$
rather than $t/n$ (Fig.~\ref{fig:locality_Frames}, with the inset showing the
same trend at $k=3$). These large-system data are obtained by
exploiting the fact that Clifford matchgates form an exact matchgate
$3$-design~\cite{wan2023matchgate,sierant2026theory}, so that averages of the $2$- and
$3$-frame potentials over Haar-random matchgates can be evaluated by efficient
stabilizer simulation rather than by sampling generic Gaussian unitaries; we
detail this in App.~\ref{app:Numerics}. The slower scaling is
transport-limited: non-Gaussianity injected at one point must spread through
the surrounding Gaussian dynamics, whereas in the global protocol, every doping
round is followed by a fully random layer that redistributes it immediately. This picture is consistent with the known diffusive transport in random matchgate circuits in the absence of impurities~\cite{dias2021diffusiveoperatorspreadingrandom,nahum2017quantum,swann2025spacetime}, suggesting that the local non-Gaussian impurity inherits the same diffusive bottleneck.
The emergence of Haar-like behaviour is therefore controlled not only by the
amount of non-Gaussianity but by its spatial distribution. 
An intermediate,
probabilistic protocol is analysed in App.~\ref{app:Probabilistic_Doping}.

\begin{figure}[h!]
    \centering
    \includegraphics[width=1\linewidth]{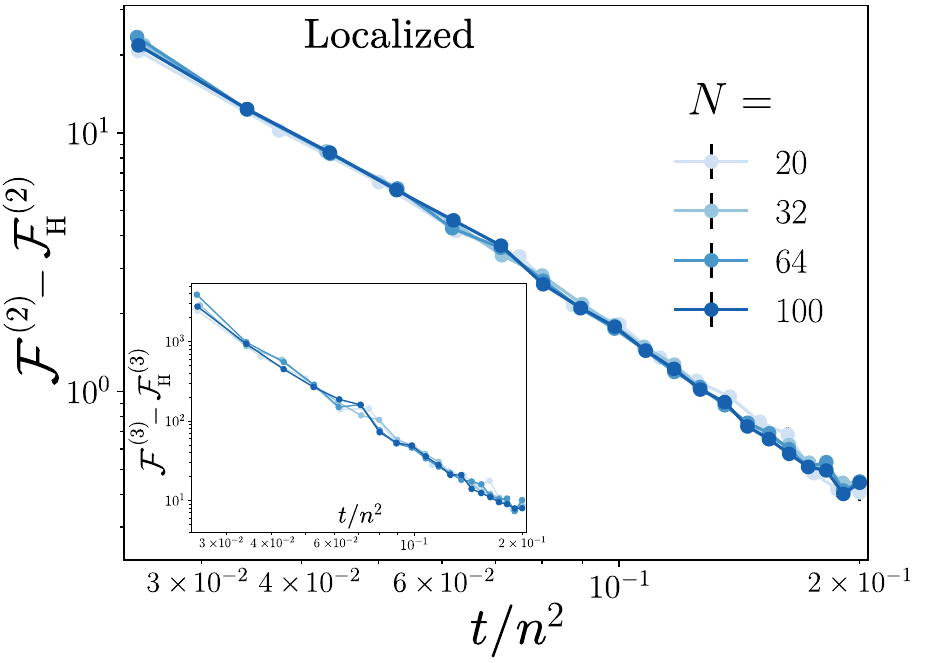}
    \caption{Deviation of the unitary $2$-frame potential from the $\mathbb{Z}_2$-symmetric Haar value, as a function of doping depth for the localized doping scheme shown in Fig.~\ref{fig:scheme}(\textbf{A}). The inset shows the corresponding value for the unitary $3$-frame potential.}
    \label{fig:locality_Frames}
\end{figure}

The remainder of the paper establishes these results, beginning in
Sec.~\ref{sec:Commutant_approach} with the commutant bases and the transfer
matrix.

\section{The commutant approach}
\label{sec:Commutant_approach}

Building on the commutant picture briefly introduced in Sec.~\ref{sec:setup}, we now construct it explicitly and derive the transfer matrix of the global protocol. From here onwards we specialise to $k=2$, the lowest nontrivial design order and the one most directly accessible to the machinery developed below. We accordingly drop the superscript and write $\mathcal F(t):=\mathcal F^{(2)}(t)$ and $\mathcal S(t;\psi_0):=\mathcal S^{(2)}(t;\psi_0)$.

\subsection{Bases of the commutant}
By~\eqref{eq:circuit}, averaging over the $t{+}1$ i.i.d.\
matchgate layers factorises the two-copy dynamics into a
sequence of matchgate twirls alternating with unitary
conjugations.  The two building blocks are the
\emph{matchgate two-copy twirl}
\begin{equation}
  \Phi_{\MG}(X)
  \;:=\;
  \mathbb E_{U\sim\mu_{\MG_n}}\bigl[
  (U\otimes U)\,X\,(U^\dagger\otimes U^\dagger)\bigr],
  \label{eq:MG_twirl}
\end{equation}
which is the $k=2$ matchgate twirl $\Phi^{(2)}_{\MG_n}$ in the
notation of Eq.~\eqref{eq:k_twirl},
and the \emph{doping conjugation}
\begin{equation}
  \Phi_A(X)
  \;:=\;
  (A\otimes A)\,X\,(A^\dagger\otimes A^\dagger),
  \label{eq:doping_conj}
\end{equation}
and the averaged two-copy channel after $t$ dopings is
\begin{equation}
  \Phi_t
  \;:=\;
  \Phi_{\MG}\circ
  \bigl[\Phi_A\circ\Phi_{\MG}\bigr]^{t} = \bigl[\Phi_{\MG}\circ\Phi_A\circ\Phi_{\MG}\bigr]^{t}.
  \label{eq:two_copy_channel}
\end{equation}
The last equality holds because $\Phi_{\MG} \circ \Phi_{\MG} = \Phi_{\MG}$, which follows from the left/right invariance of the matchgate-Haar measure. For $t=0$ the product is empty and $\Phi_0=\Phi_{\MG}$.

Both frame potentials admit a unified representation in terms
of~$\Phi_t$.  The key identity is $|x|^{2k}=|x^k|^2$ applied at
the operator level: for any $d\times d$ matrix~$M$,
$|\Tr(M)|^{2k}=|\Tr_{\mathcal H^{\otimes k}}(M^{\otimes k})|^2$,
and for any unit vector $|\psi\rangle$,
$|\langle\psi|M|\psi\rangle|^{2k}
=|\Tr(\psi^{\otimes k}\,M^{\otimes k})|^2$
where $\psi:=|\psi\rangle\langle\psi|$.
Inserting $M=U^\dagger V$ with $k=2$, averaging over
independent draws $U,V\sim\mathcal E_t$, and using orthonormality
of the operator basis $\{|i\rangle\langle j|\}$ to convert
$\mathbb E|\Tr(\cdots)|^2$ into a sum of squared traces,
one obtains
\begin{align}
  \mathcal F(t)
  &\;=\;
  \sum_\alpha
\Tr\!\bigl[\Phi_t(E_\alpha)^\dagger\,\Phi_t(E_\alpha)\bigr],
  \label{eq:FP_via_channel}
 \\
  \mathcal S(t;\psi_0)
  &\;=\;
  \Tr\!\Bigl[\Phi_t\bigl(\psi_0^{\otimes 2}\bigr)^2\Bigr],
  \label{eq:SFP_via_channel}
\end{align}
where $\{E_\alpha\}$ is any orthonormal basis of
$\mathrm{End}(\mathcal H_n^{\otimes 2})$ and
$\psi_0^{\otimes 2}:=
(|\psi_0\rangle\langle\psi_0|)^{\otimes 2}$.

The decisive simplification, already flagged in Sec.~\ref{sec:setup}, is
that $\Phi_{\MG}$ is the orthogonal projection onto the two-copy matchgate
commutant $\Com_2(\MG_n)$, the $k=2$ case of Eq.~\eqref{eq:commutant_def}.
Concretely, $\Phi_{\MG}$ is self-adjoint and idempotent with image exactly
$\Com_2(\MG_n)$, so it sends any two-copy operator to its matchgate-invariant
component. As a consequence, the \emph{transfer matrix}
$\tilde{C}\coloneqq\Phi_{\MG}\circ\Phi_A\circ \Phi_{\MG}$ is a linear
endomorphism of the commutant, and so is the doping channel
$\Phi_{t}=\tilde{C}^{\,t}$. When $\tilde{C}^{\,t}$ is self-adjoint, its
spectrum controls the design-convergence rate.

This commutant has the dimension $N+1=2n+1$ quoted in
Sec.~\ref {sec:setup}, and admits two natural bases. The basis
we use throughout, and which underlies essentially all of our analytical
results, is the basis of \emph{orthogonal projectors} $\{P_\nu\}$ associated
with the spectrum of the bridge operator~\cite{sierant2026theory}. The index $\nu$
is the single integer label. Its key
algebraic property, $P_\nu P_{\nu'}=\delta_{\nu\nu'}P_\nu$, is what lets the
two-copy doped channel be reformulated as a classical Markov chain on $\nu$
in Sec.~\ref{sec:markov}.

A second, complementary basis of \emph{pairing operators} $\{|\Upsilon^{(2)}_k\rrangle\}$~\cite{sierant2026theory,wan2023matchgate}, introduced below, is less convenient for the Markov-chain formulation but is useful for more explicit operator calculations. 
In practice, the projector basis is better adapted to the structure of the doping matrix and the resulting stochastic description, whereas the pairing basis provides a more convenient language for explicit replica-space constructions, particularly in higher-copy settings such as the three-copy commutant considered in App.~\ref{app:3copy}.

\paragraph{Spectral-projector basis.}
The primary basis is furnished by the spectral projectors of
the inter-replica \emph{bridge operator}
\begin{equation}
  \Lambda_{12}
  \;:=\;
  \sum_{\mu=1}^{N} \,\gamma_\mu^{(1)}\,\gamma_\mu^{(2)},
  \label{eq:bridge}
\end{equation}
where superscripts label the two replica copies of the system.
The bridge operator is a sum of commuting $\pm 1$-valued
inter-replica bilinears; its eigenvalues are
$\nu\in\{-2n,-2n+2,\ldots,2n-2,2n\}$, corresponding to $r=
0,1,\ldots,2n$ via $\nu=2r-2n$, and the $r$-th eigenspace has
dimension $\binom{2n}{r}$. 

An equivalent and useful representation of $P_\nu$ in terms of the
commuting family $h_\mu:=i\gamma_\mu^{(1)}\gamma_\mu^{(2)}$ is
the explicit product form
\begin{equation}
  P_\nu
  \;=\;
  \sum_{\substack{S\subseteq[N]\\|S|=r}}
  \prod_{\mu\in S}\frac{\mathbf 1+h_\mu}{2}
  \prod_{\mu\notin S}\frac{\mathbf 1-h_\mu}{2},
  \qquad
  \nu=2r-2n,
  \label{eq:Pnu_product}
\end{equation}
which picks out the $\binom{2n}{r}$-dimensional subspace on
which exactly $r$ of the $h_\mu$ take the value $+1$.  

\paragraph{Pairing operator basis.}
A second basis, complementary to the projector basis, is
provided by the pairing operators. Working in the
operator-vectorisation (double-ket) formalism, in which an
operator $X$ on $\mathcal H_n$ is identified with a vector
$|X\rrangle:=\sum_{ij}X_{ij}\,|i\rangle|j\rangle$ and the
Hilbert--Schmidt inner product becomes
$\llangle X|Y\rrangle=\Tr(X^\dagger Y)$, define for each
$k\in\{0,1,\ldots,2n\}$
\begin{equation}
  \big|\Upsilon^{(2)}_k\big\rrangle
  \;:=\;
  \binom{2n}{k}^{-1/2}
  \sum_{\substack{S\subseteq[N]\\|S|=k}}
  \big|\gamma_S\big\rrangle\otimes\big|\gamma_S\big\rrangle,
  \label{eq:ups-def}
\end{equation}
where $\gamma_S=\gamma_{\mu_1}\cdots\gamma_{\mu_k}$ is the ordered
product of Majoranas indexed by $S=\{\mu_1<\cdots<\mu_k\}$.  The
combinatorial normalisation $\binom{2n}{k}^{-1/2}$ makes
$\{|\Upsilon^{(2)}_k\rrangle\}_{k=0}^{2n}$ orthonormal in the
vectorised Hilbert--Schmidt norm.  Matchgate invariance is
manifest: a matchgate $U$ conjugates $\gamma_\mu$ by an orthogonal
matrix $O\in \mathcal{O}(N)$, which in turn rotates the weight-$k$
antisymmetric-tensor sector by an orthogonal matrix; summing
$\gamma_S\otimes\gamma_S$ over all $|S|=k$ contracts the two
tensor factors of this rotation into $O^{\mathrm T}O=\mathbf 1$,
leaving the sum invariant.  Hence each $|\Upsilon^{(2)}_k\rrangle$
lies in $\Com_2(\MG_n)$, and since there are $2n+1$ of them they span the commutant.

\subsection{Transfer matrix and spectral form of the diagnostics}

In the spectral-projector basis, the single-layer transfer
matrix of the doped circuit is the $(N{+}1)\times(N{+}1)$
matrix
\begin{equation}
\begin{aligned}
\tilde{C}_{r'r}
&:= \frac{\Tr\bigl(P_{r'}\,(A\otimes A)\,P_r\,
(A^\dagger\otimes A^\dagger)\bigr)}
{\sqrt{\Tr(P_{r'})\,\Tr(P_r)}} \\
&= \frac{\Tr\bigl(P_{r'}\,Q_{r}\bigr)}
{\sqrt{\binom{N}{r'}\binom{N}{r}}},
\quad r,r'=0,1,\ldots,N,
\end{aligned}
\label{eq:Ctilde_intro}
\end{equation}
where we have introduced the conjugated projectors
$Q_r:=(A\otimes A)\,P_r\,(A^\dagger\otimes A^\dagger)$,
normalised so that $\tilde C=\mathbf 1$ when $A$ is itself Gaussian
(in which case conjugation by $A\otimes A$ preserves every $P_r$
individually). For the parity-preserving gates we
consider, $\tilde C$ is real symmetric; the precise condition on
$A$ is given in Sec.~\ref{sec:general_q_local}, and we assume it
throughout the spectral analysis.

This single matrix is all that the two-copy design
problem depends on. Returning to the frame
potentials~\eqref{eq:FP_via_channel}
and~\eqref{eq:SFP_via_channel}, the projection onto the commutant
ensures that only the $(N{+}1)$ commutant components survive
after each matchgate twirl. Writing $\{\lambda_m\}$ for the
eigenvalues of~$\tilde C$ and $\{|v_m\rangle\}$ for the
corresponding eigenvectors, the frame potentials after $t$
doping layers take the spectral form
\begin{equation}
  \mathcal F(t)=\Tr(\tilde C^{2t})=\sum_m\lambda_m^{2t},
  \qquad
  \mathcal S(t;\psi_0)
  =\sum_m\lambda_m^{2t}\,|c_m|^2,
  \label{eq:FP_spectral}
\end{equation}
with $c_m$ the overlap of $\psi_0^{\otimes 2}$ with the $m$-th
eigenvector of~$\tilde C$. Both frame potentials involve the
\emph{same} power $\lambda_m^{2t}$ and are controlled by the same
channel~$\Phi_t$; they differ only in the weights (all equal
to~$1$ for $\mathcal F$, state-dependent for $\mathcal S$).

Equation~\eqref{eq:FP_spectral} reduces the whole
problem to the spectrum of $\tilde C$. Self-adjointness makes the
eigenvalues real, and because the matchgate twirl fixes the
Gaussian-invariant component there is a top eigenvalue
$\lambda_0=1$ with all others obeying $|\lambda_m|<1$. The frame
potentials therefore relax as a sum of decaying exponentials,
whose slowest mode is set by the distance of the largest
non-unit eigenvalue from $1$, the \emph{spectral gap} $\Delta$
(defined precisely in Eq.~\eqref{eq:gap_def_markov}). This gap controls
the design-convergence rate anticipated in
Sec.~\ref{sec:approach_results}, and determining its scaling with
$n$ is one of the goals of Sec.~\ref{sec:markov}.

\section{Markov chain description}
\label{sec:markov}

The central question of the paper is: how many non-Gaussian doping gates $t$ are required for the doped ensemble to form an $\epsilon$-approximate unitary or state design? As anticipated in Sec.~\ref{sec:Commutant_approach}, the orthogonal-projector property of the basis $\{P_\nu\}$ for $k=2$ allows a mapping to a one-dimensional \emph{classical Markov chain} on the spectrum of the bridge operator $\Lambda_{12}$, which we carry out in this section. This Markov-chain description allows us to attack both design questions directly with classical Markov-chain tools, without going through frame potentials. Two quantities play a distinguished role: (1) the \emph{spectral gap} $\Delta$ of the chain, which translates into upper and lower bounds on the convergence to an $\epsilon$-approximate unitary $2$-design, and (2) the \emph{mixing time}, which fixes the leading-order time $t$ at which the doped ensemble forms an $\epsilon$-approximate state $2$-design and also yields a further lower bound on the convergence time to an $\epsilon$-approximate unitary $2$-design.

Throughout this section we focus on the canonical SWAP doping $A=\mathrm{SWAP}_{12}$; the extension to arbitrary parity-preserving $q$-local doping gates is carried out in Sec.~\ref{sec:general_q_local}, where we show that the same framework applies with straightforward modifications.

The remainder of this section is organized as follows. In Sec.~\ref{sec:markov_setup} we map the doped two-copy dynamics to an explicit birth--death chain on the state space $8\mathbb{Z}\cap[-2n,2n]$, with stationary distribution given by the Haar distribution. In Sec.~\ref{sec:continuum} we derive the Fokker--Planck equation governing its continuum limit and read off the predicted spectral gap $\Delta_{\mathrm{cont}}=2/n+\mathcal{O}(1/n^2)$. In Sec.~\ref{sec:gap_bounds} we promote this prediction to a rigorous finite-$n$ statement by establishing matching upper and lower bounds, both of order $1/n$, that together yield $\Delta=\Theta(1/n)$. In Sec.~\ref{sec:mixing_time} we prove a $\Omega(n\log n)$ lower bound on the mixing time of the chain. Finally, in Sec.~\ref{sec:unitary_design} we combine these ingredients to obtain upper and lower bounds on the convergence time to an $\epsilon$-approximate unitary $2$-design.

\subsection{Mapping to a birth--death chain}
\label{sec:markov_setup}
The matchgate twirl $\Phi_{\MG}$ is the orthogonal projection onto $\Com_2(\MG_n)$, and acts diagonally in the spectral-projector basis $\{P_\nu\}$: for any two-copy operator $O$,
\begin{equation}
  \Phi_{\MG}(O)
  \;=\;
  \sum_{\nu}\frac{\Tr(P_\nu\,O)}{\Tr P_\nu}\,P_\nu,
  \qquad
  \Tr P_\nu\;=\;\binom{2n}{n+\nu/2}.
  \label{eq:MG_diagonal}
\end{equation}
Since $\{P_\nu\}$ is an orthogonal basis of the commutant, every element of $\Com_2(\MG_n)$ is a (real) linear combination of the $P_\nu$.  After the innermost matchgate twirl in $\Phi_t$, the iterate is therefore diagonal in this basis at every step,
\begin{equation}
  \Phi_t\bigl(O\bigr)
  \;=\;
  \sum_{\nu}p_t(\nu)\,\frac{P_\nu}{\Tr P_\nu},
  \qquad
  p_t(\nu)\;:=\;\Tr\,\!\bigl[\Phi_t(O)\,P_\nu\bigr],
  \label{eq:pt_def}
\end{equation}
with $p_0(\nu)=\Tr(P_\nu\,O)$. In particular, $p_0(\nu)=\delta_{\nu,0}$ for $O=|0\rangle \langle 0|^{\otimes n}$. When $O$ is a positive semi-definite operator, the coefficients $\{p_t(\nu)\}$ form a probability distribution on the spectrum of $\Lambda_{12}$, normalised by $\sum_\nu p_t(\nu)=\Tr[\Phi_t(O)]=\Tr[O]$. This distribution is also known as the \emph{bridge spectral distribution}~\cite{poetri26fermionic}. This key feature stems from the orthogonal-projector property, which is precisely what allows a mapping to a \emph{birth--death chain} on the bridge eigenvalue $\nu$, as shown below. 
Inserting~\eqref{eq:pt_def} into the recursion $\Phi_{t+1}=\Phi_{\MG}\circ\Phi_A\circ\Phi_t$ and using~\eqref{eq:MG_diagonal} reduces the two-copy dynamics to a linear update $p_{t+1}=p_t\,T$, with transition matrix
\begin{equation}
  T(\nu,\nu')
  \;:=\;
  \frac{\Tr\bigl[(A\otimes A)\,P_\nu\,(A^\dagger\otimes A^\dagger)\,P_{\nu'}\bigr]}{\binom{N}{n+\nu/2}},
  \, A=\mathrm{SWAP}_{12}.
  \label{eq:T_def}
\end{equation}
This is a row-stochastic matrix which is related by a similarity transformation with the transfer matrix as $T:=D\tilde C D^{-1}$, where $D_{\nu\nu}=\sqrt{\Tr P_\nu}$.
As such, the spectral form~\eqref{eq:FP_spectral} of the frame potentials
and the Markov-chain relaxation analysed here are governed by the
eigenvalues of $T$.

A direct evaluation of~\eqref{eq:T_def} on the local eigenstate decomposition of $P_\nu$ shows that, the only non-trivial action of the SWAP gate is on two local sectors $h_1=h_2=h_3=h_4=1 \longleftrightarrow h_1=h_2=h_3=h_4=-1$, which corresponds to the shift of the bridge eigenvalue $\nu$ by $\pm 8$. The bridge eigenvalues $\nu$ take even integer values in $[-2n,2n]$, so the residue $\nu\bmod 8\in\{0,2,4,6\}$ is conserved by $T$, and the full transfer matrix is block-diagonal across the four invariant sublattices
\begin{equation}
  \mathcal{S}_a\;:=\;\bigl\{\nu\in[-2n,2n]:\nu\equiv a\pmod 8\bigr\},
  \quad a\in\{0,2,4,6\},
  \label{eq:shells}
\end{equation}
\begin{equation}
  T\;=\;\bigoplus_{a\in\{0,2,4,6\}}T^{(a)}.
\end{equation}
This block-diagonal structure can alternatively be understood as a consequence of exchange and parity symmetries, as discussed in Ref.~\cite{poetri26fermionic}.

Each block $T^{(a)}$ is tridiagonal, and within each sector the transition rates are obtained by the same combinatorial counting:
\begin{equation}
  T^{(a)}(\nu,\nu')
  \;=\;
  \begin{cases}
    a_\nu, & \nu'=\nu-8,\\
    b_\nu, & \nu'=\nu+8,\\
    1-a_\nu-b_\nu, & \nu'=\nu,\\
    0, & \text{otherwise},
  \end{cases}
  \qquad \nu,\nu'\in\mathcal{S}_a,
  \label{eq:T_birth_death}
\end{equation}
where
\begin{equation}
  a_\nu \;=\; \frac{\binom{2n-4}{n+\nu/2-4}}{\binom{2n}{n+\nu/2}},
  \quad
  b_\nu \;=\; \frac{\binom{2n-4}{n+\nu/2}}{\binom{2n}{n+\nu/2}}.
\end{equation}

The role of the four sectors differs slightly between the state-design and unitary-design questions, but in both cases it is sufficient to focus on the single block $T^{(0)}$:
\begin{itemize}
    \item For the \emph{state $2$-design}, the input $\psi_0=|0\rangle^{\otimes n}$ has $p_0(\nu)=\delta_{\nu,0}\in\mathcal{S}_0$, so the dynamics remains supported on $\mathcal{S}_0=8\mathbb{Z}\cap[-2n,2n]$ at all times and only the block $T^{(0)}$ enters the convergence analysis.
    \item For the \emph{unitary $2$-design}, the relevant figures of merit (the spectral gap and the mixing time) involve all four sectors. However, the four blocks $T^{(a)}$ share the same birth--death structure~\eqref{eq:T_birth_death} with nearly identical rates $a_\nu,b_\nu$ on $\mathcal{S}_a$; their spectral gaps and mixing times therefore agree at leading order. It is consequently sufficient to analyze the canonical block $T^{(0)}$ and import the same conclusions for the remaining sectors.
\end{itemize}
We henceforth restrict the analysis to the block
\begin{equation}
  T\;\equiv\;T^{(0)}, 
\end{equation}
with the understanding that all spectral and mixing-time statements below transfer verbatim to $T^{(a)}$ for $a\in\{2,4,6\}$ at the relevant leading order.

The block $T$ is irreducible and aperiodic on $\mathcal{S}_0$, hence has a unique stationary distribution. This unique distribution is given by the Haar distribution, with explicit weights~\cite{poetri26fermionic}:
\begin{equation}
  p_{\mathrm{H}}(\nu)
  \;=\;
  \frac{4\binom{2n}{n+\nu/2}}{2^n(2^n+2)},
  \qquad \nu\in 8\mathbb{Z}\cap[-2n,2n],
  \label{eq:Haar_dist}
\end{equation}
such that $p_{\mathrm{H}}T=p_{\mathrm{H}}$. Moreover, the chain is reversible, as it satisfies the detailed balance condition

\begin{equation}
  p_{\mathrm{H}}(\nu)\,T(\nu,\nu')
  \;=\;
  p_{\mathrm{H}}(\nu')\,T(\nu',\nu).
  \label{eq:dbc}
\end{equation}

Reversibility implies that the spectrum of $T$ is real, with $1=\lambda_0>\lambda_1\geq\cdots\geq\lambda_{\min}>-1$, and the spectral gap
\begin{equation}
  \Delta\;:=\;1-\max\bigl(\lambda_1,|\lambda_{\min}|\bigr)
  \label{eq:gap_def_markov}
\end{equation}
controls the rate of convergence of $p_t$ to $p_{\mathrm{H}}$.

The Markov chain $T$ is the central object from which all rigorous design bounds in this paper follow. The trace-distance to the Haar  state ensemble at $\psi_0=|0\rangle^{\otimes n}$ is given exactly by the total variation distance of the corresponding distributions~\cite{poetri26fermionic},
\begin{equation}
  D\!\bigl(\Phi_t(\psi_0^{\otimes 2}),\,\Phi_{\mathrm H}^{(2)}(\psi_0^{\otimes 2})\bigr)
  \;=\;
  \tfrac{1}{2}\bigl\|p_t-p_{\mathrm H}\bigr\|_{1},
  \label{eq:TV_chain}
\end{equation}
with initial distribution $p_0(\nu) = \delta_{\nu,0}$.

Therefore, state $2$-design convergence in additive error reduces to the mixing
time of $T$, starting from a fixed initial distribution. The same convergence is also diagnosed by the state frame potential. Expanding
$\Phi_t(\psi_0^{\otimes2})$ in the projector basis~\eqref{eq:pt_def}, the state
frame potential~\eqref{eq:SFP_via_channel} becomes a weighted sum over
$p_t(\nu)^2$; the key point is that the Haar weights~\eqref{eq:Haar_dist} obey
$\Tr P_\nu = p_{\mathrm H}(\nu)/\mathcal S^{(2)}_{\mathrm H}$, which turns this
sum into the $\chi^2$-divergence of $p_t$ from $p_{\mathrm H}$,
\begin{equation} \label{eq:sfp}
 \mathcal{S}(t;\psi_0)/\mathcal{S}^{(2)}_{\mathrm H}-1
 = \sum_\lambda p_t(\lambda)^2/p_{\mathrm{H}}(\lambda) - 1
 = \|p_t/p_{\mathrm H}-1\|_{2,p_{\mathrm H}}^2 ,
\end{equation}
with $\|f\|_{p,\pi}:=\bigl(\sum_\nu|f(\nu)|^p\pi(\nu)\bigr)^{1/p}$. This
$\chi^2$-divergence upper-bounds $\|p_t-p_{\mathrm H}\|_1$ for $p>1$ by Jensen's
inequality. As a consequence, the state frame potential bounds the trace
distance from the Haar distribution:
\begin{equation}
    D\!\bigl(\Phi_t(\psi_0^{\otimes 2}),\,\Phi_{\mathrm H}^{(2)}(\psi_0^{\otimes 2})\bigr)
   \leq \frac{1}{2} \sqrt{\mathcal{S}(t;\psi_0)/\mathcal{S}^{(2)}_{\mathrm H}-1}.
\end{equation}
Note that this inequality also holds for higher $k\geq3$~\cite{leone2026nonclifford}.

Another commonly considered notion of approximate design is the \emph{relative} (or multiplicative) error. The $t$-doped ensemble is a relative-error $\epsilon$-approximate state $2$-design if
\begin{equation} \label{eq:rel_error_state_design}
    (1\!-\epsilon)\Phi_{\mathrm H}^{(2)}(\psi_0^{\otimes 2}) \!\leq \! \Phi_t(\psi_0^{\otimes 2} )\leq(1\!+\epsilon)\Phi_{\mathrm H}^{(2)}(\psi_0^{\otimes 2}),
\end{equation}
in the operator ordering. At the level of the chain, the relative error $\epsilon$ corresponds to the weighted $\infty$-norm
\begin{equation}
  \bigl\|p_t/p_{\mathrm H}-1\bigr\|_{\infty,p_{\mathrm H}}
   \;=\;
  \max_{\nu}\,\frac{p_t(\nu)}{p_{\mathrm H}(\nu)}-1.
  \label{eq:rel_err_def}
\end{equation}
Relative error is generically a much stronger notion than additive error, since it requires accuracy uniformly across all $\nu$, including those of exponentially small weight. As discussed in Sec.~\ref{sec:designs}, the additive notion already suffices to certify indistinguishability of state ensembles, so we use additive error for state designs throughout.

\subsection{Continuum limit}
\label{sec:continuum}

The chain $T$ defined in Eq.~\eqref{eq:T_birth_death} is a biased random walk on $8\mathbb{Z}\cap[-2n,2n]$: at each step, $\nu$ moves to $\nu\pm 8$ with probabilities $a_\nu,\,b_\nu$ and stays put with probability $1-a_\nu-b_\nu$. The Haar weight~\eqref{eq:Haar_dist} is concentrated near $\nu=0$ with width of order $\sqrt{2n}$, which suggests rescaling $\nu$ on this scale. We therefore introduce the continuous coordinate $x=\nu/\sqrt{2n},\delta x=8/\sqrt{2n}$,
and view the action of $T$ on a smooth test function $f$ as
\begin{equation}
  (Tf)(\nu)\;=\;a_\nu f(\nu-8)+b_\nu f(\nu+8)+\bigl(1-a_\nu-b_\nu\bigr)f(\nu).
\end{equation}
Expanding $f(x\pm\delta x)$ to second order and using the large-$n$ asymptotics of the rates~\eqref{eq:T_birth_death},
\begin{equation}
  a_\nu+b_\nu\;=\;\tfrac{1}{8}+\mathcal{O}(n^{-1}),
  \qquad
  b_\nu-a_\nu\;=\;-\frac{\nu}{4n}+\mathcal{O}(n^{-2}),
\end{equation}
one obtains the continuum action
\begin{equation}
  Tf(x)\;=\;f(x)+\frac{1}{n}\,\mathcal{L}f(x)+\mathcal{O}(n^{-2}),
  \label{eq:continuum_T}
\end{equation}
where
\begin{equation}
  \mathcal{L}\;=\;2\partial_x^2-2x\partial_x
  \label{eq:OU_generator}
\end{equation}
is the generator of the Ornstein--Uhlenbeck process. 

The OU generator has a discrete real spectrum $\mu_k(\mathcal{L})=-2k$ for $k=0,1,2,\dots$, with eigenfunctions given by Hermite polynomials. Through Eq.~\eqref{eq:continuum_T}, this translates into the leading-order eigenvalues of the transition matrix,
\begin{equation}
  \lambda_k\;=\;1-\frac{2k}{n}+\mathcal{O}(n^{-2}),
  \label{eq:beta_k_continuum}
\end{equation}
so that the spectral gap predicted by the continuum limit is
\begin{equation}
  \Delta\;=\;\frac{2}{n}+\mathcal{O}(n^{-2}),
  \label{eq:gap_continuum}
\end{equation}
where we do not take the minimum eigenvalue into account in the continuum limit (it will be addressed at finite $n$ in Sec.~\ref{sec:gap_bounds}).
Rigorous upper and lower bounds matching the $1/n$ scaling are derived in Sec.~\ref{sec:gap_bounds} below.

Equation~\eqref{eq:continuum_T} describes the backward evolution of observables. The forward evolution of $p_t$ is governed by the adjoint $\mathcal{L}^\dagger=2\partial_x^2+2\partial_x(x\,\cdot)$. Introducing the rescaled time $\tau=t/n$ and taking $n\to\infty$, this yields the Fokker--Planck equation
\begin{equation}
  \partial_\tau p\;=\;\partial_x(2xp)+2\partial_x^2 p.
  \label{eq:FP}
\end{equation}
For the initial condition $p_0=\delta_{\nu,0}$ relevant to $\psi_0=|0\rangle^{\otimes n}$, i.e.\ $x_0=0$, the solution is the Gaussian
\begin{equation}
    p(x,\tau)\;=\;\frac{1}{\sqrt{2\pi\sigma^2(\tau)}}\,
  \exp\!\Bigl[-\tfrac{(x-\mu(\tau))^2}{2\sigma^2(\tau)}\Bigr],
  \label{eq:gaussian_solution}
\end{equation}
where
\begin{equation}
    \mu(\tau)=x_0\,e^{-2\tau},
  \quad
  \sigma^2(\tau)=1-e^{-4\tau}.
\end{equation}
As $\tau\to\infty$, the stationary distribution is the standard Gaussian $p_{\mathrm H}(x)=(2\pi)^{-1/2}e^{-x^2/2}$, which is precisely the rescaled large-$n$ limit of~\eqref{eq:Haar_dist}.

The continuum solution~\eqref{eq:gaussian_solution} produces explicit predictions for the state design. By~\eqref{eq:TV_chain}, the trace distance to the Haar distribution can be computed from the total-variation distance between two zero-mean Gaussian distributions $p(x,\tau)$ and $p_{\mathrm H}(x)$, which admits a closed-form solution:
\begin{equation}
  \tfrac{1}{2}\|p_\tau-p_{\mathrm H}\|_1
  \;=\;
  \Phi\!\Bigl(\tfrac{x_*}{\sigma(\tau)}\Bigr)-\Phi(x_*)
  \;\stackrel{\tau\gg 1}{\simeq}\;
  \frac{1}{\sqrt{2\pi e}}\,e^{-4\tau}+\mathcal{O}(e^{-8\tau}),
  \label{eq:TV_continuum}
\end{equation}
where $\Phi$ is the standard normal CDF. This predicts that the additive error of the state $2$-design decays as $e^{-4\tau}=e^{-4t/n}$, so that an additive-error $\epsilon$-approximate state $2$-design is reached at
\begin{equation}
  t\;\simeq\;\frac{n}{4}\log\!\frac{1}{\epsilon}.
  \label{eq:tstar_continuum}
\end{equation}
The $e^{-4\tau}$ decay of~\eqref{eq:TV_continuum} is the relaxation
shown numerically in the inset of Fig.~\ref{fig:scheme}(\textbf{C}), confirming the
continuum prediction.

For comparison, the state frame potential can also be computed in closed-form using Eq.~\eqref{eq:sfp} as
\begin{equation}
\begin{aligned}
    \mathcal{S}(t;\psi_0)/\mathcal{S}^{(2)}_{\mathrm H}-1 &= \int dx \, p(x,\tau)^2/p_{\mathrm{H}}(x) - 1 \\
    &= \frac{1}{\sqrt{1-e^{-8\tau}}} - 1\\
    &\approx \frac{1}{2} e^{-8\tau} + \mathcal{O}(e^{-16\tau}),
\end{aligned}
\end{equation}
whose square root decays at the same rate as the trace distance. This asymptotic scaling also confirms our numerical results in Fig.~\ref{fig:global_2f}.

The continuum description above is the leading term of a systematic expansion of the transfer matrix in powers of $1/n$,
\begin{equation}
  T\;=\;I+\frac{1}{n}\,\mathcal{L}_0+\frac{1}{n^2}\,\mathcal{L}_1
        +\frac{1}{n^3}\,\mathcal{L}_2+\mathcal{O}(n^{-4}),
  \label{eq:T_series}
\end{equation}
obtained by Taylor-expanding the lattice shifts $f(\nu\pm 8)$ to sixth order,
with $\mathcal{L}_0=\mathcal{L}$ the OU generator~\eqref{eq:OU_generator} and
$\mathcal{L}_{1,2}$ higher-derivative corrections. Rayleigh--Schr\"odinger
perturbation theory in the Hermite basis then gives the low-lying eigenvalues to
higher order (App.~\ref{app:OU_perturbation}),
\begin{equation}
  \lambda_k\;=\;1-\frac{2k}{n}+\frac{3k(k-1)}{n^2}
              +\frac{3k(4k-3)}{2n^3}+\mathcal{O}(n^{-4}),
  \label{eq:lambda_series}
\end{equation}
refining the leading prediction~\eqref{eq:gap_continuum} to
$\Delta=2/n-3/(2n^3)+\mathcal{O}(n^{-4})$, with no $n^{-2}$ term because the
$k(k-1)$ correction vanishes at $k=1$.  Using~\eqref{eq:lambda_series} in the frame potential expression~\eqref{eq:FP_spectral}, we obtain an analytic formula for $\mathcal F(t)$ to $\mathcal{O}(n^{-3})$ accuracy, shown as the red line in Fig.~\ref{fig:global_2f}.

\subsection{Bounds on the spectral gap}
\label{sec:gap_bounds}

We now turn to rigorous, finite-$n$ bounds on the spectral gap $\Delta$ defined in~\eqref{eq:gap_def_markov}. 

\begin{thm}[Spectral gap scaling]\label{thm:spectrum}
The spectral gap of $T$ satisfies
\begin{equation}
    \Delta(T) = \Theta(1/n).
    \label{eq:gap_scaling}
\end{equation}
\end{thm}

We prove the theorem by establishing matching upper and lower bounds $\Delta(T) = \mathcal{O}(1/n)$ and $\Delta(T) = \Omega(1/n)$.

\subsubsection{Upper bound} 
We first establish an upper bound on $\Delta$:
\begin{lem}[Upper bound on the spectral gap]\label{lem:gap_upper}
For all $n\ge 1$,
\begin{equation}
    \Delta(T) \;\le\; \frac{2(2^n - 8)}{n(2^n + 2) - 2n(2n-1)} \;=\; \frac{2}{n} + \mathcal{O}(2^{-n}).
    \label{eq:upper_bound_gap}
\end{equation}
\end{lem}
\begin{proof}

For a reversible Markov chain with stationary distribution $\pi$, the second-largest eigenvalue admits the variational characterization
\begin{equation}
    1-\lambda_1 \;=\; \inf_{f}\,\frac{\mathcal E(f, f)}{\mathrm{Var}_\pi(f)},
    \label{eq:gap_variational}
\end{equation}
where the infimum is over all non-constant $f:\mathcal S\to\mathbb R$, $\mathrm{Var}_\pi(f) = \mathbb E_\pi[f^2] - \mathbb E_\pi[f]^2$, and the Dirichlet form is
\begin{equation}
    \mathcal E(f, f) \;=\; \tfrac{1}{2}\sum_{\nu, \nu'}\pi(\nu)\,T(\nu, \nu')\,(f(\nu) - f(\nu'))^2.
    \label{eq:dirichlet_form}
\end{equation}
Therefore, any specific test function $f$ yields an upper bound on $1-\lambda_1$ (and hence the spectral gap $\Delta$).

Motivated by the analogy with the Ehrenfest urn (briefly reviewed in App.~\ref{app:ehrenfest}), whose first nontrivial eigenfunction is linear in the position, and by the continuum prediction that the leading non-stationary mode of the OU generator is the first Hermite polynomial $H_1(x)\propto x$, we choose the linear test function
\begin{equation}
    f(\nu) \;=\; \nu.
\end{equation}
By the parity symmetry $p_{\mathrm H}(-\nu) = p_{\mathrm H}(\nu)$, the function $f$ has zero mean under $\pi = p_{\mathrm H}$, so $\mathrm{Var}_\pi(f) = \mathbb E_{p_{\mathrm H}}[\nu^2]$. Using the birth--death structure of $T$, the Dirichlet form reduces to
\begin{equation}
    \mathcal E(f, f) \;=\; \sum_\nu p_{\mathrm H}(\nu)\,a_\nu\,(\nu - (\nu+8))^2 \;=\; 64\sum_\nu p_{\mathrm H}(\nu)\,a_\nu,
\end{equation}
where detailed balance collapses the symmetric double sum into a single sum over upper-diagonal elements. Evaluating the two stationary expectations $\mathbb E_{p_{\mathrm H}}[a_\nu]$ and $\mathbb E_{p_{\mathrm H}}[\nu^2]$ yields, after some algebra,
\begin{equation}
    \Delta(T) \;\le\; \frac{\mathcal E(f, f)}{\mathrm{Var}_{p_{\mathrm H}}(f)} \;=\; \frac{2(2^n - 8)}{n(2^n+2) - 2n(2n-1)} \;=\; \frac{2}{n} + \mathcal{O}(2^{-n}),
\end{equation}
which is the stated upper bound.
\end{proof}

Note that the leading $2/n$ scaling matches the continuum prediction~\eqref{eq:gap_continuum}.

\subsubsection{Lower bound}
In order to establish a matching lower bound of order $\Omega(1/n)$ on the spectral gap, we make use of two classical Markov-chain inequalities based on geometric quantities of the chain: the Cheeger inequality~\cite{sinclair1989approximate,lawler1988bounds} and the dual Cheeger inequality~\cite{trevisan2009max,bauer2013bipartite}. 

For a reversible chain $T$ on a finite state space $\mathcal{S}$ with stationary distribution $\pi$, define the edge measure $Q(\nu,\nu'):=\pi(\nu)T(\nu,\nu')$, extended to subsets by $Q(A,B):=\sum_{\nu\in A,\nu'\in B}Q(\nu,\nu')$. The bottleneck ratio is defined as
\begin{equation}
    \Phi_*(T) \;:=\; \min_{S\subset\mathcal{S}:\,\pi(S)\leq 1/2}\,\frac{Q(S,S^c)}{\pi(S)}.
    \label{eq:cheeger_const_main}
\end{equation}
The Cheeger inequality bounds the second-largest eigenvalue in terms of the bottleneck ratio~\cite[Thm.~13.10]{LevinPeresWilmer2006}:
\begin{equation}
     1 - \lambda_1(T) \;\ge\; \tfrac{1}{2}\,\Phi_*(T)^{2}.
    \label{eq:cheeger_main}
\end{equation}

Geometrically, $\Phi_*$ is the smallest one-step escape rate across any cut $S \mid S^c$: for any partition, $\Phi(S) = Q(S, S^c)/\pi(S)$ is the rate at which probability mass leaves $S$, normalized by the stationary mass of $S$. The Cheeger inequality says that the gap to the largest eigenvalue $\lambda_0=1$ is at most quadratically larger than this geometric bottleneck.

The dual quantity controls the distance of the smallest eigenvalue $\lambda_{\min}(T)$ from $-1$. Define the bipartiteness ratio
\begin{equation}
\begin{split}
\beta(T) :=\;& \inf_{\substack{V_1, V_2\subset\mathcal S \\
V_1\cap V_2 = \emptyset}}
\\
&\frac{
2\,Q(V_1, V_1) + 2\,Q(V_2, V_2)
+ Q(V_1, V_3) + Q(V_2, V_3)
}{
\pi(V_1) + \pi(V_2)
},
\end{split}
\label{eq:bipart_main}
\end{equation}

\noindent
where $V_3 := \mathcal{S}\setminus(V_1\cup V_2)$. The dual Cheeger inequality bounds the smallest eigenvalue in terms of the bipartiteness ratio~\cite{trevisan2009max,bauer2013bipartite}:
\begin{equation}
      1-\lvert\lambda_{\min}(T)\lvert\;\ge\; \tfrac{1}{2}\,\beta(T)^{2}.
    \label{eq:dual_cheeger_main}
\end{equation}

Geometrically, $\beta(T)$ measures how close the chain is to being bipartite: for any pair of disjoint sets $(V_1, V_2)$, the numerator counts the probability flow that stays within a part (the $Q(V_i, V_i)$ terms) or leaks to the rest $V_3 := \mathcal S\setminus(V_1\cup V_2)$. In a perfectly bipartite chain, all transitions go between $V_1$ and $V_2$, so that $\beta = 0$ and $\lambda_{\rm min}=-1$. Instead, a large $\beta$ means every potential bipartition leaves substantial within-part flow, thus the chain is far from being bipartite. The dual Cheeger inequality translates this geometric distance from bipartiteness into a gap between $\lambda_{\rm min}$ and $-1$.

Using these inequalities, we establish the following bounds on the spectrum of $T$.
\begin{lem}[Bound on the second-largest eigenvalue]\label{lem:gap_lower}
For $n$ sufficiently large,
\begin{equation}
    1-\lambda_1(T) \;\ge\; \frac{1}{32\pi\,n}.
    \label{eq:gap_lower}
\end{equation}
\end{lem}

The proof applies the Cheeger inequality (Eq.~\eqref{eq:cheeger_main}) to $T$, with the bottleneck ratio computed by exploiting the substantial simplification of the optimization afforded by the birth--death structure. The full argument is given in App.~\ref{app:gap_lower}.

To complete the bound on the spectral gap, we next bound the distance of the minimum eigenvalue from $-1$. A common approach in Markov-chain theory is to replace the chain by the lazy chain $T'=(I+T)/2$, which shifts the spectrum away from $-1$. This corresponds to applying the SWAP gate with probability $1/2$ at each doping round. However, this also slows down the dynamics by a factor of $2$, which is undesirable for practical applications. Instead, we show that the original SWAP chain already has its spectrum bounded away from $-1$.

\begin{lem}[Bound on the smallest eigenvalue]\label{lem:min_eigenvalue}
For $n$ sufficiently large,
\begin{equation}
    1-\lvert\lambda_{\min}(T)\rvert \;\ge\; \frac{1}{32\pi\,n}.
    \label{eq:min_eigenvalue}
\end{equation}
\end{lem}

The proof applies the dual Cheeger inequality (Eq.~\ref{eq:dual_cheeger_main}) to $T$. The bipartiteness ratio $\beta(T)$ is bounded below by a two-regime argument over partitions $(V_1, V_2)$: when $V_1\cup V_2 = \mathcal S$, the bound comes from the self-loop probability at the central state $\nu = 0$, which is $\Omega(1/\sqrt n)$; when $V_1\cup V_2 \subsetneq \mathcal S$, the bound reduces to the bottleneck ratio of $T$ at the cut, which is shown to be $\Omega(1/\sqrt n)$ in the proof of Lemma~\ref{lem:gap_lower}. The full argument is in App.~\ref{app:min_eigenvalue}.

Together, Lemmas~\ref{lem:gap_upper},~\ref{lem:gap_lower}, and~\ref{lem:min_eigenvalue} establish Theorem~\ref{thm:spectrum}: the spectral gap of $T$ scales as $\Delta=\Theta(1/n)$. This precise scaling is
in agreement with the leading order $\Delta_{\mathrm{cont}}=2/n+\mathcal{O}(n^{-2})$ predicted by the continuum analysis of Sec.~\ref{sec:continuum}. While the prefactors in~\eqref{eq:gap_lower} and~\eqref{eq:min_eigenvalue} are likely not optimal, the $1/n$ scaling is all we need below: it directly translates into rigorous bounds on the convergence time to a unitary $2$-design (Sec.~\ref{sec:unitary_design}).

Figure~\ref{fig:gap_layout}(\textbf{A}) confirms this numerically:
the gap of the two-copy doping matrix follows $\Delta=2/n$ at leading order,
with subleading corrections of order $n^{-3}$ shown in the inset.
Panel~(\textbf{B}) shows that the gap of the three-copy transfer matrix obeys
the same $2/n$ scaling, indicating that the relaxation mechanism identified at
two copies persists at higher moments, even though a complete Markov-chain
proof is not available there.

\begin{figure}[h!]
    \centering
    \includegraphics[width=1\linewidth]{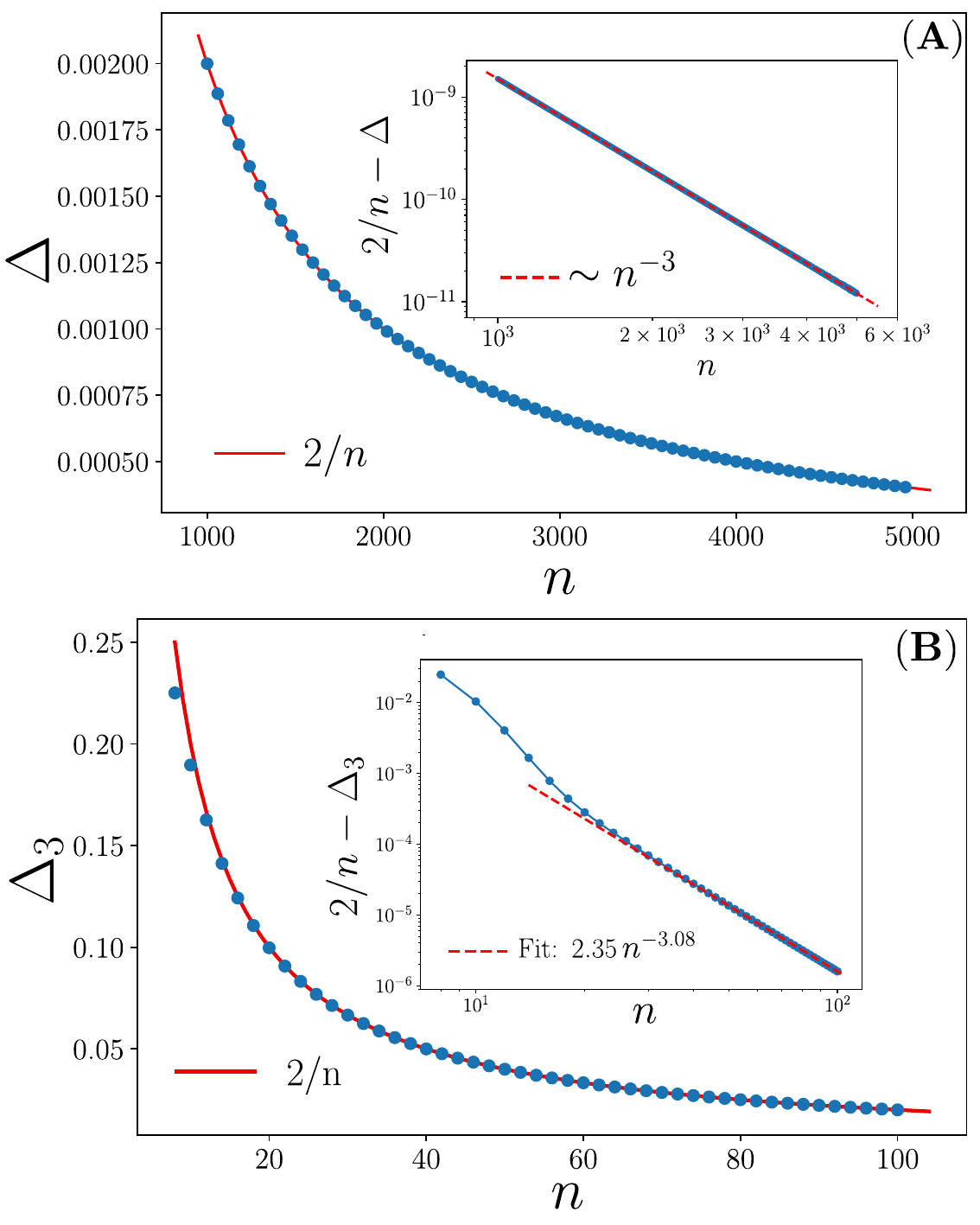}
    \caption{System-size scaling of the spectral gap $\Delta$
    [Eq.~\eqref{eq:gap_def_markov}] of the doping matrix for (\textbf{A}) two
    and (\textbf{B}) three copies. In both cases the leading behaviour is
    consistent with the Ornstein--Uhlenbeck value $\Delta\sim 2/n$, matching
    the rigorous $\Delta=\Theta(1/n)$ of Eq.~\eqref{eq:gap_scaling} at two
    copies. Insets: subleading corrections, scaling as $\sim n^{-3}$.}
    \label{fig:gap_layout}
\end{figure}

This agreement is not accidental. At finite $n$ the two-copy spectrum is
contained exactly in the three-copy one. Take operators $Y\otimes\mathbb 1$ with
$Y\in\mathrm{Com}_2(\mathcal M_n)$ on a fixed pair of replicas and the identity on
the third; since $A\,\mathbb 1\,A^\dagger=U\,\mathbb 1\,U^\dagger=\mathbb 1$, both
ingredients of the replicated dynamics leave the third factor inert, the doping
conjugation acting only on the first two replicas as $A^{\otimes2}YA^{\dagger\otimes2}$
and the three-copy twirl reducing to its two-copy form $\Phi_{\mathcal M}(Y)$. The
subspace $\mathrm{Com}_2(\mathcal M_n)\otimes\mathbb 1$ is therefore invariant, with
$C^{(3)}$ acting on it exactly as the two-copy transfer matrix $\tilde C$ (and
likewise for the other replica pairs). Hence,
$\mathrm{spec}(\tilde C)\subseteq\mathrm{spec}(C^{(3)})$, so the subleading
eigenvalue $1-\Delta$ also appears in $C^{(3)}$ and $\Delta_3\le\Delta$. Numerically, we find that the bound is saturated,
$\Delta_3=\Delta$ to $\sim10^{-14}$ for all $8\le n\le100$ including the $n^{-3}$
corrections, indicating that the slowest three-copy mode is the embedded two-copy
one. Taken together, these observations suggest that the same two-copy behavior governs
the leading relaxation at $k=3$, consistent with the $2/n$ scaling seen in
Fig.~\ref{fig:gap_layout}(B); a complete analytical treatment,
however, remains open.

\subsection{Mixing time lower bound}
\label{sec:mixing_time}

Beyond the spectral gap, a sharper characterisation of convergence is provided by the \emph{mixing time} of the chain,
\begin{equation}
  t_{\mathrm{mix}}(\epsilon)
  \;:=\;
  \min\bigl\{t\geq 0 : \max_{\nu_0}\,\tfrac{1}{2}\|p_t^{(\nu_0)}-p_{\mathrm H}\|_1 \leq \epsilon\bigr\},
\end{equation}
where $p_t^{(\nu_0)}$ denotes the chain law at time $t$ started from the deterministic state $\nu_0$. 

The mixing time of $T$ is expected to follow the same scaling as the Ehrenfest urn model~\cite{LevinPeresWilmer2006}, to which our chain is closely analogous (see App.~\ref{app:ehrenfest}): for the Ehrenfest model the mixing time scales as $t_{\mathrm{mix}}^{\mathrm{Ehr}}(\epsilon)=\frac{n}{2}\log n + \mathcal{O}(n)$. We establish rigorously (see App.~\ref{app:mixing_time}) that the mixing time of $T$ satisfies
\begin{equation}
  t_{\mathrm{mix}}(\epsilon)\;\geq\;\frac{n}{4}\log n + K_\epsilon\, n,
  \label{eq:mixing_lower_main}
\end{equation}
for every fixed $\epsilon\in(0,1)$, where $K_\epsilon$ is a constant depending only on $\epsilon$, and $n$ is taken sufficiently large.

The proof of~\eqref{eq:mixing_lower_main} rests on a simple idea: if one can find a function $f$ whose expectation under $p_t$ differs substantially from its expectation under $p_{\mathrm H}$, then $p_t$ and $p_{\mathrm H}$ must be far apart in total variation. More precisely, for any $f$ with bounded variance, the distinguishing-statistic bound~\cite{LevinPeresWilmer2006} gives
\begin{equation}
  \tfrac{1}{2}\|p_t - p_{\mathrm H}\|_1
  \;\geq\;
  1 - \frac{8\,\sigma^2}{\bigl(\mathbb{E}_{p_t}[f] - \mathbb{E}_{p_{\mathrm H}}[f]\bigr)^2},
  \label{eq:dist_stat_main}
\end{equation}
where $\sigma^2 := \max(\mathrm{Var}_{p_t}(f), \mathrm{Var}_{p_{\mathrm H}}(f))$. The right-hand side is close to $1$ — and hence the chain is far from mixed — as long as the mean gap $|\mathbb{E}_{p_t}[f] - \mathbb{E}_{p_{\mathrm H}}[f]|$ is much larger than $\sigma$.

The strategy then follows the eigenfunction-based template~\cite{LevinPeresWilmer2006}: pick $f$ to be an (approximate) eigenfunction of $T$ with eigenvalue close to $1$, so that $\mathbb{E}_{p_t}[f]$ decays slowly while $\mathbb{E}_{p_{\mathrm H}}[f]=0$ by stationarity. For the Ehrenfest model, the linear function $\phi(\nu)=\nu$ is an \emph{exact} eigenfunction of the transition matrix~\cite{LevinPeresWilmer2006}, so the mean gap decays as a pure exponential and the lower bound follows in a few elementary steps, matching the optimal mixing time at leading order. In our setting, $\phi(\nu)=\nu$ is only an \emph{approximate} eigenfunction, so the mean $M_t:=\mathbb{E}_{p_t}[\nu]$ no longer follows an exact exponential decay but acquires a subleading correction. Consequently, the mean and variance of $\phi$ along the chain must be tracked carefully to yield a meaningful lower bound, making the analysis more involved than in the Ehrenfest case.  The full proof is given in App.~\ref{app:mixing_time}.

For the Ehrenfest urn, a matching upper bound is also available, pinning down the exact leading-order behaviour of the mixing time~\cite{LevinPeresWilmer2006}. It is plausible that an analogous result holds in our setting. We do not pursue this direction here, however, since such an upper bound would not sharpen any of our bounds on unitary designs.

\section{Convergence to an approximate unitary $2$-design}
\label{sec:unitary_design}
We now translate the results of the previous section into rigorous bounds on
the time at which the doped ensemble forms an $\epsilon$-approximate unitary
$2$-design. We adopt the relative-error notion of
Eq.~\eqref{eq:relative_unitary}, the stronger of the two notions introduced in
Sec.~\ref{sec:designs}, which certifies indistinguishability from Haar against
adaptive $k$-query algorithms; the bounds derived below apply equally to the
weaker additive (diamond-norm) notion of Eq.~\eqref{eq:additive_unitary}.

By Choi--Jamio{\l}kowski, the relative-error condition~\eqref{eq:relative_unitary}
is equivalent to the operator ordering
\begin{equation}\label{eq:rel_error_choi}
    (1-\epsilon)\,\rho_{\mathrm H}\;\preceq\;\rho_{\mathcal E}\;\preceq\;(1+\epsilon)\,\rho_{\mathrm H}
\end{equation}
between the Choi states
$\rho_{\mathcal E}\coloneqq(\Phi_{\mathcal E}\otimes\mathbb 1)(P_{\mathrm{EPR}})$, where
$P_{\mathrm{EPR}}$ projects onto the maximally entangled state on
$\mathcal H^{\otimes 2}\otimes\mathcal H^{\otimes 2}$. The Choi states are again
block-diagonal across the residue classes mod $8$. We work throughout in the
single block $\nu\equiv 0\pmod 8$, that is $\nu\in\mathcal S=8\mathbb Z\cap[-2n,2n]$, while the remaining blocks can be treated
analogously. 

Within this block, the Choi state takes the form
\begin{equation}\label{eq:choi_shell}
    \rho_\mathcal{E}\;=\;\sum_{\nu,\nu'\in\mathcal S}A_{\nu,\nu'}\,
    \frac{P_\nu\otimes P_{\nu'}}{\Tr P_{\nu'}},
\end{equation}
where the matrix $A$ encodes the action on the projectors $\{P_\nu\}_{\nu\in\mathcal S}$: for the doped
chain after $t$ doping rounds $A_{\nu,\nu'}=T^t(\nu',\nu)$, while for the Haar-random ensemble $A=\mathbf 1\,p_{\mathrm H}^{\!\top}$, the
projector onto the stationary distribution. Because the $P_\nu\otimes P_{\nu'}$ are mutually orthogonal, the relative error reduces to
\begin{equation}\label{eq:eps_is_dinfty}
    \epsilon\;=\;d_\infty(t)\coloneqq\max_{\nu',\nu\in\mathcal S}\left|\frac{T^t(\nu',\nu)}{p_{\mathrm H}(\nu)}-1\right|.
\end{equation}
The relative error of the doped ensemble as a unitary $2$-design is therefore
exactly the distance of the underlying chain to stationarity in the
$p_{\mathrm H}$-weighted $\infty$-norm.

It is worth noting that, for a reversible chain, the maximum
in~\eqref{eq:eps_is_dinfty} is attained on the diagonal at even
times~\cite{LevinPeresWilmer2006},
\begin{equation}\label{eq:dinfty_diagonal}
    d_\infty(2t)\;=\;\max_{\nu\in\mathcal S}\left(\frac{T^{2t}(\nu,\nu)}{p_{\mathrm H}(\nu)}-1\right).
\end{equation}
 A similar observation has been made for the relative-error analysis of
local random circuits~\cite{heinrich2025anticoncentrationalmostneed,belkin2025apparentuniversalbehaviorsecond}; indeed, that setting can also be analyzed using a Markov-chain
mapping analogous to the one developed here.

Bounding the convergence to relative-error $\epsilon$-approximate unitary $2$-design thus reduces to bounding the convergence of $T$ to
$p_{\mathrm H}$ in $d_\infty$, which we carry out from both sides below.

\paragraph{Upper bound.}
For a reversible chain, $d_\infty(t)$ is upper-bounded
by the spectral gap through the standard
bound~\cite{LevinPeresWilmer2006},
\begin{equation}\label{eq:dinfty_upper}
    \epsilon\;=\;d_\infty(t)\;\le\;\frac{(1-\Delta)^{\,t}}{\min_{\nu\in \mathcal{S}} p_{\mathrm H}(\nu)}
    .
\end{equation}
Combined with $\Delta=\Theta(1/n)$ from~\eqref{eq:gap_scaling} and with
$\min_\nu p_{\mathrm H}(\nu)=\Theta(2^{-2n})$, hence
$\log(1/\min_\nu p_{\mathrm H})=\Theta(n)$, this gives
\begin{equation}
    t_{2\text{-des}}(\epsilon)\;=\;\mathcal{O}\!\bigl(n\,(n+\log(1/\epsilon))\bigr).
    \label{eq:design_upper}
\end{equation}

\paragraph{Lower bound from the spectral gap.}
The same spectral gap bounds $d_\infty$ from below for reversible chains~\cite{LevinPeresWilmer2006}:
\begin{equation}
    \epsilon=d_\infty(t)\ge(1-\Delta)^t.
\end{equation}
 Inserting the upper bound on $\Delta$ in~\eqref{eq:upper_bound_gap} therefore yields
\begin{equation}
  t_{2\text{-des}}(\epsilon)\;\geq\;\frac{\log(1/\epsilon)}{-\log(1-\Delta)}
  \;\geq\;\frac{n}{2}\log(1/\epsilon)\bigl(1+o(1)\bigr).
  \label{eq:design_lower_gap}
\end{equation}

\paragraph{Lower bound from the mixing time.}
A second lower bound follows directly from the fact that the weighted $\infty$-norm distance dominates the total-variation distance to
stationarity~\cite{LevinPeresWilmer2006},
\begin{equation}\label{eq:dinfty_geq_TV}
  d_\infty(t)\;\ge\;d(t)\;\coloneqq\;2\max_{\nu_0}\,\bigl\|p_t^{(\nu_0)}-p_{\mathrm H}\bigr\|_{\mathrm{TV}}.
\end{equation}
 Hence, the relative error cannot fall below $\epsilon$ before the chain has
mixed in total-variation, and the mixing-time lower bound~\eqref{eq:mixing_lower_main} translates into
\begin{equation}
  t_{2\text{-des}}(\epsilon)\;\geq\;t_{\mathrm{mix}}(\epsilon/2)\;\geq\;\frac{n}{4}\log n + K_{\epsilon/2}\,n.
  \label{eq:design_lower_mixing}
\end{equation}
This bound is strictly stronger than~\eqref{eq:design_lower_gap} in the regime
of constant $\epsilon$.

This analysis establishes that the doped matchgate ensemble forms an $\epsilon$-approximate unitary $2$-design with
 \begin{equation}
\Omega\!\left(n\max\!\bigl\{\log n,\log(1/\epsilon)\bigr\}\right)
\le
t_{2\text{-des}}
\le
O\!\left(n^2+n\log(1/\epsilon)\right).
\label{eq:2_design_bound}
\end{equation}
 
We close this section by noting that the gap between the upper and lower bounds
in~\eqref{eq:2_design_bound} can potentially be closed using more refined
Markov-chain tools. The likely origin of the gap is the upper bound from the
spectral gap~\eqref{eq:dinfty_upper}, which is controlled by the worst-case
prefactor $1/\min_\nu p_{\mathrm H}(\nu)=\Theta(2^{2n})$ and thereby carries an
additional $\log(1/\min_\nu p_{\mathrm H})=\Theta(n)$ factor that we expect to be
loose, exactly as it is for the Ehrenfest urn, where the same bound overestimates
the $\Theta(n\log n)$ mixing time. A
sharper analysis based on a logarithmic Sobolev
inequality~\cite{Diaconis1996}
typically replaces the factor $\log(1/\min_\nu p_{\mathrm H})$ by
$\log\log(1/\min_\nu p_{\mathrm H})=\Theta(\log n)$. Provided the log-Sobolev
constant of $T$ is of the same order as its spectral gap, $\alpha=\Theta(1/n)$ ---
as holds for the Ehrenfest urn --- this would tighten the upper bound to
$\Theta(n\log n)$, matching the lower bound in~\eqref{eq:2_design_bound} up to constants. This leads us to the following conjecture, whose proof would
require a log-Sobolev estimate for $T$ that we leave to future work.
\begin{conj}[Relative-error unitary $2$-design time]\label{conj:design_time}
The SWAP-doped matchgate circuit forms a relative-error $\epsilon$-approximate unitary $2$-design with
\begin{equation}\label{eq:conj_design_time}
    t_{2\text{-des}}(\epsilon)\;=\;\Theta\!\bigl(n\,(\log n+\log(1/\epsilon))\bigr).
\end{equation}
\end{conj}

\section{Extension to general $q$-qubit doping gates}
\label{sec:general_q_local}

The Markov-chain analysis of the previous sections was carried out for the SWAP doping, which is a particularly simple representative of two-qubit non-Gaussian gates. We now show that the same framework extends to an arbitrary $q$-qubit parity-preserving doping gate, with the only structural change being a richer band structure of the transfer matrix. Throughout this section, $A$ denotes a parity-preserving unitary supported on the first $q$ qubits of the chain.

\subsection{General $2$-qubit doping: $A=e^{i\theta Z_1 Z_2}$}
\label{sec:zz_doping}

Before turning to arbitrary $q$, it is illuminating to consider the simplest non-trivial case $q=2$. Up to conjugation by Gaussian unitaries (which act trivially on the matchgate commutant), every parity-preserving two-qubit gate can be written in the one-parameter form~\cite{hebenstreit2019all}
\begin{equation}
  A_\theta\;:=\;e^{i\theta\,Z_1 Z_2}\;=\;\cos\theta\,\mathbb{I}\;+\;i\sin\theta\,Z_1Z_2,
  \quad \theta\in[0,\pi/2),
  \label{eq:Atheta}
\end{equation}
with $\theta=0$ corresponding to the identity (no doping) and $\theta=\pi/4$ corresponding, after conjugation by Gaussian unitaries, to the SWAP gate~\cite{hebenstreit2019all}. The general case $\theta\in(0,\pi/2)$ interpolates between these two extremes.

 The only non-trivial action of $A_\theta$ on the matchgate commutant is to mix the two local sectors $h_1=h_2=h_3=h_4=\pm1$. A direct computation gives
\begin{equation}
  T(\theta)\;=\;\bigl(1-p(\theta)\bigr)\,\mathbb{I}\;+\;p(\theta)\,T(\pi/4), \quad p(\theta) = \sin^2(2\theta)
  \label{eq:T_theta}
\end{equation}
i.e.\ with probability $1-p(\theta)$ the gate $A_\theta$ leaves the configuration unchanged, and with probability $p(\theta)$ it executes the same transition as the SWAP gate. In other words, $T(\theta)$ is the \emph{lazy version} of the SWAP chain $T(\pi/4)$ analysed in Secs.~\ref{sec:markov_setup}--\ref{sec:unitary_design}, with laziness parameter $1-p(\theta)$. 

 Equation~\eqref{eq:T_theta} implies that every spectral and mixing-time statement established for the SWAP chain transfers verbatim to $A_\theta$, with some rescaling by $p(\theta)$. In particular:
\begin{itemize}
    \item if $\{\lambda_i\}$ are the eigenvalues of $T(\pi/4)$ then $\{1-p(\theta)(1-\lambda_i)\}$ are those of $T(\theta)$, so the spectral gap rescales as
    \begin{equation}
      \Delta(\theta)\;=\;p(\theta)\,\Delta(\pi/4)\;=\;\sin^2(2\theta)\,\Delta_{\mathrm{SWAP}}\;=\;\Theta\bigl(1/n\bigr);
      \label{eq:gap_theta}
    \end{equation}
    \item the mixing time and the unitary $2$-design convergence time both rescale by $1/p(\theta)$,
    \begin{equation}
      t_{\mathrm{mix}}^{(\theta)}(\epsilon)\;=\;\frac{t_{\mathrm{mix}}^{\mathrm{SWAP}}(\epsilon)}{\sin^2(2\theta)},
      \qquad
      t_{2\text{-des}}^{(\theta)}(\epsilon)\;=\;\frac{t_{2\text{-des}}^{\mathrm{SWAP}}(\epsilon)}{\sin^2(2\theta)},
    \end{equation}
    showing that doping with a "weaker" non-Gaussian rotation simply slows the chain by a constant factor.
\end{itemize}
The whole one-parameter family of two-qubit parity-preserving doping gates therefore exhibits identical properties, and the SWAP analysis already covers it once the laziness factor $\sin^2(2\theta)$ is restored. 

\subsection{General $q$-qubit framework}
\label{sec:general_q_local_framework}

We now turn to general treatment of $q$-qubit parity-preserving doping. As in Sec.~\ref{sec:Commutant_approach}, the matchgate twirl projects the doped channel onto the commutant $\Com_2(\MG_n)$, and the only $A$-dependent input to the dynamics is the local transition matrix
\begin{equation}
  R_A(\mu,\mu')\;:=\;\frac{\Tr\!\bigl[(A\otimes A)\,\Pi_\mu^{(q)}\,(A^\dagger\otimes A^\dagger)\,\Pi_{\mu'}^{(q)}\bigr]}{\binom{2q}{q+\mu/2}},
  \label{eq:R_local}
\end{equation}
where $\Pi_\mu^{(q)}$ is the local projector on the eigenspace of $\Lambda_{12\cdots q}$ (the $q$-qubit bridge operator) at eigenvalue $\mu\in\{-2q,-2q+2,\dots,2q\}$. As in the SWAP case, the parity-preserving structure of $A$ forces $R_{\mu,\mu'}\neq 0$ only when $\mu\equiv\mu'\pmod 8$~\cite{poetri26fermionic}, so the global transfer matrix $T$ is block-diagonal across $8\mathbb{Z}$ shells.

As an example, for the SWAP gate the only nonzero entries of $R$ are
\begin{equation}
  \begin{aligned}
    &R_{\rm SWAP}(-4,4) = R_{\rm SWAP}(4,-4) = 1, \\
    &R_{\rm SWAP}(-2,-2) = R_{\rm SWAP}(0,0) = R_{\rm SWAP}(2,2) = 1.
  \end{aligned}
  \label{eq:R_SWAP}
\end{equation}
with all other entries vanishing. 

\paragraph{Multi-step birth--death structure.} The corresponding transition matrix $T_A$, defined analogously to Eq.~\eqref{eq:T_def}, has nonzero elements
\begin{equation}
  T_A(\nu,\nu\pm 8j)\;=\;a_\nu^{(j)},\,b_\nu^{(j)},
  \qquad
  T(\nu,\nu)\;=\;1-\sum_{j=1}^{\lfloor q/2\rfloor}\bigl(a_\nu^{(j)}+b_\nu^{(j)}\bigr),
  \label{eq:T_q_local}
\end{equation}
with rates
\begin{equation}
\begin{aligned}
  a_\nu^{(j)} &\;=\;\frac{1}{\binom{2n}{n+\nu/2}}
  \sum_\mu R_{\mu,\mu-8j}\,\binom{2q}{q+\mu/2}\binom{2n-2q}{n+\nu/2-q-\mu/2},\\[2pt]
  b_\nu^{(j)} &\;=\;\frac{1}{\binom{2n}{n+\nu/2}}
  \sum_\mu R_{\mu,\mu+8j}\,\binom{2q}{q+\mu/2}\binom{2n-2q}{n+\nu/2-q-\mu/2}.
\end{aligned}
  \label{eq:ab_q_local}
\end{equation}
 The chain $T_A$ is therefore a multi-step birth--death chain on the same state space $8\mathbb{Z}\cap[-2n,2n]$ as before, with up to $\lfloor q/2\rfloor$ neighbours on each side, i.e. the transition matrix is $\lfloor q/2\rfloor$-banded. Under the (generic) condition $a_\nu^{(1)},b_\nu^{(1)}\neq0$ for all $\nu$ in the bulk, the chain is irreducible and aperiodic, with stationary distribution $p_{\mathrm H}$ from~\eqref{eq:Haar_dist} (independent of $A$, since the Haar distribution $p_{\mathrm H}$ is always left invariant by the doping channel).

\paragraph{Reversibility.} A simple sufficient condition for detailed balance~\eqref{eq:dbc} of the chain $T_A$ is that the local transition matrix satisfies a local detailed balance
\begin{equation}
  \binom{2q}{q+\mu/2}\,R_A(\mu,\mu')\;=\;\binom{2q}{q+\mu'/2}\,R_A(\mu',\mu),
  \label{eq:R_dbc}
\end{equation}
or equivalently $\Tr[(A\otimes A)\Pi_\mu(A^\dagger\otimes A^\dagger)\Pi_{\mu'}]=\Tr[(A\otimes A)\Pi_{\mu'}(A^\dagger\otimes A^\dagger)\Pi_\mu]$. Indeed, substituting~\eqref{eq:ab_q_local} into~\eqref{eq:dbc} and using~\eqref{eq:R_dbc} shows that the condition propagates from $R$ to $T$. The condition~\eqref{eq:R_dbc} is automatic when $A^\dagger=A$ (Hermitian doping), or more generally when $A^\dagger=M_1 A M_2$ for matchgates $M_1,M_2$.

Note that the Markov chain remains a birth--death chain for $q=2,3$. As a result, the chain is always reversible~\cite{LevinPeresWilmer2006}. 

Note also that the process can be made reversible in a simple manner by considering a protocol in which, at each doping round, one applies either $A$ or $A^\dagger$ with equal probability~\cite{haferkamp2022efficient}. The resulting transfer matrix governing the dynamics is then $\Tilde{C}'=\tfrac{1}{2}\left(\Tilde{C}+\Tilde{C}^T\right)$,
which is symmetric by construction.

\paragraph{Continuum limit and universal scaling.} The continuum-limit analysis of Sec.~\ref{sec:continuum} extends to arbitrary $q$-qubit $A$. For $\nu=\mathcal{O}(\sqrt n)$, expanding the second binomial in~\eqref{eq:ab_q_local} around $\nu=0$ gives
\begin{equation}
\begin{aligned}
  a_\nu^{(j)} &\;=\;2^{-2q}\sum_\mu R_{\mu,\mu-8j}\,\binom{2q}{q+\mu/2}\Bigl(1+\frac{\mu\nu}{2n}+\mathcal{O}(\nu^2/n^2)\Bigr),\\[2pt]
  b_\nu^{(j)} &\;=\;2^{-2q}\sum_\mu R_{\mu,\mu+8j}\,\binom{2q}{q+\mu/2}\Bigl(1+\frac{\mu\nu}{2n}+\mathcal{O}(\nu^2/n^2)\Bigr).
\end{aligned}
  \label{eq:ab_q_bulk}
\end{equation}
The symmetry $R_A(\mu,\mu')=R_A(-\mu,-\mu')$, which follows from the parity invariance of $A$~\cite{poetri26fermionic}, then implies
\begin{equation}
  a_\nu^{(j)}+b_\nu^{(j)}\;=\;s_j\bigl(1+\mathcal{O}(\nu^2/n^2)\bigr),
  \qquad
  a_\nu^{(j)}-b_\nu^{(j)}\;=\;\Theta(\nu/n),
  \label{eq:ab_q_drift_diff}
\end{equation}
with
\begin{equation}
  s_j\;=\;2^{1-2q}\sum_\mu R_{\mu,\mu-8j}\,\binom{2q}{q+\mu/2},
  \label{eq:s_j}
\end{equation}
exactly as in the SWAP case but with $A$-dependent prefactors. Proceeding as in Sec.~\ref{sec:continuum} with the continuum variables $\nu=x\sqrt{2n}$, $t=n\tau$, we obtain the continuum action
\begin{equation}
    T_A = I + \frac{C_A}{2n} \mathcal{L},
    \label{eq:action_q_local}
\end{equation}
where $\mathcal{L}$ is the OU generator~\eqref{eq:OU_generator} and the diffusion constant $C_A>0$ reads
\begin{equation} \label{eq:diff_constant}
    C_A = \frac{1}{4}\sum_j (8j)^2 s_j.
\end{equation}
The Gaussian solutions and stationary distribution are analogous from Eq.~\eqref{eq:gaussian_solution}, and the spectrum~\eqref{eq:beta_k_continuum} of the continuum generator now reads $1-C_A k/n$. In particular, the spectral gap in the continuum is $\Delta_{\mathrm{cont}}=C_A/n$. We also recover the Fokker--Planck equation,
\begin{equation}
  \partial_\tau p\;=\;C_A\,\partial_x(x\,p)\;+\;C_A\,\partial_x^2 p,
  \label{eq:FP_q_local}
\end{equation}
with a single $A$-dependent diffusion constant $C_A>0$. Note that the fact that the drift and diffusion constants match is a direct consequence of the stationary
distribution being the Haar distribution $p_{\mathrm H}$, which is independent of $A$. The convergence time to an $\epsilon$-approximate state $2$-design scales as
\begin{equation}
  t_{2\text{-des}}^{\mathrm{state}}(\epsilon)\simeq\frac{n}{2C_A}\log(1/\epsilon).
  \label{eq:state_design_q}
\end{equation}

The Fokker--Planck derivation above holds for any transfer matrix $T$ of the multi-step birth--death form~\eqref{eq:T_q_local}--\eqref{eq:ab_q_local}, requiring only ergodicity of the chain and not reversibility. Ergodicity is a generic property in this family: it follows whenever the leading-band rates satisfy $a_\nu^{(1)},b_\nu^{(1)}\neq0$, which holds for any $A$ that is not fine-tuned to a measure-zero locus in the space of parity-preserving $q$-qubit unitaries. Consequently, the Fokker--Planck description and the resulting $\Theta(n\log(1/\epsilon))$ convergence to a state $2$-design are \emph{universal} across the entire family of parity-preserving $q$-qubit doping gates: the only $A$-dependent quantity is the diffusion constant $C_A$.

We now examine to what extent the rigorous finite-$n$ statements proved for SWAP doping --- the spectral-gap scaling $\Delta=\Theta(1/n)$ of Sec.~\ref{sec:gap_bounds} and the $\Omega(n\log n)$ mixing-time lower bound of Sec.~\ref{sec:mixing_time} --- carry over to a $q$-qubit doping gate $A$.
\begin{itemize}
    \item \emph{Spectral gap.} Suppose the local transition matrix is reversible against its binomial weight (condition~\eqref{eq:R_dbc}); the chain $T$ is then reversible and admits a well-defined spectral gap, which admits the variational
characterization through the Dirichlet form~\eqref{eq:gap_variational}. In particular, any trial function yields an upper bound
$\Delta(T_A)\le \mathcal E(\phi,\phi)/\Var_{p_{\mathrm H}}(\phi)$.
Following the proof for the SWAP case (Lemma~\ref{lem:gap_upper}), we take the linear test function $\phi(\nu)=\nu$. Since the only transitions
are $\nu\to\nu\pm 8j$, the increments are $\phi(\nu)-\phi(\nu')=\mp 8j$, and the
Dirichlet form collapses to a sum over jump sizes,
\begin{equation}
  \mathcal E(\phi,\phi)\;=\;\frac12\sum_\nu p_{\mathrm H}(\nu)
  \sum_{j=1}^{\lfloor q/2\rfloor}(8j)^2\bigl(a_\nu^{(j)}+b_\nu^{(j)}\bigr),
  \label{eq:dirichlet_phi}
\end{equation}
which, using~\eqref{eq:ab_q_drift_diff} and $\langle\nu^2\rangle_{p_{\mathrm H}}=\Theta(n)$, yields
\begin{equation}
  \mathcal E(\phi,\phi)\;=\;\Bigl(\frac12\sum_{j=1}^{\lfloor q/2\rfloor}(8j)^2 s_j\Bigr)
  \bigl(1+\mathcal{O}((n^{-1})\bigr)
  .
  \label{eq:dirichlet_estimates}
\end{equation}
As $\Var_{p_{\mathrm H}}(\phi)=\langle\nu^2\rangle_{p_{\mathrm H}}\simeq2n$ (up to exponentially small corrections), we obtain
\begin{equation}
\begin{aligned}
  \Delta(T_A)&\;\le\;
  \frac{\tfrac12\sum_{j}(8j)^2 s_j}{\langle\nu^2\rangle_{p_{\mathrm H}}}
  \bigl(1+\mathcal{O}(n^{-1})\bigr)\\
  &\;=\;\frac{C_A}{n}+\mathcal{O}(n^{-2}),
  \label{eq:gap_upper_q_dirichlet}
\end{aligned}
\end{equation}
where the diffusion constant is given in Eq.~\eqref{eq:diff_constant}, in
agreement with the continuum prediction $\Delta_{\mathrm{cont}}=C_A/n$.
    This extends the lower bound on the unitary $2$-design convergence time to
    \begin{equation}
        t_{2\text{-des}}^{\mathrm{unitary}}(\epsilon)\;=\;\Omega\bigl(n\log(1/\epsilon)\bigr)
    \end{equation}
    for arbitrary $q$-qubit gate $A$.
    
    A matching lower bound on $\Delta(T_A)$ is more delicate: the bottleneck-ratio argument of Sec.~\ref{sec:gap_bounds} does not extend straightforwardly beyond $q=2$, where it has been established directly in Sec.~\ref{sec:zz_doping}. We show in App.~\ref{app:gap_lower_q_qubit} that $1-\lambda_1(T_A)=\Omega(1/n)$ for any reversible $A$ outside a measure-zero locus in the space of parity-preserving $q$-qubit unitaries, through comparison with the chain $T_{\rm SWAP}$. While we do not establish the corresponding bound on the minimum eigenvalue, this can be remedied by considering the lazy version of the chain, $T_A'=(I+T_A)/2$\footnote{For $q\geq4$, the chain may fail to have nonzero self-loop probability, making the analysis of the minimum eigenvalue in the non-lazy chain more delicate.}. Under this modification, the upper bound~\eqref{eq:design_upper} extends to give
    \begin{equation}
        t_{2\text{-des}}^{\mathrm{unitary}}(\epsilon)\;=\;O\bigl(n\,(n+\log(1/\epsilon))\bigr).
    \end{equation}

    \item \emph{Mixing-time lower bound.} 
    A direct computation shows that, on the linear function $\phi(\nu)=\nu$,
    \begin{equation} \label{eq:T_phi_q_qubit}
        (T_A\phi)(\nu)\;=\;\beta\,\nu\;+\;\sum_{j=1}^{\lfloor q/2\rfloor-1} \gamma_n^{(j)}\,\nu^{2j+1},
    \end{equation}
    with $\beta=1-C_A/n+\mathcal{O}(n^{-2})$ and $\gamma_n^{(j)}=\mathcal{O}(n^{-2j-1})$.
    The distinguishing-statistic argument of Sec.~\ref{sec:mixing_time} can then be applied using $\phi(\nu)=\nu$ as an approximate eigenfunction, with the explicit decomposition~\eqref{eq:T_phi_q_qubit} providing the higher-order correction terms required by the variance and mean estimate. A careful re-derivation along the lines of App.~\ref{app:mixing_time} then yields the $\Omega(n\log n)$ mixing-time lower bound, and consequently
    \begin{equation}
        t_{2\text{-des}}^{\mathrm{unitary}}(\epsilon)\;=\;\Omega(n\log n).
    \end{equation}
\end{itemize}
The combination establishes, for any $q$-local parity-preserving $A$ satisfying~\eqref{eq:R_dbc}, the leading-order behavior is
\begin{equation}
    t_{2\text{-des}}^{\mathrm{state}}(\epsilon)\;=\;\Theta\bigl(n\log(1/\epsilon)\bigr),
\end{equation}
and
\begin{equation}
    t_{2\text{-des}}^{\mathrm{unitary}}(\epsilon)\;=\;\Omega(n\log\max(n,1/\epsilon)),
\end{equation}
with the upper bound 
\begin{equation}
    t_{2\text{-des}}^{\mathrm{unitary}}=\mathcal{O}(n(n+\log(1/\epsilon))),
\end{equation}
contingent on the laziness assumption.

\section{Low-depth unitary designs from glued doped circuits}
\label{sec:glued}

In this section, we construct a random unitary ensemble that forms an $\epsilon$-approximate unitary $2$-design and achieves improved bounds on both the circuit depth and the SWAP count.

We exploit the glued circuit construction introduced in Ref.~\cite{schuster2024random} (cf also Ref.~\cite{magni2025anticoncentrationstatedesigndoped,grevink2025glueshortdepthdesignsunitary}). In this construction, the $n$ qubits are arranged along a one-dimensional line and partitioned into $m$ local patches, each containing $\xi = n/m$ qubits. The resulting random unitary ensemble corresponds to a two-layer circuit in which each small random unitary acts on two neighboring patches, with the unitaries arranged in a brickwork pattern across the two layers. If each small random unitary is drawn from an $\epsilon/n$-approximate unitary $k$-design on $2\xi$ qubits with circuit depth $d$, then the overall construction forms an $\epsilon$-approximate unitary $k$-design on $n$ qubits with total depth $2d$, provided that $\xi \geq \log_2(nk^2/\epsilon)$. While Ref.~\cite{schuster2024random} focuses on the non-symmetric case, we show in
App.~\ref{app:pp_glued} that the construction extends to parity-preserving systems, where
we focus particularly on $k=2$.

\begin{figure}[h!]
    \centering
    \includegraphics[width=1\linewidth]{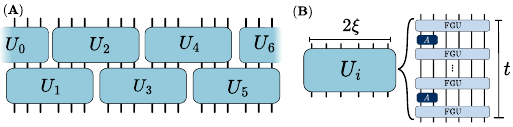}
    \caption{(\textbf{A}) Schematic of the glued-circuit construction. The \(n\) qubits are partitioned into local patches of size \(\xi\), and overlapping unitaries acting on neighboring pairs of patches (\(2\xi\) qubits) are arranged in a brickwork geometry. (\textbf{B}) Each local unitary in the glued circuit is implemented by a \(t\)-doped matchgate circuit acting on \(2\xi\) qubits.}
    \label{fig:Glued}
\end{figure}

Using this framework, we construct a circuit by replacing each small random unitary on $2\xi$ qubits with a doped matchgate circuit. Setting $\xi = \log_{2}(4n/\epsilon)$ and applying Eq.~\eqref{eq:2_design_bound}, each local circuit forms an $\epsilon/n$-approximate unitary $2$-design using $\mathcal{O}(\log^2(n/\epsilon))$ SWAP gates. Since the construction contains $\Theta(n/\xi)$ such local unitaries, the overall circuit yields an $\epsilon$-approximate unitary $2$-design using a total of 
\begin{equation}
    t=\mathcal{O}(n\log(n/\epsilon))
\end{equation}
 SWAP gates.

Moreover, since matchgate circuits on $2\xi$ qubits can be implemented in 1D with depth $\mathcal{O}(\xi)$~\cite{jiang2018quantum,Kivlichan2018,langer2026matchgate,morralyepes2026disentangling}, the glued construction yields an $\epsilon$-approximate unitary $2$-design with circuit depth
\begin{equation}
    d = \mathcal{O}(\log^3(n/\epsilon)),
\end{equation}
that is, with polylogarithmic depth. Should Conjecture~\ref{conj:design_time}
hold, this would improve the depth bound to $d = \mathcal{O}(\log^2(n/\epsilon))$.

\section{Applications}
\label{sec:Applications}

\subsection{Fermionic classical shadows}
\label{sec:classical_shadow}
In practice, approximate unitary $k$-designs can often replace Haar-random unitaries—which are typically very costly to implement—in a wide range of protocols without significantly affecting their performance or runtime~\cite{schuster2024random}. As a concrete example, we examine in more detail the application in classical shadows~\cite{huang2020predicting}. In a classical shadow protocol, one samples a random unitary from an ensemble $\mathcal{E}$, applies it to an unknown state $\rho$, and measures in the computational basis to obtain an outcome $\ket{b}$. The corresponding measurement channel is
\begin{equation}
\begin{aligned}
    \mathcal{M} (\rho)
    &= \mathbb{E}_{ U\sim\mathcal{E}}\sum_b
    \bra{b}U\rho U^\dagger\ket{b}\,
    U^\dagger\ketbra{b}{b}U \\
    &=
    \Tr_1\!\left(
    \sum_{b\in\{0,1\}^{n}}
    \mathbb{E}_{ U\sim\mathcal{E}}
    \!\left[
    (U^{\dagger})^{\otimes 2}
    \ketbra{b}{b}^{\otimes 2}
    U^{\otimes 2}
    \right]
    (\rho\otimes I)
    \right),
\end{aligned}
\end{equation}
and the resulting classical shadow is
\begin{equation}
\hat\rho
=
\mathcal{M}^{-1}
\!\left(
\hat{U}^\dagger
\ketbra{\hat{b}}{\hat{b}}
\hat{U}
\right),
\end{equation}
which provides an unbiased estimator for $\rho$. These samples can then be used to estimate expectation values of observables of interest, since
\begin{equation}
    \Tr(O\rho) = \mathbb{E}_{ U\sim\mathcal{E}}\sum_b
    \bra{b}U\rho U^\dagger\ket{b}\,\Tr\left(O\mathcal{M}^{-1}\left(
    U^\dagger\ketbra{b}{b}U\right)\right).
\end{equation}

The classical shadow framework extends naturally to the setting of fermionic quantum computation. In particular, for Haar-random parity-preserving unitaries, the measurement channel can be obtained straightforwardly from standard Schur--Weyl duality restricted to the parity-preserving Hilbert space, yielding
\begin{equation}
    \mathcal{M}(X)
    =
    \frac{
    X+\Tr(X)I/2+\Tr(XP)P/2
    }{2^{n-1}+1},
\end{equation}
where $P\coloneqq\prod_{j=1}^{n}Z_j$ is the parity operator. The inverse map is
\begin{equation}
\label{eq:inverse_map}
    \mathcal{M}^{-1}(X)
    =
    (2^{n-1}+1)X
    -
    \frac{\Tr(X)I-\Tr(XP)P}{2}.
\end{equation}

To bound the number of classical shadow samples $\hat\rho$ (and hence the number of copies of $\rho$) required to estimate expectation values to a desired error with high probability, one studies the variance of the estimator, which depends on both the unitary ensemble and the observable being measured. The variance reads~\cite{huang2022learning}:
\begin{equation}
\mathrm{Var}[\hat o]
\le
\mathbb{E}_{U,b}\!\left[
\bigl(
\Tr[
O\,\mathcal{M}^{-1}(U^\dagger\ketbra{b}{b}U)
]
\bigr)^2
\right],
\end{equation}
which is controlled by the third moment of the ensemble. For Haar-random unitaries, this variance is bounded by the Hilbert--Schmidt norm $\Tr(O^2)$~\cite{huang2020predicting}. For example, the fidelity $\bra{\psi}\rho\ket{\psi}$ with an arbitrary pure state $\ket{\psi}$ can be estimated efficiently, requiring only a polynomial number of samples within the classical shadow protocol using Haar-random unitaries. This task is particularly relevant for certifying that an experimental device prepares a desired target state. In this sense, the protocol is substantially more powerful than the matchgate shadow protocol~\cite{wan2023matchgate}, which only allows efficient estimation of fidelities with pure fermionic Gaussian states. More generally, any observable with bounded Hilbert--Schmidt norm can be estimated efficiently using classical shadows.

The generation of unitary $3$-designs is therefore crucial to control the performance of classical shadows. In particular, if one uses an approximate $3$-design with relative error $\epsilon$, together with the inverse map in Eq.~\eqref{eq:inverse_map} (which only approximately inverts the twirl), then the estimator of $\Tr(\rho O)$ acquires a bias of at most $2\epsilon \Tr(O)$ and an additional variance bounded by $\propto\epsilon \Tr(O)^2$ for any positive observable $O$~\cite{schuster2024random}. For observables with bounded Hilbert--Schmidt norm, this additional variance is subleading compared to the variance bound obtained from Haar-random unitaries.

Our results demonstrate that doped matchgate circuits can be used in classical shadow protocols to mimic Haar-random parity-preserving unitaries. In particular, our glued circuit construction (Sec.~\ref{sec:glued}) provides a low-depth construction and an explicit bound of $t=\mathcal{O}(n\log(n/\epsilon))$ on the number of SWAP gates required to achieve a desired bias. Furthermore, our numerical results for the $3$-design suggest that the same number of SWAP gates suffices to ensure only a small additional variance bound. Remarkably, the circuit requires only polylogarithmic depth, making it particularly useful in near-term implementation.

Moreover, since Clifford--matchgate circuits form an exact matchgate $3$-design~\cite{wan2023matchgate,sierant2026theory}, one may replace the random matchgate unitary with a random Clifford--matchgate unitary. In this case, the entire circuit becomes a Clifford circuit. Since Clifford circuits admit efficient classical simulation, the classical post-processing stage of the protocol can likewise be implemented efficiently.

\subsection{Emergence of generic entanglement}
\label{sec:entanglement}

The convergence to a unitary $2$-design translates into a bound on the average entanglement of states produced by the doped ensemble, with the strongest implication appearing in the regime $t=\Theta(n^{2})$. Since the upper bound~\eqref{eq:design_upper} scales as $t = \mathcal{O}(n(n+\log(1/\epsilon)))$, taking $\epsilon = 2^{-cn}$ for any fixed $c>0$ still yields $t=\mathcal{O}(n^{2})$; the lower bound~\eqref{eq:design_lower_gap} matches this at $\Omega(n^{2})$ for the same $\epsilon$. Hence $\Theta(n^{2})$ doping rounds are necessary and sufficient to achieve a relative-error $2$-design with \emph{exponentially} small error in $n$. As we now show, this exponential decay is precisely what drives the average entanglement entropy to its Haar/Page value.

Fix a bipartition $\mathcal H = \mathcal H_A\otimes\mathcal H_B$ with $d_A = 2^{n_A}$, $d_B = 2^{n_B}$, $n_A+n_B=n$, and $n_A\leq n_B$. For initial state $|\psi_0\rangle = |0\rangle^{\otimes n}$ and reduced density matrix $\rho_A(U):=\Tr_B(U|\psi_0\rangle\langle\psi_0|U^\dagger)$, the average purity is
\begin{equation}
    \mathbb E_{U\sim\mathcal E_t}\bigl[\Tr\rho_A(U)^{2}\bigr]
    \;=\; \Tr\bigl[(\mathbb{F}_{A}\otimes I_{B})\,\Phi_t\bigl(|\psi_0\rangle\langle\psi_0|^{\otimes 2}\bigr)\bigr],
    \label{eq:purity_def}
\end{equation}
where $\mathbb{F}_A$ exchanges the two copies of subsystem $A$. 

Recall that the relative error implies an additive-error bound, and the diamond-norm bound directly controls the difference between the $t$-doped ensemble and the Haar ensemble: applying $\Phi_t - \Phi_{\mathrm H}^{(2)}$ to the input $|\psi_0\rangle\langle\psi_0|^{\otimes 2}$ yields $\|(\Phi_t - \Phi_{\mathrm H}^{(2)})(|\psi_0\rangle\langle\psi_0|^{\otimes 2})\|_1 \leq \epsilon$, and tracing against $\mathbb{F}_A\otimes I_{B}$, gives by H\"older's inequality
\begin{equation}
    \bigl|\mathbb E_{\mathcal E_t}[\Tr\rho_A^{2}] - \mathbb E_{\mathrm{H}}[\Tr\rho_A^{2}]\bigr|
    \;\leq\;\bigl\|\mathbb{F}_A\otimes I_{B}\bigr\|_\infty\epsilon\;=\;\epsilon.
    \label{eq:purity_bound}
\end{equation}
The Haar value is given by the standard formula,
\begin{equation}
    \mathbb E_{\mathrm{H}}[\Tr\rho_A(U)^{2}] \;=\; \frac{d_A+d_B}{d_{AB} +1}\;\simeq\;2^{-n_A},
    \label{eq:haar_purity}
\end{equation}
for $0<n_A/n<1/2$. 

This purity bound translates into an entanglement-entropy bound through the R\'enyi-2 entropy $S_2(\rho_A):=-\log\Tr\rho_A^{2}$, which lower bounds the von Neumann entropy: $S(\rho_A)\geq S_2(\rho_A)$. Applying Jensen's inequality to the concave function $f(x)=-\log x$,
\begin{equation}
\begin{split}
    \mathbb{E}_{\mathcal E_t}[S(\rho_A)]
    &\geq \mathbb{E}_{\mathcal E_t}[S_2(\rho_A)] \geq \\
    & -\log \mathbb{E}_{\mathcal E_t}[\Tr \rho_A^{2}]
    \gtrsim -\log\!\bigl(2^{-n_A} + \epsilon\bigr).
    \label{eq:entanglement_bound}
\end{split}
\end{equation}
The right-hand side is informative only when $\epsilon$ is small compared to $2^{-n_A}$, i.e., when $\epsilon$ is exponentially small in the smaller subsystem size; in this regime, $-\log(2^{-n_A}+\epsilon)\approx n_A$, recovering Page's value in the leading order. Choosing $\epsilon = 2^{-(1/2+\eta)n}$ with $\eta>0$ ensures $\epsilon \ll 2^{-n_A}$ uniformly across all bipartitions. Combining,
\begin{equation}
    \mathbb E_{\mathcal E_t}[S(\rho_A)] \;  \simeq\; n_A   \qquad\text{for }t=\Theta(n^{2})
    \label{eq:entanglement_main}
\end{equation}
with an exponentially small correction, which matches the Haar average to leading order. For the balanced cut $n_A=n_B=n/2$, this gives $\mathbb E_{\mathcal E_t}[S(\rho_A)]\simeq n/2 - 1$, again with exponentially small correction.

These results are consistent with the numerical simulations of Ref.~\cite{paviglianiti2026emergence}, which observed that the deviation
$\mathbb E_{\mathrm H}[S(\rho_A)]-\mathbb E_{\mathcal E_t}[S(\rho_A)] \sim e^{-C n_{\mathrm{NG}}/n}$ for some constant $C>0$, where $n_{\mathrm{NG}}$ denotes the number of non-Gaussian gates injected into the circuit. In particular, this scaling suggests that the error becomes exponentially small once $n_{\mathrm{NG}}=\Theta(n^{2})$, in agreement with the analytical results above.

\section{Conclusion: summary and outlook}
\label{sec:conclusions}

We have introduced doped matchgate circuits as a minimal, analytically
tractable framework for studying how Haar-like randomness emerges from
non-Gaussian fermionic dynamics. Starting from a classically simulable ensemble
of fermionic Gaussian unitaries, we showed that repeatedly inserting local
non-Gaussian gates drives the dynamics toward unitary-design behaviour, and that
this buildup rests on a single structural fact: averaging over the random
matchgate layers projects the replicated dynamics onto the low-dimensional
matchgate commutant, so that the only role of a dopant is to redistribute weight
among a few commutant sectors. For the two-copy global protocol this
reduction is exact and turns design formation into a classical birth--death
process on the spectrum of the bridge operator, whose stationary state is the
Haar distribution in the relevant parity sector and whose continuum limit is
an Ornstein--Uhlenbeck process. The emergence of randomness is thereby recast in
the familiar language of drift, diffusion, spectral gaps, and mixing times.

This stochastic picture yields rigorous bounds on the formation of approximate
unitary $2$-designs. The same continuum description, and with
it the same scaling, extends to a broad family of parity-preserving $q$-local
dopants, with the microscopic gate entering only through an effective diffusion
constant; the buildup of randomness is therefore a robust consequence of
injecting non-Gaussianity into Gaussian dynamics rather than an artifact of a
fine-tuned gate.

These two-copy bounds also serve as building blocks for concrete constructions.
We assembled doped blocks into a glued circuit that forms an approximate unitary
$2$-design in polylogarithmic depth using only a sparse set of non-Gaussian
gates, and we showed that doped matchgate circuits can replace Haar-random
parity-preserving unitaries in fermionic classical-shadow protocols and drive
the average entanglement entropy to its Page value. Beyond the rigorously
controlled two-copy setting, the same commutant approach extends to three
copies and reaches large system sizes numerically: the unitary and state
$3$-frame potentials relax on the same $t/n$ scale as their two-copy
counterparts, and the three-copy spectral gap obeys the same $\mathcal{O}(1/n)$ law. While
these data do not amount to a rigorous higher-design theorem, they indicate that
the same stochastic mechanism governs design formation beyond $k=2$. In particular, both the state frame potentials for $k=3$ and $k=4$ are consistent with formation on the same $\Theta(n)$ timescale as the $2$-design, suggesting a uniform scaling of design emergence across higher-order designs.

A central lesson is that the emergence of randomness is controlled not only by
how much non-Gaussianity is injected, but by how it is distributed in space. In
the globally scrambled protocol, each doping event is immediately spread across
the system by the following random Gaussian layer. When the dopant instead acts
repeatedly on a fixed bond of a brickwork circuit, convergence becomes
transport-limited and the frame potentials relax on a diffusive $t/n^2$
scale. Global and local doping therefore realize qualitatively distinct routes
from integrable Gaussian dynamics to Haar-like randomness.

Several directions follow. The most immediate is to close the gap between our
lower and upper bounds on unitary $2$-design formation: since the two-copy
problem is an explicit birth--death chain, a sharper estimate of its
relative-error mixing time would pin down the optimal design depth. As noted above, this can potentially be achieved using the log-Sobolev inequality~\cite{Diaconis1996}. A second is
to make the higher-copy theory rigorous; the analytic three-copy construction
and the small-system data at $k=4$ suggest that the mechanism persists, but a
complete treatment will require a more complete understanding of the higher-copy
commutant. A third is a systematic theory of spatially structured doping, which
would classify how the density, geometry, and randomness of non-Gaussian gates
set the scaling of design formation, with probabilistic and spatially
distributed protocols interpolating between the globally mixed and
transport-limited regimes.

A closely related question concerns the optimal scaling of non-Gaussian resources for state-design formation. Ref.~\cite{poetri26fermionic} establishes a lower bound of $\Omega(\sqrt{n})$ non-Gaussian gates, which applies even in the presence of arbitrary Gaussian protocols, including measurements and feedforward. Our construction using unitary gates achieves approximate state $2$-designs using $\Theta(n)$ non-Gaussian gates. Whether this scaling can be improved, for instance by exploiting measurement-assisted protocols or more refined adaptive Gaussian strategies, remains an open problem.

A further avenue concerns fermionic non-Gaussianity as a resource in its own
right~\cite{Sierant26faf,Gottlieb05, dias2024classical,reardonsmith2024improved, Cudby24gaussian, hebenstreit2019all, Lumia24, Lyu24NGE, Coffman25magic, walter2025random,Deneris26, ares2026non, ares2026asymmetry, debertolis2025naturalsuperorbitalsrepresentationmanybody,Turner17free, pachos2018quantifying, pachos2022quantifying, meichanetzidis2018free}. Recent progress has made it experimentally accessible through
covariance-based witnesses and two-copy Bell-sampling protocols that test how
far a state lies from the Gaussian manifold~\cite{Bittel24optimal, poetri26fermionic,haug2026w}
and quantify its non-Gaussianity~\cite{poetri26fermionic}, the latter built on
the very two-copy bridge operator that governs our dynamics. Connecting the
dynamical buildup of non-Gaussianity under doping to its direct measurement is a
natural next step, as is the extension to mixed states, where non-Gaussianity is
considerably more subtle to define and detect. A promising route for the latter is to develop quantitative witnesses~\cite{eisert2007quantitative,haug2025efficientwitnessingtestingmagic,tarabunga2025quantifying}, to lower bound mixed-state non-Gaussianity.

Together, these results establish doped matchgate dynamics as a bridge between free-fermion integrability and quantum randomness, in which global Gaussian scrambling turns the buildup of Haar-like behaviour into an emergent classical stochastic process, and provide a concrete framework for constructing and analyzing low-depth designs from minimally non-Gaussian circuits.
More broadly, our work contributes to the growing effort to understand how quantum resources~\cite{chitambar2019quantumresourcetheories,Gour2025} behave in many-body systems. While entanglement~\cite{amico2008entanglement, horodecki2009quantum,vidal2000entanglement} has long played a central role in the characterization of equilibrium and nonequilibrium quantum matter~\cite{calabrese2009entanglement, Eisert10area, Laflorencie16pr,
Lauchli08, Kim13ballistic, DeChiara06, Znidaric08, Kiefer20, Aceituno24, nahum2017quantum, zhou2020entanglementmembrane, sierant2023membrane}, other resources, most notably non-stabilizerness and fermionic non-Gaussianity, have recently emerged as diagnostics of many-body complexity~\cite{Odavic23, Bejan24, Fux24separation, Tirrito2025MagicGauss, Aditya25Mbempa, aditya2026coher, collura2026nonlocalnonstabilizernessfreefermion,iannotti2026nonlocalmagicresourcesfermionic,aditya2026coherencedynamicsquantummanybody,aditya2026higherordersymmetricquantummpemba, tirrito2025univ,rbt4-psfd,dowling2026noiseinducedsimulabilitytransitionoperator,Hoshino25SREcft, Trigueros2025NoisyMagic, turkeshi2024magic, Tirritoanticoncentration2025,  Jasser2025,  Bera2025SYK, Odavic2025, grabarits2026uni, nehra2025topol, Tarabunga2024crit, Falcao25mbl, sierant202phasetransitions, xiao2026di, liu2026di, khasseh2026hidd, xiao2026non, iannotti2026nonstab,tarabunga2024mps,Tarabunga24rk,Tarabunga24transition,timsina2025robustness,tarabunga2025efficientmutualmagic,c7k1-xcwy,dowling2026pagecurvelocaloperatorentanglement,p7xt-s9nz,paviglianiti2026truecostfactoringlinking,capecci2025rolenonstabilizernessquantumoptimization,capecci2026quantumresourcesnonstoquasticquantum,Varikuti2026impactofclifford,varikuti2025deepthermalizationmeasurementsquantum,santra2025quantumresourcesnonabelianlattice,dowling2026noiseinducedsimulabilitytransitionoperator,xiao2026diffusivedynamicsnonstabilizerness,nzrp-49mr,557f-6tpb,96bk-xf8p,wang2025magictransitionmonitoredfree,liu2026diffusiverelaxationparticipationentropy,PhysRevLett.133.190402,PhysRevLett.134.150404,gljy-1ykf,lwwp-6rqk,Huang2026,7tdt-kx4l,PhysRevLett.133.010602,PhysRevLett.133.150604,PhysRevLett.134.150403,10.21468/SciPostPhys.18.5.165,msm2-vmg7,31sq-k4m3,frau2024nonstabilizernessversusentanglementmatrix,tirrito2024quantifying,tarabunga2024nonstabilizernessasymmetry}. Our results provide a reference point for understanding the generation and spreading of non-Gaussianity in dynamical settings, including both ergodic~\cite{Sierant26faf,ares2026non} and non-ergodic dynamics~\cite{Falcao26mbl,bera2015mbl}, and connect to recent studies of magic and related resources in ground states of lattice gauge theories~\cite{Santra25complexity,Falcao25,tarabunga23gauge} and condensed-matter systems~\cite{zavatti26magic,zavatti26corr}, as well as in broader contexts ranging from nuclear to particle physics~\cite{robin2026complexity}.

\textbf{Acknowledgements.---}
P.S.T thanks B. Kraus, F. Pollmann, S.H. Lin, B. Jobst, R. Morral-Yepes, M. Langer, and S. Fraenkel, for discussions and collaborations on related topics. 
F.B.T has received funding from the European Research Council (ERC) under the European Union’s Horizon 2020 research and innovation programme (grant agreement No. 853443). F.B.T gratefully acknowledges the resources on the LiCCA HPC cluster of the University of Augsburg, co-funded by the Deutsche Forschungsgemeinschaft (DFG, German Research Foundation)-Project-ID 499211671.
X.T. acknowledges support from DFG under Germany's Excellence Strategy – Cluster of Excellence Matter and Light for Quantum Computing (ML4Q) EXC 2004/2 – 390534769, and DFG Collaborative Research Center (CRC) 183 Project No. 277101999 - project B01, and DFG Emmy Noether Programme proposal ``Digital Quantum Matter Ouf-of-Equilibrium'' No. 560726973. 
P.S. acknowledges fellowship within the “Generación D” initiative, Red.es, Ministerio para la Transformación Digital y de la Función Pública, for talent attraction (C005/24-ED CV1), funded by the European Union NextGenerationEU funds, through PRTR. P.S.T. acknowledges funding from the European Research Council (ERC) under the European Union (ERC, DynaQuant, No. 101169765).

\textit{Note Added.}
While finalizing this manuscript, we became aware of independent forthcoming work by Leone and Bittel~\cite{leonida}, which establishes the optimality of the state-design construction considered in our work.

\bibliography{bib.bib}

\clearpage

\appendix
\section{3-copy doping matrix}
\label{app:3copy}

In this appendix, we provide a detailed derivation of the doping matrix for the three-copy matchgate twirling channel. We begin by introducing the structure of the three-copy twirl and its natural basis.

\subsection{Three-copy basis and combinatorial structure}

We consider $n$ qubits with $N = 2n$ Majorana operators $\{\gamma_\mu\}_{\mu=1}^N$ satisfying
\begin{equation}
\{\gamma_\mu, \gamma_\nu\} = 2 \delta_{\mu,\nu} \,\boldsymbol{1}.
\end{equation}
For any subset $S \subseteq \{1,\dots,N\}$, we define the Majorana monomial
\begin{equation}
\gamma_S := \prod_{\mu \in S} \gamma_\mu.
\end{equation}

We work in normalized Liouville space, defining
\begin{equation}
\big|\gamma_S\!\rrangle := \frac{\gamma_S}{\sqrt{2^n}},
\end{equation}
so that
\begin{equation}
\llangle \gamma_S \,|\, \gamma_{S'} \rrangle
=
\frac{1}{D} \operatorname{tr}(\gamma_S \gamma_{S'})
=
\delta_{S,S'},
\qquad D = 2^n.
\end{equation}

At three copies, the natural basis vectors are labeled by triples
\begin{equation}
\boldsymbol{k} = (k_1, k_2, k_3),
\qquad
k_i \ge 0,
\qquad
k_1 + k_2 + k_3 \le N.
\end{equation}
We also define the complementary occupation number
\begin{equation}
k_0 := N - k_1 - k_2 - k_3,
\end{equation}
so that $(k_1,k_2,k_3,k_0)$ forms a four-way composition of $N$.

The vector $\boldsymbol{k}$ should be understood as encoding a partition of the $N$ Majorana indices into four \textbf{disjoint} bins $A_1,\, A_2,\, A_3,\, A_0,$ with $|A_i| = k_i, \, i=0,1,2,3.$

We define the multinomial normalization
\begin{equation}
\mathcal N(\boldsymbol{k})
=
\binom{N}{k_1,k_2,k_3,k_0}^{-1/2}.
\end{equation}

For a given partition $(A_1,A_2,A_3,A_0)$, the three canonical replica operators are
\begin{equation}
O_{12} := \gamma_{A_1}\gamma_{A_2}, 
\qquad
O_{23} := \gamma_{A_2}\gamma_{A_3}, 
\qquad
O_{31} := \gamma_{A_3}\gamma_{A_1}.
\end{equation}

The normalized three-copy basis vector is then defined as the uniform superposition over all such partitions:
\begin{equation}
\big|\Upsilon^{(3)}_{\boldsymbol{k}}\big\rrangle
:=
\mathcal N(\boldsymbol{k})
\sum_{\substack{A_1,A_2,A_3 \subseteq [N] \\ |A_i| = k_i \\ \text{disjoint}}}
\big|\gamma_{A_1}\gamma_{A_2}\!\rrangle
\big|\gamma_{A_2}\gamma_{A_3}\!\rrangle
\big|\gamma_{A_3}\gamma_{A_1}\!\rrangle.
\end{equation}

The three-copy matchgate twirl is defined as
\begin{equation}
\mathcal E^{(3)}(A)
=
\int_{O(2N)} d\mu(Q)\, \,
U_Q^{\otimes 3} A U_Q^{\dagger \otimes 3}.
\end{equation}

As shown in Ref.~\cite{wan2023matchgate}, this channel acts as an orthogonal projector onto the invariant subspace spanned by the $\{|\Upsilon^{(3)}_{\boldsymbol{k}}\rrangle\}$ basis:
\begin{equation}
\mathcal E^{(3)}
=
\sum_{\boldsymbol{k}}
\big|\Upsilon^{(3)}_{\boldsymbol{k}}\big\rrangle
\big\llangle \Upsilon^{(3)}_{\boldsymbol{k}}\big|.
\end{equation}

Equivalently, for any operator $A$,
\begin{equation}
\mathcal E^{(3)}(A)
=
\sum_{\boldsymbol{k}}
a_{\boldsymbol{k}}
\big|\Upsilon^{(3)}_{\boldsymbol{k}}\big\rrangle,
\qquad
a_{\boldsymbol{k}} = \llangle \Upsilon^{(3)}_{\boldsymbol{k}} | A \rrangle.
\end{equation}

This representation makes explicit that the three-copy problem factorizes into a local computation on subsets of Majoranas and a global multinomial counting problem over partitions of $N$, a structure that is central in deriving the doping matrix.

\subsection{Non-Gaussian doping gate}

We now introduce the non-Gaussian gate that will generate nontrivial mixing between the three-copy shells. We consider a quartic, parity-preserving interaction corresponding to a $Z_1$--$Z_2$ coupling,
\begin{equation}
A := e^{i \frac{\pi}{4} Z_1 Z_2}.
\end{equation}
In the Majorana representation, this operator can be written as
\begin{equation}
Z_1 Z_2 = -\,\gamma_1 \gamma_2 \gamma_3 \gamma_4 \equiv -\Gamma_T,
\end{equation}
so that
\begin{equation}
A = \frac{1 - i \Gamma_T}{\sqrt{2}} = e^{-i \frac{\pi}{4} \Gamma_T}.
\end{equation}

Importantly, this gate acts nontrivially only on the four Majorana operators $T = \{1,2,3,4\}$, and leaves all other Majoranas invariant (up to fermionic reordering signs). Thus, its action on any Majorana monomial $\gamma_S$ depends only on the local overlap
\begin{equation}
r := |S \cap T|.
\end{equation}

This locality property implies that the full $N$-Majorana problem factorizes into local signed computation on the active set $T$, and a global combinatorial counting problem over the remaining $N-4$ inactive Majoranas.
This separation allows us to compute the three-copy doping matrix analytically.

We study the induced channel on three replicas,
\begin{equation}
\mathcal C^{(3)}_A : X \otimes X \otimes X \;\mapsto\; (A X A^\dagger)^{\otimes 3},
\end{equation}
and its matrix elements in the $\{|\Upsilon^{(3)}_{\boldsymbol{k}}\rrangle\}$ basis,
\begin{equation}
C_{\boldsymbol{\ell},\boldsymbol{k}}
:=
\llangle \Upsilon^{(3)}_{\boldsymbol{\ell}} \big| \mathcal C^{(3)}_A \big| \Upsilon^{(3)}_{\boldsymbol{k}} \rrangle.
\end{equation}

Finally, we note that this gate is equivalent to the SWAP gate up to a Gaussian unitary, since Gaussian conjugations act within the invariant subspace spanned by the $\{|\Upsilon^{(3)}_{\boldsymbol{k}}\rrangle\}$ basis and do not induce shell mixing, the effective transfer matrix describing SWAP doping coincides with that of the quartic gate $A$.

\subsection{Reduction to a finite combinatorial problem}

We now exploit the locality of the dopant to reduce the computation of the matrix elements $C_{\boldsymbol{\ell},\boldsymbol{k}}$ to a finite combinatorial problem.

Recall that the basis vector $|\Upsilon^{(3)}_{\boldsymbol{k}}\rrangle$ is a uniform superposition over all partitions of the $N$ Majorana indices into four disjoint bins $(A_1,A_2,A_3,A_0)$ with fixed cardinalities $(k_1,k_2,k_3,k_0)$.

Since the gate $A$ acts nontrivially only on the four Majoranas in $T=\{1,2,3,4\}$, all remaining $N-4$ Majoranas act as spectators. As a consequence, the full sum over partitions factorizes into: a sum over how the four active Majoranas are distributed among the bins, a multinomial counting over the $N-4$ inactive Majoranas, and a local signed contribution from the active block.

Conceptually, the problem reduces to: (i) choosing how the active Majoranas are distributed, (ii) evaluating a finite local signed contribution, and (iii) counting spectator configurations multinomially.We now make this factorization explicit.

\subsubsection{Active and inactive decomposition}

Fix an input sector $\boldsymbol{k}=(k_1,k_2,k_3)$. We first parametrize how the four active Majoranas are distributed among the bins by defining
\begin{equation}
t_i := |A_i \cap T|,
\qquad i=1,2,3,
\end{equation}
and
\begin{equation}
t_0 := 4 - t_1 - t_2 - t_3.
\end{equation}
Thus $(t_1,t_2,t_3,t_0)$ is a composition of $4$.

Once the active occupancies are fixed, the inactive $N-4$ Majoranas must fill the bins with occupancies
\begin{equation}
a_i := k_i - t_i,
\qquad i=1,2,3,
\end{equation}
and
\begin{equation}
a_0 := k_0 - t_0,
\qquad
a_0 = (N-4) - a_1 - a_2 - a_3.
\end{equation}

\subsubsection{Inactive multinomial factor}

Fix $(t_1,t_2,t_3)$ and define $(a_1,a_2,a_3,a_0)$ as above. The number of ways to distribute the inactive $N-4$ Majoranas among the four bins is
\begin{equation}
\binom{N-4}{a_1,a_2,a_3,a_0}
=
\frac{(N-4)!}{a_1!a_2!a_3!a_0!}.
\label{eq:outer-multinomial}
\end{equation}

This factor contains all dependence on the spectator Majoranas. The remaining nontrivial contribution arises entirely from the active block $T$.

\subsubsection{Active block and parity structure}

Fix a particular placement of the active block $T$ into the bins, i.e., subsets
\begin{equation}
T = A_1 \sqcup A_2 \sqcup A_3 \sqcup A_0,
\qquad
|A_i| = t_i.
\end{equation}

The associated replica monomials are
\begin{equation}
O_{12} = \gamma_{A_1}\gamma_{A_2},
\, \,
O_{23} = \gamma_{A_2}\gamma_{A_3},
\, \,
O_{31} = \gamma_{A_3}\gamma_{A_1}.
\end{equation}

Their degrees are
\begin{equation}
|O_{12}| = t_1 + t_2,\quad
|O_{23}| = t_2 + t_3,\quad
|O_{31}| = t_3 + t_1.
\end{equation}

Since $\Gamma_T^2 = \boldsymbol{1}$, the action of the gate
\(
A = e^{-i\frac{\pi}{4}\Gamma_T}
\)
by conjugation on Majorana monomials supported on $T$ is particularly simple.

Let $\gamma_S$ be a Majorana monomial and denote by
\(
|S \cap T|
\)
the number of active Majoranas it contains. Then one finds
\begin{equation}
A\,\gamma_S\,A^\dagger
=
\begin{cases}
\gamma_S, & |S \cap T|\ \text{even},\\[2pt]
i(-\Gamma_T)\,\gamma_S, & |S \cap T|\ \text{odd}.
\end{cases}
\label{eq:conj-rule}
\end{equation}

That is, monomials with even support on the active block are left invariant, monomials with odd support acquire an additional factor $i(-\Gamma_T)$. Thus, the action of the dopant on a Majorana monomial depends only on the parity of its overlap with the active set $T$.

Applying this rule to the replica monomials $O_{12}, O_{23}, O_{31}$, we find that their transformation is fully determined by the three parities
\begin{equation}
\begin{aligned}
e_{12} &:= (t_1+t_2)\bmod 2, \\
e_{23} &:= (t_2+t_3)\bmod 2, \\
e_{31} &:= (t_3+t_1)\bmod 2.
\end{aligned}
\end{equation}

Because $e_{12}+e_{23}+e_{31}=0 \pmod{2}$, only four parity patterns are possible:
\begin{equation}
(0,0,0),\quad (1,1,0),\quad (0,1,1),\quad (1,0,1).
\end{equation}

Remarkably, the action of the quartic gate induces a parity-controlled permutation of the bins $(A_0,A_1,A_2,A_3)$.

We now determine the induced transformation of the active occupancies.

\paragraph{Case 1: $(0,0,0)$.}
All monomials are even:
\begin{equation}
(t_1',t_2',t_3') = (t_1,t_2,t_3).
\end{equation}

\paragraph{Case 2: $(1,1,0)$.}
\begin{equation}
\begin{aligned}
(A_1,A_2,A_3,A_0) &\to (A_3,A_0,A_1,A_2), \\
(t_1',t_2',t_3') &= (t_3,t_0,t_1).
\end{aligned}
\end{equation}

\paragraph{Case 3: $(0,1,1)$.}
\begin{equation}
\begin{aligned}
(A_1,A_2,A_3,A_0) &\to (A_2,A_1,A_0,A_3), \\
(t_1',t_2',t_3') &= (t_2,t_1,t_0).
\end{aligned}
\end{equation}

\paragraph{Case 4: $(1,0,1)$.}
\begin{equation}
\begin{aligned}
(A_1,A_2,A_3,A_0) &\to (A_0,A_3,A_2,A_1), \\
(t_1',t_2',t_3') &= (t_0,t_3,t_2).
\end{aligned}
\end{equation}

Thus each active configuration induces a deterministic map
\begin{equation}
\Phi_{\boldsymbol{t}}(\boldsymbol{k})
=
(k_1 - t_1 + t_1',\; k_2 - t_2 + t_2',\; k_3 - t_3 + t_3').
\end{equation}

\subsubsection{Fermionic sign structure}

Fix an active configuration $(A_1,A_2,A_3,A_0)$ of the four Majoranas in $T$.

Under conjugation, the tensor product
\begin{equation}
O_{12} \otimes O_{23} \otimes O_{31}
\end{equation}
is mapped to a new tensor product which must be reordered into the canonical representative associated with the transformed bins.

This reordering produces a fermionic sign
\begin{equation}
w(A_1,A_2,A_3) \in \{\pm 1\},
\end{equation}
arising solely from the antisymmetric nature of Majorana operators.

For fixed occupancies $(t_1,t_2,t_3)$, different placements of the four active Majoranas lead to different signs. The net contribution is therefore obtained by summing over all such configurations:
\begin{equation}
S(t_1,t_2,t_3)
:=
\sum_{\substack{
A_0\sqcup A_1\sqcup A_2\sqcup A_3 = T \\
|A_i| = t_i
}}
w(A_1,A_2,A_3).
\end{equation}

Thus, $S(t_1,t_2,t_3)$ encodes the signed multiplicity of active configurations with occupancies $(t_1,t_2,t_3,t_0)$. Since $|T|=4$, it is a finite function that only needs to be computed once.

\subsubsection{Assembly of the matrix element}

Combining the local and global contributions, each active configuration $\boldsymbol{t}=(t_1,t_2,t_3)$ contributes: a multinomial weight from the $N-4$ inactive Majoranas, a signed degeneracy factor $S(t_1,t_2,t_3)$ from the active block, and a deterministic transition $\boldsymbol{k} \mapsto \boldsymbol{\ell} = \Phi_{\boldsymbol{t}}(\boldsymbol{k})$.

Including normalization, the transfer matrix takes the form
\begin{equation}
\begin{aligned}
C^{(3)}_{\boldsymbol{\ell},\boldsymbol{k}}
&=
\mathcal N(\boldsymbol{\ell})\mathcal N(\boldsymbol{k})
\sum_{\boldsymbol{t}}
S(t_1,t_2,t_3) \\
&\quad \times
\binom{N-4}{k_1-t_1,\;k_2-t_2,\;k_3-t_3,\;k_0-t_0}
\,\delta_{\boldsymbol{\ell},\,\Phi_{\boldsymbol{t}}(\boldsymbol{k})}.
\end{aligned}
\end{equation}

\subsection{Generalization to arbitrary doping angle}

The derivation above extends naturally to the one-parameter family of quartic dopants
\begin{equation}
A_\theta := e^{i \theta Z_1 Z_2}
=
e^{-i\theta \Gamma_T},
\qquad
\Gamma_T = \gamma_1\gamma_2\gamma_3\gamma_4.
\end{equation}

Using $\Gamma_T^2=\boldsymbol{1}$, conjugation of a Majorana monomial $\gamma_S$ takes the form
\begin{equation}
A_\theta \gamma_S A_\theta^\dagger
=
\begin{cases}
\gamma_S, & |S\cap T|\ \text{even},\\[2pt]
\cos(2\theta)\,\gamma_S
+
i\sin(2\theta)(-\Gamma_T)\gamma_S,
& |S\cap T|\ \text{odd}.
\end{cases}
\label{eq:general_theta_conj}
\end{equation}

Thus, for the odd replica monomials appearing in the three-copy channel, the action of the dopant becomes a superposition of the identity contribution and the transformed monomial appearing in the $\theta=\pi/4$ case.

Expanding the tensor-product action therefore produces: (i) terms in which no odd monomial is transformed, (ii) terms in which exactly one odd monomial is transformed, and (iii) terms in which both odd monomials are transformed. The single-transformation terms are orthogonal to the invariant subspace spanned by the basis
$\{|\Upsilon_{\boldsymbol{k}}^{(3)}\rrangle\}$ and are therefore annihilated by the matchgate twirl. Consequently, only the identity contribution and the fully transformed contribution survive projection.

The resulting three-copy transfer matrix takes the simple form
\begin{equation}
C^{(3)}(\theta)
=
\cos^2(2\theta)\,\mathbb{1}
+
\sin^2(2\theta)\,
C^{(3)}\!\left(\frac{\pi}{4}\right).
\label{eq:Ctheta_lazy}
\end{equation}

Equivalently,
\begin{equation}
C^{(3)}(\theta)
=
\mathbb{1}
-
\sin^2(2\theta)
\left[
\mathbb{1}
-
C^{(3)}\!\left(\frac{\pi}{4}\right)
\right].
\end{equation}

Thus the arbitrary-angle dopant corresponds, at the level of the projected three-copy dynamics, to a ``lazy'' version of the maximally non-Gaussian $\theta=\pi/4$ transfer matrix, with effective transition probability
\begin{equation}
p(\theta)=\sin^2(2\theta).
\end{equation}

As a consequence, the eigenvectors are unchanged and the eigenvalues transform as
\begin{equation}
\lambda_j(\theta)
=
1
-
\sin^2(2\theta)
\left[
1-\lambda_j\!\left(\frac{\pi}{4}\right)
\right].
\end{equation}

In particular, if the spectral gap at $\theta=\pi/4$ scales as $\Delta^{(3)}\!\left(\frac{\pi}{4}\right) \sim \frac{2}{n}$, then
\begin{equation}
\Delta^{(3)}(\theta)
=
\sin^2(2\theta)\,
\Delta^{(3)}\!\left(\frac{\pi}{4}\right),
\end{equation}
In complete analogy to what happens for $2$-copies.

\section{Higher-order continuum expansion of the transfer-matrix spectrum}
\label{app:OU_perturbation}

The continuum analysis of Sec.~\ref{sec:continuum} captures the leading
behaviour $\lambda_k = 1 - 2k/n + \mathcal{O}(n^{-2})$ of the low-lying
eigenvalues of the chain $T$. Here we extend that analysis to a systematic
expansion in $\varepsilon := 1/n$, obtaining the eigenvalues to order $n^{-3}$.
In particular, this fixes the subleading correction to the spectral gap observed
numerically in the inset of Fig.~\ref{fig:gap_layout}(\textbf{A}).
Throughout, the expansion is a bulk, fixed-$k$ one ($k$ held fixed while
$n\to\infty$): it controls the low Ornstein--Uhlenbeck modes and the
spectral gap, but is not uniform in $k$ and does not by itself describe
boundary-localized or high-frequency modes.

\subsection{Operator expansion of the transfer matrix}

We work in the canonical sector $\mathcal{S}_0 = 8\mathbb{Z}\cap[-2n,2n]$ and
use the bulk coordinate of Sec.~\ref{sec:continuum},
\begin{equation}
  x = \frac{\nu}{\sqrt{2n}},\qquad
  \delta x = \frac{8}{\sqrt{2n}} = 4\sqrt{2}\,\varepsilon^{1/2},\qquad
  \varepsilon = \frac{1}{n}.
  \label{eq:app_scaling}
\end{equation}
For a smooth test function $f$, the backward action of $T$ in
Eq.~\eqref{eq:T_birth_death} reads
$(Tf)(\nu)=a_\nu f(\nu-8)+b_\nu f(\nu+8)+(1-a_\nu-b_\nu)f(\nu)$.
Writing $f(\nu\pm 8)=f(x\pm\delta x)$ and Taylor-expanding, the even and odd
derivatives are weighted by the symmetric and antisymmetric rate combinations
$S_\nu := a_\nu+b_\nu$ and $D_\nu := b_\nu - a_\nu$,
\begin{equation}
  Tf-f \;=\; \sum_{m\ \mathrm{odd}} \frac{\delta x^m}{m!}\,D_\nu\,\partial_x^m f
            \;+\; \sum_{m\ \mathrm{even}} \frac{\delta x^m}{m!}\,S_\nu\,\partial_x^m f .
  \label{eq:app_bookkeeping}
\end{equation}
The rates in Eq.~\eqref{eq:T_birth_death} are ratios of quartic falling
factorials,
\begin{equation}
\begin{aligned}
  a_\nu &= \frac{\prod_{i=0}^{3}(n+\nu/2-i)}{\prod_{i=0}^{3}(2n-i)},\qquad
  b_\nu = \frac{\prod_{i=0}^{3}(n-\nu/2-i)}{\prod_{i=0}^{3}(2n-i)} .
\end{aligned}
\end{equation}
Substituting $\nu = x\sqrt{2n}$ and expanding in $\varepsilon$ gives
\begin{align}
  S_\nu &= \tfrac{1}{8}
         + \varepsilon\!\left(\tfrac{3x^2}{8}-\tfrac{3}{8}\right)
         + \varepsilon^2\!\left(\tfrac{x^4}{32}-\tfrac{3}{32}\right)
         + \mathcal{O}(\varepsilon^3),
         \label{eq:app_S}\\[2pt]
  D_\nu &= -\tfrac{\sqrt{2}}{4}\,x\,\varepsilon^{1/2}
         + \tfrac{\sqrt{2}}{8}(-x^3+3x)\,\varepsilon^{3/2}
         \nonumber\\
        &\quad + \tfrac{\sqrt{2}}{16}(-3x^3+7x)\,\varepsilon^{5/2}
         + \mathcal{O}(\varepsilon^{7/2}).
         \label{eq:app_D}
\end{align}
Inserting Eqs.~\eqref{eq:app_S}--\eqref{eq:app_D} and the powers of $\delta x$ into
Eq.~\eqref{eq:app_bookkeeping} and collecting equal powers of $\varepsilon$
yields
\begin{equation}
  T \;=\; I + \varepsilon\,\mathcal{L}_0 + \varepsilon^2\,\mathcal{L}_1
          + \varepsilon^3\,\mathcal{L}_2 + \mathcal{O}(\varepsilon^4),
  \label{eq:app_T_expansion}
\end{equation}
with $\mathcal{L}_0$ the OU generator of Eq.~\eqref{eq:OU_generator} and
\begin{align}
  \mathcal{L}_0 &= 2\partial_x^2 - 2x\partial_x,
  \label{eq:app_L0}\\[2pt]
  \mathcal{L}_1 &= (-x^3+3x)\partial_x + (6x^2-6)\partial_x^2
                 - \tfrac{32}{3}x\partial_x^3 + \tfrac{16}{3}\partial_x^4,
  \label{eq:app_L1}\\[2pt]
  \mathcal{L}_2 &= \left(-\tfrac{3}{2}x^3+\tfrac{7}{2}x\right)\partial_x
                 + \tfrac{x^4-3}{2}\,\partial_x^2
                 + \left(-\tfrac{16}{3}x^3+16x\right)\partial_x^3
                 \nonumber\\
                &\quad + (16x^2-16)\partial_x^4
                 - \tfrac{256}{15}x\partial_x^5 + \tfrac{256}{45}\partial_x^6 .
  \label{eq:app_L2}
\end{align}

\subsection{Rayleigh--Schr\"odinger perturbation theory}

The eigenproblem $T\phi_k=\lambda_k\phi_k$ is equivalent to
$\mathcal{L}(\varepsilon)\,\phi_k=\mu_k\,\phi_k$ with
$\mathcal{L}(\varepsilon)=\mathcal{L}_0+\varepsilon\mathcal{L}_1
+\varepsilon^2\mathcal{L}_2$ and $\lambda_k = 1+\varepsilon\mu_k$. We expand
\begin{equation}
  \mu_k = \mu_k^{(0)} + \varepsilon\,\mu_k^{(1)} + \varepsilon^2\,\mu_k^{(2)},
  \qquad
  \phi_k = \mathrm{He}_k + \varepsilon\,\phi_k^{(1)} + \mathcal{O}(\varepsilon^2),
\end{equation}
where $\mathrm{He}_k$ are the probabilists' Hermite polynomials. They diagonalize
the unperturbed problem: using $\partial_x\mathrm{He}_k = k\,\mathrm{He}_{k-1}$
and $x\,\mathrm{He}_k = \mathrm{He}_{k+1}+k\,\mathrm{He}_{k-1}$ one finds
$\mathcal{L}_0\,\mathrm{He}_k = -2k\,\mathrm{He}_k$, so $\mu_k^{(0)}=-2k$. Because
$\mathcal{L}_0$ is self-adjoint with respect to the Gaussian weight
$e^{-x^2/2}$, the $\{\mathrm{He}_k\}$ are orthogonal in $L^2(e^{-x^2/2})$, and the
projection $[\,\cdot\,]_{\mathrm{He}_k}$ onto the $\mathrm{He}_k$ component
implements the standard non-degenerate perturbation projector. We adopt the
gauge in which $\phi_k^{(j)}$ carries no $\mathrm{He}_k$ component.

The first correction is the diagonal coefficient of
$\mathcal{L}_1\mathrm{He}_k$. Repeated use of the two Hermite identities above
gives
\begin{equation}
\begin{aligned}
  \mathcal{L}_1\mathrm{He}_k =\,& -k\,\mathrm{He}_{k+2}
       + 3k(k-1)\,\mathrm{He}_k \\
     & - \tfrac{5}{3}k(k-1)(k-2)\,\mathrm{He}_{k-2}\\
     & - \tfrac{1}{3}k(k-1)(k-2)(k-3)\,\mathrm{He}_{k-4},
\end{aligned}
\label{eq:app_L1_He}
\end{equation}
from where $\mu_k^{(1)} = 3k(k-1)$ and
\begin{equation}
  \lambda_k = 1 - \frac{2k}{n} + \frac{3k(k-1)}{n^2} + \mathcal{O}(n^{-3}).
\end{equation}

For the $n^{-3}$ term we first solve the order-$\varepsilon$ equation
$(\mathcal{L}_0+2k)\phi_k^{(1)} = -(\mathcal{L}_1-\mu_k^{(1)})\mathrm{He}_k$.
Since $(\mathcal{L}_0+2k)\mathrm{He}_{k+m}=-2m\,\mathrm{He}_{k+m}$, the
off-diagonal terms of Eq.~\eqref{eq:app_L1_He} invert termwise to
\begin{equation}
\begin{aligned}
  \phi_k^{(1)} =\,& -\tfrac{k}{4}\,\mathrm{He}_{k+2}
     + \tfrac{5}{12}k(k-1)(k-2)\,\mathrm{He}_{k-2}\\
     & + \tfrac{1}{24}k(k-1)(k-2)(k-3)\,\mathrm{He}_{k-4}.
\end{aligned}
\end{equation}
The second correction has the standard two contributions,
$\mu_k^{(2)} = [\mathcal{L}_2\mathrm{He}_k]_{\mathrm{He}_k}
            + [\mathcal{L}_1\phi_k^{(1)}]_{\mathrm{He}_k}$.
Extracting the diagonal coefficient of $\mathcal{L}_2\mathrm{He}_k$ gives
$[\mathcal{L}_2\mathrm{He}_k]_{\mathrm{He}_k}
 = -\tfrac{1}{6}k(20k^2-51k+37)$, while only the $\mathrm{He}_{k\pm 2}$
components of $\phi_k^{(1)}$ feed back through $\mathcal{L}_1$ (which changes the
Hermite degree by at most two), yielding
$[\mathcal{L}_1\phi_k^{(1)}]_{\mathrm{He}_k}
 = \tfrac{5}{6}k(4k^2-3k+2)$. Their sum is
\begin{equation}
  \mu_k^{(2)} = \frac{3k(4k-3)}{2}.
\end{equation}

\subsection{Low-mode eigenvalues and the spectral gap}

Collecting the three orders, the low-lying spectrum of $T$ expands as
\begin{equation}
  \;
  \lambda_k = 1 - \frac{2k}{n} + \frac{3k(k-1)}{n^2}
            + \frac{3k(4k-3)}{2n^3} + \mathcal{O}(n^{-4}).
  \;
  \label{eq:app_lambda_final}
\end{equation}
The spectral gap of Eq.~\eqref{eq:gap_def_markov} is therefore
\begin{equation}
  \Delta = 1-\lambda_1 = \frac{2}{n} - \frac{3}{2n^3} + \mathcal{O}(n^{-4}).
  \label{eq:app_gap_expansion}
\end{equation}

\medskip
Finally, these eigenvalues fix the thermodynamic-limit form of the unitary frame
potential. Inserting them into the spectral representation~\eqref{eq:FP_spectral} the
trace becomes a geometric sum over the OU levels. Each of the four parity blocks
$T^{(a)}$ contributes the same spectrum in the leading order, so as $n \rightarrow \infty$
\begin{equation}
  \mathcal{F}(t)-\mathcal{F}^{(2)}_{\mathrm H}\;\to\;\frac{4}{e^{4\tau}-1},
  \label{eq:app_FP_continuum}
\end{equation}
where the prefactor matches the expected Haar frame value $4=2^{k-1}k!\big|_{k=2}=\mathcal{F}^{(2)}_{\mathrm H}$. This is
the thermodynamic limit curve compared with the exact transfer-matrix data in
Fig.~\ref{fig:global_2f}(\textbf{A}); at large $\tau$ it reduces to $4e^{-4\tau}$,
set by the gap mode $\lambda_1^{2t}$.

\section{The Ehrenfest urn and its analogy with the $T$ chain}
\label{app:ehrenfest}

The classical \emph{Ehrenfest urn}~\cite{LevinPeresWilmer2006} models $n$ balls distributed between two urns: at each step a ball is picked uniformly at random and moved to the other urn. The state $X_t\in\{0,1,\dots,n\}$ is the number of balls in the first urn. For the analogy with our chain, we assume $n$ even and pass to the centred state $Y_t:=X_t-n/2\in\{-n/2,\dots,n/2\}$, on which the dynamics is a reversible birth--death Markov chain with transition rates
\begin{equation}
    P(Y_t\to Y_t-1)\;=\;\frac{1}{2}+\frac{Y_t}{n},
    \qquad
    P(Y_t\to Y_t+1)\;=\;\frac{1}{2}-\frac{Y_t}{n}.
\end{equation}
The stationary distribution is binomial, $\pi_{\mathrm{Ehr}}(k)=\binom{n}{n/2+k}2^{-n}$, and the eigenvalues of the transition matrix are $\{1-2j/n\}_{j=0}^{n}$, giving spectral gap $\Delta_{\mathrm{Ehr}}=2/n$. The eigenfunctions are Krawtchouk polynomials; in particular, the linear function $\phi_{\mathrm{Ehr}}(k):=k$ is an \emph{exact} eigenfunction with eigenvalue $1-2/n$. The mixing time is known at leading order~\cite{LevinPeresWilmer2006},
\begin{equation}
    t_{\mathrm{mix}}^{\mathrm{Ehr}}(\epsilon)\;=\;\tfrac{n}{2}\log n\;+\;\mathcal{O}(n).
\end{equation}

The analogy with the chain $T$ defined by Eq.~\eqref{eq:T_birth_death} is most transparent in the \emph{bulk} regime $|\nu|=\mathcal{O}(\sqrt n)$, where the typical fluctuations of $p_{\mathrm H}$ live. Writing $a_\nu,b_\nu$ as ratios of falling factorials,
\begin{equation}
\begin{aligned}
    a_\nu&\;=\;\frac{(n+\nu/2)(n+\nu/2-1)(n+\nu/2-2)(n+\nu/2-3)}{(2n)(2n-1)(2n-2)(2n-3)},\\[2pt]
    b_\nu&\;=\;\frac{(n-\nu/2)(n-\nu/2-1)(n-\nu/2-2)(n-\nu/2-3)}{(2n)(2n-1)(2n-2)(2n-3)},
\end{aligned}
\label{eq:ab_falling_app}
\end{equation}
and expanding to first order in $\nu/n$ in the bulk $|\nu|=\mathcal{O}(\sqrt n)$,
\begin{equation}
    a_\nu\;=\;\frac{1}{16}\;+\;\frac{\nu}{8n}\;+\;\mathcal{O}(1/n),
    \qquad
    b_\nu\;=\;\frac{1}{16}\;-\;\frac{\nu}{8n}\;+\;\mathcal{O}(1/n),
\end{equation}
where we used $\nu^2/n^2=\mathcal{O}(1/n)$ for $|\nu|=\mathcal{O}(\sqrt n)$. To leading order, the chain $T$ in the bulk is therefore a \emph{lazy Ehrenfest urn}: the chain holds with probability $1-(a_\nu+b_\nu)=7/8+\mathcal{O}(1/n)$, and otherwise jumps by $\pm 8$ with the linear-in-position rates above, which are exactly the Ehrenfest form rescaled by the move probability $1/8$.

\section{Lower bound on the spectral gap}

In this appendix, we provide the rigorous proof of the lower bound on the spectral gap stated in the main text. 

\subsection{Bound on $\lambda_{1}$}
\label{app:gap_lower}

In this section, we prove Lemma~\ref{lem:gap_lower} of the main text: $1-\lambda_1 \ge 1/(32\pi n)$ for $n$ sufficiently large. 

\begin{proof}[Proof of Lemma~\ref{lem:gap_lower}]
The proof applies the Cheeger inequality (Eq.~\eqref{eq:cheeger_main}) to the reversible Markov chain $T$ defined by Eq.~\eqref{eq:T_birth_death}, with stationary distribution $\pi=p_{\mathrm H}$ from Eq.~\eqref{eq:Haar_dist} and state space $\mathcal{S}=8\mathbb{Z}\cap[-2n,2n]$. By the Cheeger inequality, it suffices to establish the following bound on the bottleneck ratio:
\begin{equation}\label{eq:phi_star_target}
    \Phi_*\;\geq\;\frac{1}{4\sqrt{\pi n}}.
\end{equation}

The birth--death structure of $T$ dramatically restricts the optimization in the definition of $\Phi_*$ in~\eqref{eq:cheeger_const_main}. Since $T$ has nonzero transitions only between consecutive sites $\nu$ and $\nu\pm 8$, any disconnected set $S=A\cup B$ (with $A$ and $B$ separated by at least one site) satisfies $Q(A\cup B,(A\cup B)^c)=Q(A,A^c)+Q(B,B^c)$, so $\min(\Phi(A),\Phi(B))\leq \Phi(A\cup B)$ and the minimizer of $\Phi$ can always be taken to be a single consecutive interval. By similar argument, one can show that the minimizer is further a half-line of the form
\begin{equation}
    S_*\;=\;[-2n,\nu_*]\cap\mathcal{S},
    \qquad
    \nu_*\,\leq\,-8,
\end{equation}
instead of an interior interval $[\nu_1,\nu_2]\cap\mathcal{S}$.
Such a half-line has a single boundary edge $(\nu_*,\nu_*+8)$, so the bottleneck ratio simplifies to
\begin{equation}\label{eq:hardy_app}
    \Phi(S_*)\;=\;\frac{p_{\mathrm H}(\nu_*)\,b_{\nu_*}}{F(\nu_*)},
    \qquad F(\nu_*)\,:=\,\sum_{\nu\leq\nu_*}p_{\mathrm H}(\nu),
\end{equation}
and the problem reduces to the one-parameter optimization $\Phi_*=\min_{\nu_*\leq-8}\,1/H(\nu_*)$ with $H(\nu_*):=F(\nu_*)/[p_{\mathrm H}(\nu_*)b_{\nu_*}]$.

We claim that $H$ is maximized over $\{\nu\leq -8\}$ at the boundary $\nu_*=-8$. To see this, factor
\begin{equation}\label{eq:H_factored_app}
    \frac{1}{H(\nu)}\;=\;\frac{p_{\mathrm H}(\nu)\,b_{\nu}}{F(\nu)}
    \;=\;\frac{b_{\nu}}{F(\nu)/p_{\mathrm H}(\nu)},
\end{equation}
and analyze numerator and denominator separately. The rate
\begin{equation}
\begin{aligned}
    b_\nu\;&=\;\frac{\binom{2n-4}{n+\nu/2}}{\binom{2n}{n+\nu/2}} \\
    \;&=\;\frac{(n-\nu/2)(n-\nu/2-1)(n-\nu/2-2)(n-\nu/2-3)}{(2n)(2n-1)(2n-2)(2n-3)}
\end{aligned}
\end{equation}
is non-increasing in $\nu$ on $[-N,-8]$: for $\nu\leq -8$ each of the four numerator factors is at least $n$ and strictly increases as $\nu$ decreases. Hence $b_\nu\geq b_{-8}$ on this range. The denominator is given by
\begin{equation}
    \frac{F(\nu)}{p_{\mathrm H}(\nu)}
    \;=\;\sum_{i=0}^{(\nu+2n)/8}\frac{p_{\mathrm H}(\nu-8i)}{p_{\mathrm H}(\nu)},
\end{equation}
where
\begin{equation}
    \frac{p_{\mathrm H}(\nu-8i)}{p_{\mathrm H}(\nu)}
    \;=\;\prod_{l=1}^{4i}\frac{n+\nu/2-l+1}{n-\nu/2+l}.
\end{equation}
The ratio is thus non-decreasing in $\nu$: each factor in the product lies in $[0,1]$ and grows as $\nu$ increases (numerator increases, denominator decreases), so the same is true of every summand and of the whole sum. In particular, $F(\nu)/p_{\mathrm H}(\nu)\leq F(-8)/p_{\mathrm H}(-8)$ for $\nu\leq -8$. Combining the two monotonicities in~\eqref{eq:H_factored_app} gives $1/H(\nu)\geq 1/H(-8)$ throughout $\{\nu\leq -8\}$, and therefore $\Phi_*=1/H(-8)=p_{\mathrm H}(-8)\,b_{-8}/F(-8)$.

It remains to evaluate this expression asymptotically. The denominator is bounded trivially by $F(-8)<1/2$. For the numerator, Eqs.~\eqref{eq:Haar_dist} and~\eqref{eq:T_birth_death} together with Stirling's approximation give
\begin{equation}
    p_{\mathrm H}(-8)\,b_{-8}
    \;=\;\frac{4\binom{2n-4}{n-4}}{2^{n}(2^n+2)}
    \;=\;\frac{1}{4\sqrt{\pi n}}\bigl(1+\mathcal{O}(1/n)\bigr),
\end{equation}
so that, for $n$ sufficiently large, $p_{\mathrm H}(-8)\,b_{-8}\geq 1/(8\sqrt{\pi n})$. Hence
\begin{equation}
    \Phi_*\;=\;\frac{p_{\mathrm H}(-8)\,b_{-8}}{F(-8)}
    \;\geq\;\frac{1}{4\sqrt{\pi n}},
\end{equation}
which yields $1-\lambda_1\geq 1/(32\pi n)$ by the Cheeger inequality.
\end{proof}

\subsection{Bound on $\lambda_{\mathrm{min}}$}
\label{app:min_eigenvalue}

In this section, we prove Lemma~\ref{lem:min_eigenvalue}: $1-\lvert\lambda_\mathrm{min}\rvert \ge 1/(32\pi n)$ for $n$ sufficiently large.

\begin{proof}[Proof of Lemma~\ref{lem:min_eigenvalue}]
The proof applies the dual Cheeger inequality (Eq.~\eqref{eq:dual_cheeger_main}) to the reversible Markov chain $T$ defined by Eq.~\eqref{eq:T_birth_death}, with stationary distribution $\pi=p_{\mathrm H}$ from Eq.~\eqref{eq:Haar_dist} and state space $\mathcal{S}=8\mathbb{Z}\cap[-2n,2n]$. By the dual Cheeger inequality, it suffices to establish the following bound on the bipartiteness ratio:
\begin{equation}\label{eq:beta_star_target}
    \beta\;\geq\;\frac{1}{4\sqrt{\pi n}}.
\end{equation}

Given a partition $(V_1, V_2)$, the structure of the bipartiteness ratio~\eqref{eq:bipart_main} naturally separates two cases depending on whether $V_3$ is empty. When $V_3 = \emptyset$ the boundary terms vanish and only the within-part flows $Q(V_i, V_i)$ contribute. Each within-part flow is bounded below by the self-loop contribution at any single state, $Q(V_i, V_i) \ge \sum_{\nu\in V_i}p_{\mathrm H}(\nu)\,T(\nu,\nu)$, so retaining only the central state $\nu = 0$ gives
\begin{equation}
\begin{aligned}
    \beta(V_1, V_2) &\;\ge\; 2\sum_{\nu \in V}p_{\mathrm H}(\nu)\,T(\nu,\nu) \\
    &\;\ge\; 2\,p_{\mathrm H}(0)\,T(0,0) \\
    &\;=\; \frac{4\binom{2n}{n}}{2^{n}(2^n+2)}\frac{n (7 n-11)}{4n^2-8n+3} \\
    &\;=\; \frac{7}{\sqrt{\pi n}}\bigl(1+\mathcal{O}(1/n)\bigr),
    \label{eq:case1}
\end{aligned}
\end{equation}
where we used Stirling's approximation in the last inequality. 

When instead $V_3\ne\emptyset$, dropping the non-negative self-loop terms in~\eqref{eq:bipart_main} reduces $\beta$ to the bottleneck ratio of the cut $V_1\cup V_2 \mid V_3$ for the chain $T$,
\begin{equation}
    \beta(V_1, V_2) \;\ge\; \frac{Q(V_1, V_3) + Q(V_2, V_3)}{p_{\mathrm H}(V_1) + p_{\mathrm H}(V_2)} \;\ge\; \Phi(T) \;\ge\; \frac{1}{4\sqrt{\pi n}},
    \label{eq:case2}
\end{equation}
where the last inequality was established in the proof of Lemma~\ref{lem:gap_lower} in the previous subsection. The bounds~\eqref{eq:case1} and~\eqref{eq:case2} together imply 
\begin{equation}
    \beta
    \;\geq\;\frac{1}{4\sqrt{\pi n}},
\end{equation}
which yields $1-\lvert\lambda_\mathrm{min}\rvert\geq 1/(32\pi n)$ by the dual Cheeger inequality.

\end{proof}

\section{Lower bound on the mixing time}
\label{app:mixing_time}

In this appendix, we establish the lower bound on the mixing time of the chain $T$ stated in Eq.~\eqref{eq:mixing_lower_main} of the main text:

\begin{thm}\label{thm:mixing_lower}
For every $\epsilon\in(0,1)$ there exist a constant $K_\epsilon$ and integer $n_\epsilon>0$ such that for all $n\geq n_\epsilon$,
\begin{equation}\label{eq:mixing_lower}
    t_{\mathrm{mix}}(\epsilon)\;\geq\;\tfrac{1}{4}\,n\log n\;+\;K_\epsilon\,n.
\end{equation}
\end{thm}

The proof applies the distinguishing-statistic method:

\begin{lem}[Distinguishing statistic, {\cite[Prop.~7.9]{LevinPeresWilmer2006}}]\label{lem:wilson}
Let $\mu,\pi$ be probability measures on a common state space and $f$ a real-valued measurable function with finite second moment under both. Let $\sigma^2=\max(\mathrm{Var}_\mu(f),\mathrm{Var}_\pi(f))$. Then
\begin{equation}\label{eq:wilson_bound}
    \|\mu-\pi\|_{\mathrm{TV}}\;\geq\;1-\frac{8\,\sigma^2}{(\mathbb{E}_\mu f-\mathbb{E}_\pi f)^2}.
\end{equation}
\end{lem}

 We first establish four preliminary lemmas before turning to the proof of Theorem~\ref{thm:mixing_lower}.

\begin{lem}[Action of $T$ on $\phi$]\label{lem:T_action}
For all $\nu\in 8\mathbb{Z}\cap[-2n,2n]$,
\begin{equation}\label{eq:Tphi_exact}
    (T\phi)(\nu)\;=\;\beta_n\,\nu\;-\;\gamma_n\,\nu^3,
\end{equation}
where
\begin{equation}
    \beta_n\;=\;1-\frac{4(n^2-3n+1)}{n(n-1)(2n-1)},
    \qquad
    \gamma_n\;=\;\frac{1}{n(n-1)(2n-1)}.
\end{equation}
\end{lem}

\begin{proof}
By definition of $T$,
\begin{equation}
\begin{aligned}
    (T\phi)(\nu)\;&=\;\sum_{\nu'}T(\nu,\nu')\,\nu'\\
    \;&=\;a_\nu(\nu-8)+b_\nu(\nu+8)+(1-a_\nu-b_\nu)\nu\\
    \;&=\;\nu+8(b_\nu-a_\nu).
\end{aligned}
\end{equation}
Substituting the closed-form expressions of $a_\nu,b_\nu$ from~\eqref{eq:T_birth_death} and simplifying yields~\eqref{eq:Tphi_exact}.
\end{proof}

\begin{lem}[Variance bound]\label{lem:V_bound}
For every deterministic starting state $\nu_0\in 8\mathbb{Z}\cap[-2n,2n]$, every $n\geq 6$, and every $t\geq 0$,
\begin{equation}\label{eq:Var_uniform}
    \mathrm{Var}_{p_t}(\nu)\;\leq\;32n,
\end{equation}
where $p_t$ denotes the law of $\nu_t$ given $\nu_0$.
\end{lem}

\begin{proof}
Throughout the proof we let $(\nu_t)_{t\geq 0}$ denote the Markov chain with transition matrix $T$ and deterministic initial state $\nu_0$, so that the law of $\nu_t$ is $p_t$; thus $\mathrm{Var}(\nu_t)$ stands for $\mathrm{Var}_{p_t}(\nu)$ and $\mathbb{E}[\,\cdot\mid\nu_t]$ for conditional expectation given the chain at time $t$. Set $g(x):=\beta_n x-\gamma_n x^3$ so that $\mathbb{E}[\nu_{t+1}\mid\nu_t]=g(\nu_t)$ by~\eqref{eq:Tphi_exact}. Its derivative $g'(x)=\beta_n-3\gamma_n x^2$ obeys
\begin{equation}
    -\beta_n\;\leq\;\beta_n-\frac{12n^2}{n(n-1)(2n-1)}\;\leq\;g'(x)\;\leq\;\beta_n,
\end{equation}
for $|x|\leq 2n$ and $n\geq 6$, using $12n^2/[n(n-1)(2n-1)]\leq 2\beta_n$ for $n\geq 6$. Hence $g$ is $\beta_n$-Lipschitz on the state space, so for any random variable $X$ supported there,
\begin{equation}\label{eq:g_var_contraction}
    \mathrm{Var}\bigl(g(X)\bigr)\;\leq\;\beta_n^2\,\mathrm{Var}(X).
\end{equation}
Furthermore,
\begin{equation}\label{eq:cond_var}
    \mathrm{Var}(\nu_{t+1}\mid\nu_t)\;\leq\;\mathbb{E}\!\left[(\nu_{t+1}-\nu_t)^2\,\big|\,\nu_t\right]\;=\;64(a_{\nu_t}+b_{\nu_t})\;\leq\;64,
\end{equation}
where the first inequality uses the variational definition $\mathrm{Var}(Y)=\inf_c\mathbb{E}[(Y-c)^2]$ with the choice $c=\nu_t$. Combining the law of total variance with~\eqref{eq:g_var_contraction}--\eqref{eq:cond_var},
\begin{equation}
\begin{aligned}
    \mathrm{Var}(\nu_{t+1})\;&=\;\mathrm{Var}\bigl(g(\nu_t)\bigr)+\mathbb{E}\bigl[\mathrm{Var}(\nu_{t+1}\!\mid\!\nu_t)\bigr]\\
    &\leq\;\beta_n^2\,\mathrm{Var}(\nu_t)+64.
\end{aligned}
\end{equation}
Since $\mathrm{Var}(\nu_0)=0$, iterating gives $\mathrm{Var}(\nu_t)\leq 64/(1-\beta_n^2)$. Finally $1-\beta_n^2\geq 2/n$ for $n\geq 6$, yielding~\eqref{eq:Var_uniform}.
\end{proof}

\begin{lem}\label{lem:Mtilde}
Let $\widetilde M_0\in[0,2n]$ and $\widetilde M_{t+1}:=g(\widetilde M_t)$. For all $n\geq 9$ and all $t\geq 0$,
\begin{equation}\label{eq:Mtilde_lb}
    \widetilde M_t\;\geq\;\widetilde M_0\,\beta_n^{\,t}\,e^{-2}.
\end{equation}
\end{lem}

\begin{proof}
For $x\in[0,2n]$ we have $g(x)=x(\beta_n-\gamma_n x^2)$ with $\beta_n-\gamma_n x^2\geq\beta_n-4n^2\gamma_n\geq 0$ for $n\geq 9$, and $g(x)\leq\beta_n x\leq 2n$; hence $g$ maps $[0,2n]$ into $[0,2n]$, so $\widetilde M_t\in[0,2n]$ for all $t$, with $\widetilde M_{t+1}\leq\beta_n\widetilde M_t$ as the cubic correction is non-negative. Hence $\widetilde M_t\leq\widetilde M_0\beta_n^{\,t}$. Taking logs of $\widetilde M_{t+1}/\widetilde M_t=\beta_n(1-(\gamma_n/\beta_n)\widetilde M_t^2)$ and using $\log(1-y)\geq -2y$ for $y\in[0,1/2]$ (here $\gamma_n\widetilde M_t^2/\beta_n\leq 4\gamma_n n^2/\beta_n\leq 1/2$ for $n\geq 9$),
\begin{equation}
\begin{aligned}
    \log\widetilde M_t\;&\geq\;\log\widetilde M_0+t\log\beta_n-\frac{2\gamma_n}{\beta_n}\sum_{s=0}^{t-1}\widetilde M_s^2\\
    \;&\geq\;\log\widetilde M_0+t\log\beta_n-\frac{2\gamma_n}{\beta_n}\frac{\widetilde M_0^2}{1-\beta_n^2}.
\end{aligned}
\end{equation}
The last term satisfies $2\gamma_n\widetilde M_0^2/[\beta_n(1-\beta_n^2)]\leq 8\gamma_n n^2/[\beta_n(1-\beta_n^2)]\leq 2$ for $n\geq 9$. Exponentiating yields~\eqref{eq:Mtilde_lb}.
\end{proof}

\begin{lem}[Mean bound]\label{lem:Mt_bound}
Let $\nu_0:=8\lfloor n/4\rfloor$, so that $2n-8\leq\nu_0\leq 2n$, and let $M_t:=\mathbb{E}_{p_t}[\nu]$ denote the mean of the chain at time $t$ initialized at $\nu_0$. There exist absolute constants $C_0,C_1>0$ such that for all $n\geq 9$ and all $t\geq 0$,
\begin{equation}\label{eq:Mt_lb}
    M_t\;\geq\;C_0\,n\,\beta_n^{\,t}\;-\;C_1.
\end{equation}
\end{lem}

\begin{proof}
Let $\widetilde M_t$ be the deterministic recursion of Lemma~\ref{lem:Mtilde} with $\widetilde M_0=\nu_0$. By~\eqref{eq:Mtilde_lb},
\begin{equation}\label{eq:Mtilde_lower_proof}
    \widetilde M_t\;\geq\;\nu_0\,\beta_n^{\,t}\,e^{-2}\;\geq\;C_0\,n\,\beta_n^{\,t},
\end{equation}
for some positive constant $C_0$.

Set $V_t:=\mathrm{Var}_{p_t}(\nu)$, $\xi:=\nu-M_t$, and $D_t:=M_t-\widetilde M_t$. Taking the expectation of~\eqref{eq:Tphi_exact} under $p_t$ and using $\mathbb{E}_{p_t}[\nu^3]=M_t^3+3M_tV_t+\mathbb{E}_{p_t}[\xi^3]$,
\begin{equation}\label{eq:M_decomp}
    M_{t+1}\;=\;g(M_t)\;-\;\varepsilon_t,
    \qquad
    \varepsilon_t\;:=\;3\gamma_n M_t V_t+\gamma_n\,\mathbb{E}_{p_t}[\xi^3].
\end{equation}
Subtracting $\widetilde M_{t+1}=g(\widetilde M_t)$ and using the factoring $a^3-b^3=(a-b)(a^2+ab+b^2)$,
\begin{equation}
\begin{aligned}
    g(M_t)-g(\widetilde M_t)\;&=\;\beta_n D_t-\gamma_n(M_t^3-\widetilde M_t^3)\\
    \;&=\;\bigl[\beta_n-\gamma_n(M_t^2+M_t\widetilde M_t+\widetilde M_t^2)\bigr]D_t.
\end{aligned}
\end{equation}
Hence
\begin{equation}\label{eq:D_rec}
    D_{t+1}\;=\;\bigl[\beta_n-\gamma_n(M_t^2+M_t\widetilde M_t+\widetilde M_t^2)\bigr]D_t\;-\;\varepsilon_t.
\end{equation}
The residual $\varepsilon_t$ is bounded uniformly: by Lemma~\ref{lem:V_bound} and $|\xi|\leq 4n$,
\begin{equation}\label{eq:varepsilon_bound}
    |\varepsilon_t|\;\leq\;3\gamma_n(2n)(32n)\;+\;\gamma_n(4n)(32n)\;\leq\;\frac{C}{n},
\end{equation}
for some positive constant $C$.
The bracket in~\eqref{eq:D_rec} is bounded above by $\beta_n$ since $M_t^2+M_t\widetilde M_t+\widetilde M_t^2=(M_t+\widetilde M_t/2)^2+\tfrac{3}{4}\widetilde M_t^2\geq 0$, and bounded below by $\beta_n-12n^2\gamma_n>0$ for $n\geq 9$. Iterating with $D_0=0$,
\begin{equation}\label{eq:D_bound}
    |D_t|\;\leq\;\sum_{s=0}^{t-1}\beta_n^{\,t-1-s}\,\frac{C}{n}\;\leq\;\frac{C}{n(1-\beta_n)}\;\leq\;C_1,
\end{equation}
for some positive constant $C_1$. Combining with~\eqref{eq:Mtilde_lower_proof} yields~\eqref{eq:Mt_lb}.
\end{proof}

\begin{proof}[Proof of Theorem~\ref{thm:mixing_lower}]
Take $\nu_0:=8\lfloor n/4\rfloor$ as the (deterministic) starting state, and let $M_t:=\mathbb{E}_{p_t}[\nu]$. We apply Lemma~\ref{lem:wilson} with $\mu=p_t$, $\pi=p_{\mathrm H}$, and $f=\phi$. It is straightforward to show $\mathbb{E}_{p_{\mathrm H}}[\phi]=0$ (by symmetry of $p_{\mathrm H}$) and $\mathrm{Var}_{p_{\mathrm H}}(\phi)=2n(1+o(1))$ via standard binomial computation. Then, using $V_t=\mathrm{Var}_{p_t}(\phi)\leq 32n$ from Lemma~\ref{lem:V_bound},
\begin{equation}\label{eq:wilson_app}
    \|p_t-p_{\mathrm H}\|_{\mathrm{TV}}\;\geq\;1-\frac{8(32n)}{M_t^2}\;=\;1-\frac{C_2\,n}{M_t^2},
\end{equation}
where $C_2=256$.

Fix $\epsilon\in(0,1)$ and set $\delta:=1-\epsilon\in(0,1)$. By~\eqref{eq:wilson_app}, the bound $\|p_t-p_{\mathrm H}\|_{\mathrm{TV}}>\epsilon$ holds whenever $C_2 n/M_t^2<\delta$, i.e.,
\begin{equation}\label{eq:Mt_threshold}
    M_t\;>\;\sqrt{C_2 n/\delta}.
\end{equation}
By Lemma~\ref{lem:Mt_bound}, this is implied by $C_0 n\beta_n^{\,t}-C_1>\sqrt{C_2 n/\delta}$, which for $n\geq n_\epsilon$ large enough (so that $C_1\leq\sqrt{C_2 n/\delta}$) is implied by
\begin{equation}\label{eq:t_threshold}
    \beta_n^{\,t}\;>\;\frac{2}{C_0}\sqrt{\frac{C_2}{n\delta}}.
\end{equation}
Taking logs and using $-\log\beta_n=2/n+\mathcal{O}(n^{-2})$,
\begin{equation}
    -t\log\beta_n\;<\;\tfrac{1}{2}\log n+\tfrac{1}{2}\log(1/\delta)+\log\bigl(C_0/(2\sqrt{C_2})\bigr),
\end{equation}
which is equivalent to
\begin{equation}
    t\;<\;\tfrac{n}{4}\log n+K_\epsilon\,n,
\end{equation}
where $K_\epsilon$ is a constant depending on $\epsilon$. Set
\begin{equation}
    t^\star\;:=\;\left\lfloor\tfrac{n}{4}\log n+K_\epsilon\,n\right\rfloor.
\end{equation}
Then~\eqref{eq:t_threshold} holds at $t=t^\star$, hence $M_{t^\star}>\sqrt{C_2 n/\delta}$ and~\eqref{eq:wilson_app} gives $\|p_{t^\star}-p_{\mathrm H}\|_{\mathrm{TV}}>\epsilon$. By definition of mixing time, $t_{\mathrm{mix}}(\epsilon)\geq t^\star+1\geq\tfrac{n}{4}\log n+K_\epsilon\,n$, completing the proof.
\end{proof}

\section{Lower bound on the spectral gap of $q$-qubit doping}
\label{app:gap_lower_q_qubit}

In this appendix, we extend the lower bound on the spectral gap to general $q$-qubit doping. 

\begin{lem}\label{lem:lower_bound_gap_q_qubit}
Let $A$ be a $q$-qubit doping gate whose induced chain $T_A$ is reversible,
and suppose the two boundary entries of its local transition matrix are
nonzero,
\begin{equation}
  R_{2q,\,2q-8}\neq 0
  \qquad\text{and}\qquad
  R_{-2q+8,\,-2q}\neq 0 .
  \label{eq:boundary_entries}
\end{equation}
Then the spectral gap of $T_A$ satisfies
\begin{equation}
    1-\lambda_1(T_A) = \Omega(1/n).
    \label{eq:gap_scaling_q_qubit}
\end{equation}
\end{lem}

The hypothesis~\eqref{eq:boundary_entries} is generic: $R_{2q,2q-8}$ and
$R_{-2q+8,-2q}$ are polynomial functions of the entries of $A$, so they
vanish only on a measure-zero locus in the space of parity-preserving
$q$-qubit unitaries. Hence~\eqref{eq:boundary_entries} holds for all but a
measure-zero set of doping gates $A$. Note that, by parity $R_{-2q+8,-2q}=R_{2q-8,2q}$ and the reversibility
condition~\eqref{eq:R_dbc}, the hypothesis~\eqref{eq:boundary_entries} in fact reduces to the single
requirement $R_{2q,2q-8}\neq0$.

\begin{proof}

Denote the Markov chain generated by the SWAP chain as $T_{\rm SWAP}$. We will show that the Dirichlet forms of $T_{\rm SWAP}$ and $T_A$ satisfy
\begin{equation}\label{eq:dirichlet_inequality}
    \frac{1}{C} \mathcal E_{\rm SWAP}(f, f) \leq \mathcal E_{A}(f, f),
\end{equation}
for some constant $C$. By the variational characterization of the spectral gap~\eqref{eq:gap_variational}, this bound directly implies the corresponding bound
\begin{equation}
    \frac{1}{C} (1-\lambda_1(T_{\rm SWAP})) \leq (1-\lambda_1(T_{A})).
\end{equation}
Combined with Lemma~\ref{lem:gap_lower}, this yields~\eqref{eq:gap_scaling_q_qubit}.

To show Eq.~\eqref{eq:dirichlet_inequality}, since both chains have the same stationary distribution, it suffices to show that $a^{\rm SWAP}_\nu\le C a^{(1)}_\nu$. We exploit that all terms in the rate~\eqref{eq:ab_q_local} are non-negative, so retaining any single term in the sum gives a lower bound on $a^{(1)}_\nu$. We now divide into two cases and choose this term accordingly.
\begin{itemize}
    \item For $\nu\geq 0$, we retain the term $\mu=2q$ in the sum:
    \begin{equation}
        \frac{a_\nu^{(1)}}{a^{\rm SWAP}_\nu} \;\ge\;\frac{\binom{2n-2q}{n+\nu/2-2q}}{\binom{2n-4}{n+\nu/2-4}}R_{2q,2q-8} = R_{2q,2q-8}\prod_{i=0}^{2q-5}\frac{n+\nu/2-4-i}{2n-4-i}.
    \end{equation}
    For $\nu\geq 0$, each factor satisfies $(n+\nu/2-4-i)(2n-4-i)\geq (n-2q+1)/(2n-4)$, so the product is bounded below by $(n-2q+1)^{2q-4}/(2n-4)^{2q-4}\geq 2^{4-2q}(1+o(1))$ for fixed $q$. Therefore,
    \begin{equation}
        \frac{a_\nu^{(1)}}{a^{\rm SWAP}_\nu} \;\ge\; 2^{4-2q}R_{2q,2q-8} (1+o(1)).
    \end{equation}
    \item For $\nu< 0$, we retain the term $\mu=-2q+8$ in the sum:
    \begin{equation}
    \begin{split}
    \frac{a_\nu^{(1)}}{a^{\mathrm{SWAP}}_\nu}
    &\ge R_{-2q+8,-2q}\binom{2q}{4}
       \frac{\binom{2n-2q}{n+\nu/2-4}}{\binom{2n-4}{n+\nu/2-4}}\\[2pt]
    &= R_{-2q+8,-2q}\binom{2q}{4}
       \prod_{i=0}^{2q-5}\frac{n+\nu/2-i}{2n-4-i}.
    \end{split}
    \end{equation}
    For $\nu< 0$, each factor satisfies $(n-\nu/2-i)/(2n-4-i)> (n-2q+5)/(2n-4)$, so the product is bounded below by $(n-2q+5)^{2q-4}/(2n-4)^{2q-4}\geq 2^{4-2q}(1+o(1))$ for fixed $q$. Therefore,
    \begin{equation}
        \frac{a_\nu^{(1)}}{a^{\rm SWAP}_\nu} \;\ge\; 2^{4-2q}R_{-2q+8,-2q}\binom{2q}{4} (1+o(1)).
    \end{equation}
\end{itemize}
Combining both cases, we find that there exists a constant $C$ (dependent on $q$ and $A$) such that for all $n$ sufficiently large, $a^{\rm SWAP}_\nu\le C a^{(1)}_\nu$. This yields the claim.
\end{proof}

\section{Glued circuit in parity-preserving systems}
\label{app:pp_glued}
 
In this appendix, we show that the glued-circuit construction of
Ref.~\cite{schuster2024random} extends to parity-preserving systems. We focus on $k=2$,
though we expect the argument to generalize to higher $k$. The physical states on $2$
copies of an $n$-mode fermionic system are $2n$-mode fermionic states that commute with
$P \otimes P$; such a state decomposes as $\rho_e \oplus \rho_o$ on the even and odd
subspaces.
 
For a given unitary ensemble, consider the $k$-fold twirl channel
\begin{equation}
  \Phi^{(k)}_{\mathcal{E}}(X)
  :=
  \mathbb{E}_{U\sim\mathcal{E}}\bigl[U^{\otimes k}\,X\,(U^\dagger)^{\otimes k}\bigr].
\end{equation}
Focusing on $k=2$, we write $\Phi_{\mathcal{E}} \equiv \Phi^{(2)}_{\mathcal{E}}$. This
channel preserves the even and odd subspaces, so it decomposes as
$\Phi_{\mathcal{E}} = \Phi_{\mathcal{E}}^{(o)} + \Phi_{\mathcal{E}}^{(e)}$. The ensemble $\mathcal{E}$ forms a relative-error $\epsilon$-approximate unitary $2$-design
if
\begin{equation}
  (1-\epsilon)\,\Phi^{(o,e)}_{\mathrm{H}}
  \;\preceq\;
  \Phi^{(o,e)}_{\mathcal{E}}
  \;\preceq\;
  (1+\epsilon)\,\Phi^{(o,e)}_{\mathrm{H}},
  \label{eq:unitary_design_def}
\end{equation}
where $B \preceq A$ denotes that $A - B$ is a completely positive map.
 
The $2$-fold Haar twirl channel takes the form
\begin{equation}
  \Phi^{(e)}_{\mathrm{H}}(A)
  = \sum_{\pi,\sigma\in S_2}
    \Wg_e(\pi^{-1} \sigma)\,\Tr\!\bigl(\sigma^{-1} \Pi_e A\bigr)\,\Pi_e\pi,
\end{equation}
where $\Pi_e$ is the projector onto the even subspace, $\Wg_e((1)) = 2/(d^2-4)$ and $\Wg_e((1\,2)) = -4/(d(d^2-4))$. Following
Ref.~\cite{schuster2024random}, the Haar twirl on the even sector is well approximated by the
simpler channel
\begin{equation}
  \Phi^{(e)}_{a}(A)
  = \frac{2}{d^2}\sum_{\pi\in S_2}\Tr\!\bigl(\pi^{-1} \Pi_e A\bigr)\,\Pi_e\pi,
\end{equation}
within a relative error $\epsilon \sim 1/d$.
Intuitively, this follows from the fact that the Weingarten coefficients are dominated by
their diagonal entries, with off-diagonal contributions suppressed by an exponentially
small factor.
 
In the odd sector, the Haar twirl is already diagonal:
\begin{equation}
  \Phi^{(o)}_{\mathrm{H}}(A)
  = \frac{2}{d^2}\sum_{\pi\in S_2}\Tr\!\bigl(\pi^{-1}\Pi_o A\bigr)\,\Pi_o\pi,
\end{equation}
where $\Pi_o$ is the projector onto the odd subspace.
 
Using these results and following the steps of Ref.~\cite{schuster2024random}, one shows that
the glued-circuit construction with $\xi = \Theta(\log_2(n/\epsilon))$ forms an
$\epsilon$-approximate unitary $2$-design on $n$ qubits.

\section{State 4-frame potential}
\label{app:4-frame}

In the main text, our analysis focused on the emergence of design behaviour at low replica number. In particular, we developed an analytic treatment of the $k=2$ case through the commutant and Markov-chain framework, and extended the discussion to $k=3$ using the corresponding three-copy commutant together with numerical transfer-matrix methods. Beyond these orders, however, the problem becomes substantially more difficult.

Two complementary numerical strategies used throughout this work become less effective at higher replica numbers. First, Clifford matchgate circuits are known to form a matchgate $3$-design~\cite{wan2023matchgate,sierant2026theory}, implying that Clifford-based sampling faithfully reproduces observables only up to third order. Consequently, they cannot be used to probe the $k=4$ state-frame potential. Second, the commutant approach becomes increasingly complex as the replica number grows: while the two- and three-copy commutants remain tractable, the corresponding invariant space at four copies is sufficiently large that an explicit transfer-matrix construction is presently impractical.

For this reason, we restrict ourselves here to a direct numerical investigation of the state $4$-frame potential in the globally doped protocol of Eq.~\eqref{eq:circuit}. Specifically, we consider the same parity-preserving quartic dopant used throughout the paper, $A = e^{i\frac{\pi}{4} Z_1 Z_2}$,
and estimate the state-frame potential,
\begin{equation}
\mathcal S^{(4)}(t;\psi_0)
= \mathbb E_{U,V}\left|\langle \psi_0|U^\dagger V|\psi_0\rangle\right|^{8},
\end{equation}
for the reference product state $|\psi_0\rangle = |0\rangle^{\otimes n}$.

Figure~\ref{fig:4-frame} shows exact results for system sizes up to $n=16$. Despite the absence of an analytic description, the qualitative behaviour is consistent with the phenomenology observed for $k \leq 3$: the frame potential upon doping converges towards the Haar value, and the data suggest a convergence timescale compatible with the same scaling observed at lower replica number.

\begin{figure}[h!]
    \centering
    \includegraphics[width=\linewidth]{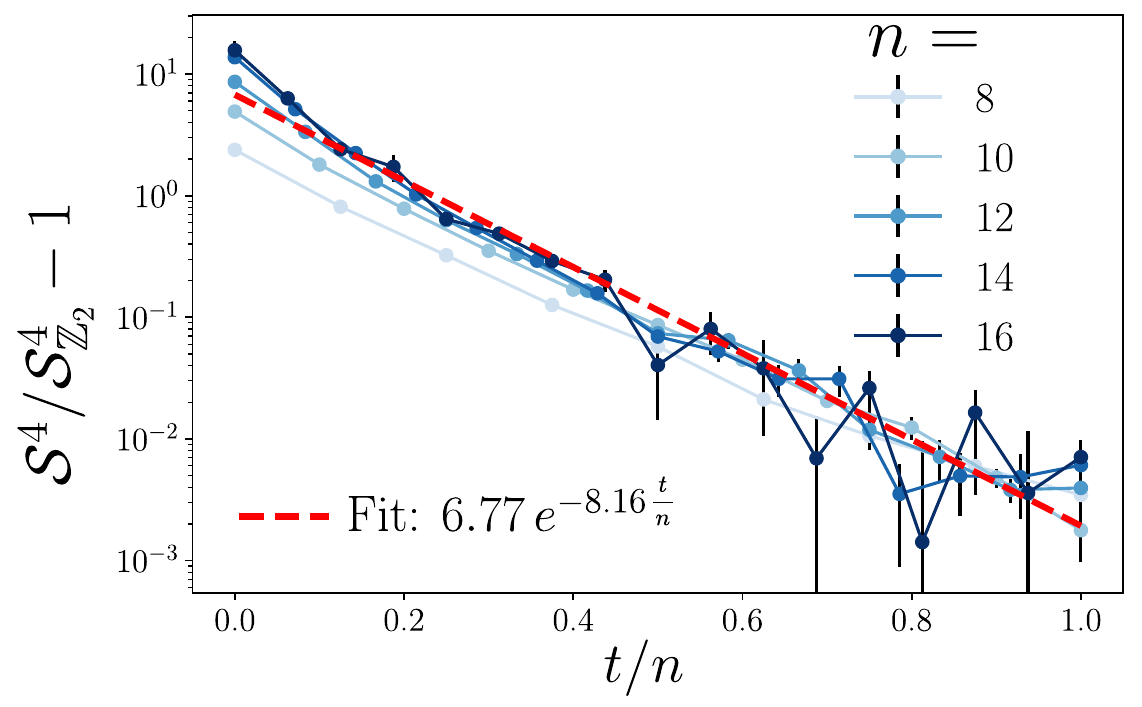}
    \caption{Relative deviation of the state 4-frame potential as a function of the normalized doping depth t/n for several system sizes n. Averaged over $10^7$ samples. The data exhibit an approximately exponential decay, shown by the dashed red line.}
    \label{fig:4-frame}
\end{figure}

At the same time, the numerical data exhibit substantially larger fluctuations than in the $k=2$ and $k=3$ cases. This originates from two related effects.

First, the absolute scale of the $4$-frame potential is extremely small. Since rare overlap events dominate the estimator, making Monte Carlo convergence considerably slower than at lower orders. In practice, this requires very large sample counts in order to obtain stable averages.

Second, fluctuations of higher-order frame potentials are themselves controlled by yet higher moments of the overlap distribution. As $k$ increases, the observable becomes progressively more sensitive to the tails of the overlap distribution, amplifying finite-sample noise. This effect is particularly pronounced for $k=4$, where the variance of the estimator receives dominant contributions from rare trajectories with anomalously large overlap.

To mitigate these issues, we used between $10^7$ and $10^8$ random circuit samples depending on system size. Even at this level of sampling, residual noise remains visible, particularly at larger doping depth where the frame potential approaches its asymptotic value.

Although the accessible system sizes are limited, the results nevertheless provide evidence that the global doping protocol continues to exhibit qualitatively similar convergence behaviour beyond the range directly accessible through commutant methods. Combined with the trace-distance bound in Eq.~\eqref{eq:frame_to_trace} and the lower bound from the convergence time of state $2$-design, these results suggest that the convergence to additive-error state $4$-design occurs at $t=\Theta(n)$, the same timescale as state $2$-design formation. Extending the analytic framework to higher-copy commutants remains an important open direction, and may confirm whether the scaling observed for $k=2$ and $k=3$ persists universally for higher moments.

\section{Numerical Techniques}
\label{app:Numerics}

In this appendix, we provide additional numerical details relevant to the computations performed in this work. In particular, we will provide an overview of the main computational techniques used for addressing the global and local cases, with an emphasis on those that probe large system sizes. 

\subsection{Clifford matchgate simulations}

To probe the effect of locality in doped matchgate circuits at large system sizes, we make extensive use of Clifford matchgate circuits. As discussed in the main text, Clifford matchgates form an exact matchgate $3$-design~\cite{wan2023matchgate,sierant2026theory}. Consequently, for observables involving moments up to third order, averaging over Clifford matchgates reproduces exactly the corresponding averages over Haar-random matchgate circuits. This allows us to access the $2$- and $3$-frame potentials without sampling generic fermionic Gaussian unitaries.

The primary advantage of this approach is computational scalability. Generic matrix representations of $n$-qubit unitaries require memory and runtime exponential in system size, making direct evaluation of frame potentials rapidly impractical. In contrast, Clifford matchgate circuits admit an efficient stabilizer description, enabling simulations at substantially larger $n$.

We do not review the stabilizer formalism in detail here, as it is standard~\cite{Aaronson2004,Gottesman1998,gidney2021stim}. Instead, we briefly describe how the relevant overlaps entering the frame potentials are evaluated without explicitly constructing the unitary matrices.

The key observation is that traces of unitary products can be rewritten as state overlaps using the Choi--Jamiołkowski isomorphism. For any unitary $U$ acting on an $n$-qubit Hilbert space,
\begin{equation}
    \mathrm{Tr}(U)
    =
    2^n \,
    \langle \Phi^+ | (U \otimes \mathbb I) | \Phi^+ \rangle,
\end{equation}
where $|\Phi^+\rangle=\frac{1}{\sqrt{2^n}}\sum_{x\in\{0,1\}^n}|x\rangle \otimes |x\rangle$ is the maximally entangled state between two copies of the system.

This identity allows the frame-potential estimator $|\mathrm{Tr}(U^\dagger V)|^{2k}$ to be recast as an overlap amplitude between stabilizer states evolved under doubled circuits. Operationally, we prepare a maximally entangled state between two copies of the system, apply the Clifford circuit to one copy, undo the Bell-pair preparation, and compute the amplitude of returning to the all-zero state by stabilizer postselection.

In practice, all Clifford matchgate simulations were performed using stabilizer tableau propagation, which scales polynomially in system size and enables sampling over very large ensembles of random circuits.

\subsection{3-frame potentials}

For the globally doped protocol, the three-copy dynamics reduces to repeated application of a finite-dimensional transfer matrix acting on the three-copy matchgate commutant. The construction of this basis and of the corresponding doping matrix \(C^{(3)}\) is described in App.~\ref{app:3copy}. Here we focus only on the numerical procedure used to evaluate the associated frame potentials.

For a circuit with \(t\) doping layers, the averaged three-copy evolution is represented by \((C^{(3)})^t\). In analogy to the 2-copy case in Eq.\,(\ref{eq:FP_spectral}), the unitary \(3\)-frame potential is obtained from the Frobenius norm of this transfer matrix,
\begin{equation}
\mathcal F^{(3)}(t)
=
\bigl\|(C^{(3)})^t\bigr\|_F^2
=
\Tr\!\left[((C^{(3)})^t)^\dagger (C^{(3)})^t\right].
\end{equation}
For the parity-preserving dopant $A$, \(C^{(3)}\) is real and symmetric, so this simplifies to
\begin{equation}
\mathcal F^{(3)}(t)
=
\Tr\!\left[(C^{(3)})^{2t}\right]
=
\sum_\alpha \lambda_\alpha^{2t},
\end{equation}
where \(\lambda_\alpha\) are the eigenvalues of \(C^{(3)}\).

For moderate system sizes, we evaluate this expression directly by constructing the transfer matrix and computing either its Frobenius norm or full eigenspectrum. However, the dimension of the three-copy commutant grows as $D_3=\binom{2n+3}{3}$, making explicit matrix construction impractical at large $n$.

For large system sizes, direct evaluation of the trace becomes impractical due to the growth of the three-copy commutant dimension. We therefore estimate the frame potential stochastically using
\begin{equation}
\Tr M = \mathbb E_r \left[r^\top M r\right],
\end{equation}
valid for random probe vectors satisfying \(\mathbb E[r_i r_j]=\delta_{ij}\). Applying this identity to \(M=(C^{(3)})^{2t}\) gives
\begin{equation}
\mathcal F^{(3)}(t)
=
\mathbb E_r \left[
r^\top (C^{(3)})^{2t} r
\right].
\end{equation}
In practice, we use independent Rademacher probe vectors and apply \(C^{(3)}\) iteratively without explicitly forming the full matrix. Instead, the action of \(C^{(3)}\) is computed directly from the combinatorial transition rules derived in App.~\ref{app:3copy}. This matrix-free implementation substantially reduces memory requirements and enables simulations up to \(n \approx 250\), where the basis is of size $\sim 2\cdot 10^7$.

The same transfer-matrix framework also provides access to the state \(3\)-frame potential. For an ensemble of pure states, this quantity is defined as
$
\mathcal S^{(3)}
=
\mathbb E_{\psi,\phi}
\left[
|\langle \psi | \phi \rangle|^6
\right].
$
Writing \(\rho_\psi = |\psi\rangle\langle\psi|\), the overlap can be expressed as
$
|\langle \psi | \phi \rangle|^6
=
\Tr\!\left[
(\rho_\psi^{\otimes 3})
(\rho_\phi^{\otimes 3})
\right].
$

Averaging independently over the ensemble therefore gives
$
\mathcal S^{(3)}
=
\Tr\!\left[(M_3)^2\right],
$
where
$
M_3
=
\mathbb E_{\psi}
\left[
\rho_\psi^{\otimes 3}
\right]
$
is the three-copy moment operator.

For the globally doped protocol, the relevant ensemble consists of states generated by evolving the vacuum under alternating matchgate layers and dopants. Starting from the fermionic Gaussian ensemble,
$
M_3^{(0)}
=
\mathbb E_{U\in \MG_n}
\left[
(U|0\rangle\langle0|U^\dagger)^{\otimes 3}
\right],
$
the effect of doping is incorporated through repeated application of the three-copy transfer matrix \(C^{(3)}\).

To evaluate this numerically, we expand the moment operator in the orthonormal pairing basis introduced in App.~\ref{app:3copy},
$
M_3^{(0)}
=
\sum_{\boldsymbol k}
a^{(0)}_{\boldsymbol k}
\big|\Upsilon^{(3)}_{\boldsymbol k}\big\rrangle .
$
The coefficients \(a^{(0)}_{\boldsymbol k}\) specify the projection of the initial moment operator onto each basis element of the three-copy commutant. Equivalently, they define a coefficient vector \(a^{(0)}\) representing \(M_3^{(0)}\) in the pairing basis. For the fermionic Gaussian vacuum ensemble, these coefficients are known analytically from the three-copy matchgate twirl \cite{wan2023matchgate}. Writing \(\boldsymbol{k}=(k_1,k_2,k_3)\), only even values contribute, reflecting the even fermionic parity of the vacuum state. Defining \(k_i = 2\ell_i\) and
$
\ell_4 = n-\ell_1-\ell_2-\ell_3,
$
the coefficients take the form
\begin{equation}
a^{(0)}_{2\ell_1,2\ell_2,2\ell_3}
=
\frac{(-1)^{\ell_1+\ell_2+\ell_3}}{2^{3n/2}}
\frac{
\binom{n}{\ell_1,\ell_2,\ell_3,\ell_4}
}{
\sqrt{
\binom{2n}{2\ell_1,2\ell_2,2\ell_3,2\ell_4}
}
},
\end{equation}
while
$
a^{(0)}_{k_1,k_2,k_3}=0
\
\text{if any } k_i \text{ is odd}.
$

Under \(t\) doping layers, the coefficients evolve linearly according to
$
a^{(t)}=(C^{(3)})^t a^{(0)}.
$
This evolution captures the dynamics of the moment operator entirely within the commutant basis, avoiding explicit construction of the full operator in Hilbert space.

Because the pairing basis is orthonormal, the Hilbert--Schmidt norm reduces to the Euclidean norm of the coefficient vector. Writing
$
a^{(t)} = (C^{(3)})^t a^{(0)},
$
the state \(3\)-frame potential becomes
$
\mathcal S^{(3)}(t)
=
\Tr[(M_3^{(t)})^2]
=
\bigl(a^{(t)}\bigr)^\top a^{(t)}
=
\left\|(C^{(3)})^t a^{(0)}\right\|_2^2 .
$
Equivalently, in the orthonormal pairing basis this may be written as
$
\mathcal S^{(3)}(t)
=
\sum_{\boldsymbol k}
\left|a^{(t)}_{\boldsymbol k}\right|^2,
$
which makes explicit that the frame potential is obtained from the squared norm of the evolved coefficient vector.  Numerically, our matrix-free implementation enables efficient evaluation of these matrix--vector multiplications without ever explicitly constructing the doping matrix.

At zero doping, this expression yields the fermionic Gaussian state \(3\)-frame potential,
$
\mathcal S^{(3)}_{\mathrm{FGS}}
=
\frac{2}{2^n C_n},
$
where
$
C_n=\frac{1}{n+1}\binom{2n}{n}
$
is the \(n\)-th Catalan number~\cite{sierant2026theory}. This closed-form result provides a nontrivial consistency check on both the normalization of the three-copy basis and the initialization of \(a^{(0)}\).

\section{Probabilistic Doping}
\label{app:Probabilistic_Doping}

The probabilistic protocol provides an intermediate setting between the globally doped circuits of Eq.~\eqref{eq:circuit} and the localized brickwork construction discussed in the main text. Rather than inserting the non-Gaussian gate deterministically at every layer or at a fixed spatial location, we now allow non-Gaussian gates to appear stochastically throughout the circuit.

Concretely, we consider a brickwork matchgate circuit in which, at every layer and for every allowed bond, the Gaussian gate is independently replaced by the fixed non-Gaussian dopant $A$ with probability $p$. Throughout this appendix, we focus on the dilute regime $p = \frac{K}{N}$, with $K=\mathcal O(1)$. In this scaling, each circuit layer contains on average a constant number of non-Gaussian gates, so that the total number of dopants grows linearly with circuit depth.

Figure~\ref{fig:Prob_dope} shows the resulting unitary $2$-frame potential for several system sizes. In contrast to the localized protocol, whose convergence required the diffusive scaling variable $t/n^2$, the probabilistic data exhibit a substantially faster relaxation. At the same time, the collapse is not as clean as in the fully global protocol when plotted as a function of $t/n$, particularly for the smaller system sizes. 

\begin{figure}[h!]
\centering
\includegraphics[width=\linewidth]{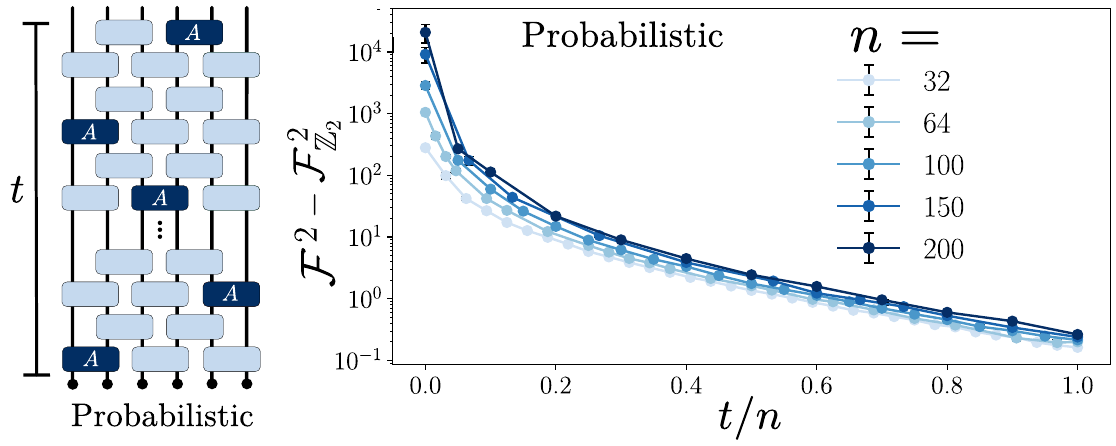}
\caption{Probabilistic doping protocol and resulting unitary $2$-frame potential. Left: schematic illustration of the circuit. Right: $\mathcal F^2-\mathcal F^2_{\mathbb Z_2}$ as a function of the rescaled depth $t/n$ for several system sizes $n$. The probabilistic protocol exhibits faster relaxation than the localized construction, while displaying a weaker scaling collapse than the globally doped setting.}
\label{fig:Prob_dope}
\end{figure}

In the localized setting, non-Gaussianity is injected only at a single spatial point and must subsequently propagate through the system under the surrounding Gaussian dynamics, leading to the transport-limited $\mathcal O(n^2)$ timescale observed in the main text. In the probabilistic protocol, by contrast, non-Gaussian resources are distributed throughout the circuit volume, allowing multiple regions of the system to become non-Gaussian simultaneously. The resulting dynamics therefore avoid the bottleneck associated with transport from a single injection point, while remaining substantially sparser than the globally doped construction.

The probabilistic protocol is therefore intermediate between the two extremes. In the globally doped setting, the same local non-Gaussian gate is inserted between fully random matchgate layers, leading to a frame-potential relaxation occurring on an $\mathcal O(n)$ scale. In contrast, in the localized setting, non-Gaussianity is repeatedly injected only at a fixed bond, so that randomness must spread diffusively through the surrounding Gaussian dynamics, resulting in the much slower $\mathcal O(n^2)$ scaling observed in the main text.

More generally, our setup is intentionally minimal. We focus on dilute probabilities in order to remain in the weak-doping regime and to allow a direct comparison with the global and localized protocols. It is natural to expect that broader classes of parity-preserving non-Gaussian gates, alternative spatial distributions, or deeper random circuits could further improve convergence properties while preserving the overall linear scaling of the total non-Gaussian resources.

\end{document}